\documentclass[usenames,dvipsnames,12pt]{article}
\usepackage{hyperref}
\usepackage{amssymb, amsthm,amsmath}
\usepackage{stmaryrd,ulem}
\usepackage{float}
\usepackage{amsmath}
\usepackage{graphicx,psfrag,epsf}
\usepackage{enumerate}
\usepackage{natbib}
\usepackage{url} 
\usepackage{tikz}
\usepackage{pgfplots}
\usepackage{subcaption}
\usepackage{a4wide}
\usepackage[linecolor=blue!60!,backgroundcolor=blue!10!,textwidth=2.5cm,textsize=scriptsize]{todonotes}
\usepackage{comment}
\usepackage{enumitem}

\input{macro}

\usepackage[utf8]{inputenc}
\usetikzlibrary{patterns}
\usetikzlibrary{shapes}
\usepackage[linesnumbered,ruled]{algorithm2e}

\usepackage{bm}
\usepackage{dsfont}
\usepackage{amssymb}
\usepackage{amsfonts}
\usepackage{amsthm}
\usepackage{mathtools}
\usepackage{xcolor}
\usepackage{multirow}

\newtheorem{theorem}{Theorem}[section]
\newtheorem{lemma}[theorem]{Lemma}
\newtheorem{proposition}[theorem]{Proposition}
\newtheorem{corollary}[theorem]{Corollary}
\newtheorem{remark}{Remark}[section]

\newtheorem{assumption}{Assumption}

\newtheorem{definition}{Definition}

\newcommand{\IMS}{{\mbox{{\tiny IMS}}}}
\newcommand{\MS}{{\mbox{{\tiny MS}}}}

\newcommand{\et}[1]{\textcolor{black}{#1}}

\newcommand{\ind}[1]{{\mathbf{1}\{#1\}}}
\newcommand{\range}[1]{[#1]}
\newcommand{\R}{\mathbb{R}}
\newcommand{\E}{\ensuremath{\mathbb E}}

\renewcommand{\P}{\ensuremath{\mathbb P}}
\newcommand{\cH}{{\mathcal{H}}}

\newcommand{\FDR}{\mathrm{FDR}}

\newcommand{\ABH}{\mathrm{ABH}}
\newcommand{\Pow}{\mathrm{Pow}}

\usepackage{booktabs}

\SetKwComment{Comment}{/* }{ */}

\title{
On min-Storey estimators for multiple testing and conformal novelty detection
}
\date{}

\author{Gao, Zijun and Roquain, Etienne}

\begin{document}

\maketitle

\begin{abstract}
In a multiple testing task, finding an appropriate estimator of the proportion $\pi_0$ of non-signal in the data to boost power of false discovery rate (FDR) controlling procedures is a long-standing research theme, sometimes referred to as 'adaptive FDR control'. The interest in this theme has been reinforced in the recent years with conformal novelty detection, for which it turns out that similar tools can be used in combination with any 'blackbox' machine learning algorithm. Nevertheless, perhaps surprisingly, finding a solution for 'adaptive FDR control' that is optimal in a broad sense is still an open problem. This paper fills this gap by introducing new $\pi_0$-estimators, referred to as min-Storey (MS) and interval-min-Storey (IMS), which are built upon the so-called 'Storey estimator'. Plugging these estimators in the adaptive Benjamini-Hochberg (BH) procedure is shown to deliver FDR control both in the independent and conformal settings. In addition, these methods satisfy an optimal power property over any (regular) alternative distribution.  
The excellent behaviors of the new adaptive procedures are illustrated with numerical experiments  both in the independent and conformal models 
for various distribution structures. 
\end{abstract}


\section{Introduction} \label{sec:intro}

\subsection{Background}

Since the introduction of the pioneer paper \cite{BH1995}, false discovery rate (FDR) is one prominent research theme of statistics influencing a wide range of applied and theoretical developments\footnote{As a synthetic indication of its influence, this body of work received the Rousseeuw Prize for Statistics in 2024.}.
The central result is that, prescribing a nominal level $\alpha$ and denoting by $\pi_0$ the proportion of non-signal in the data, the Benjamini-Hochberg (BH) procedure allows to control the FDR at level $\pi_0 \alpha$ under independence of the $p$-values, or more broadly when these $p$-values are positively regressively dependent on each one of a subset (PRDS) \citep{BY2001}. Recently, an important connection with machine learning -- and more precisely to conformal novelty detection -- has been made by showing that the conformal $p$-values do meet the PRDS condition, and thus can be used into the BH procedure to control the FDR at level $\pi_0 \alpha$ \cite{bates2023testing,marandon2024adaptive}. 

A well-known issue is that, since $\pi_0 \alpha$ can be less than $\alpha$, applying BH at level $\alpha$ induces some loss in the actually achieved FDR. Obviously, one easy solution is to apply BH at level $\alpha/\pi_0$ in the first place to end up with a control at level $\alpha$. However, the latter is not possible in general as $\pi_0$ is generally unknown, and the latter method is classically referred to as {\it  oracle} (adaptive) BH procedure. A well-established line of research is to design a $\pi_0$-estimator $\hat{\pi}_0$ so that the BH procedure applied at level $\alpha/\hat{\pi}_0$ still ensures an FDR at most $\alpha$, see for instance \cite{BH2000,benjamini2006adaptive,SC2007,Sar2008,BR2009,FDR2009,liang2012adaptive,heesen2015inequalities,li2019multiple,guo2020adaptive,dohler2023unifiedclassnullproportion}  among others and \cite{bates2023testing,marandon2024adaptive,gao2025adaptive} for the specific conformal setting.  

Despite the rich literature, our paper studies the question of how to construct optimal $\pi_0$ estimators while preserving finite-sample FDR control, which remains less well understood.
Multiple challenges arise therein. 
First, an optimal $\pi_0$ estimator should adapt to the empirical $p$-value distribution so as to remain optimal over a broad class of null and non-null structures, whereas existing optimality (and even consistency) results typically require the $U(0,1)$ null and a decreasing non-null density that may not hold both in one-sided and two-sided testing, 
see \cite{Neu2013} (e.g., Lemma~2 therein) and the related works \cite{GW2002,GW2004,Chi2007}. More formally, it has been shown therein that the $\lambda$-Storey BH procedure converges to the oracle BH procedure by taking $\lambda$ converging to $1$ at some well chosen rate when the alternative $p$-value density is decreasing with a limit equal to $0$ at $1$. 
Second, a proportion of existing estimators rely on a capping argument \citep{li2019multiple,gao2025adaptive}, in which $\hat{\pi}_0$ is based only on $p$-values above some $\kappa$ and the final BH threshold is capped at $\kappa$, potentially leading to power sub-optimality; by contrast, \citep{BR2009,bates2023testing,marandon2024adaptive} consider monotonicity-based approaches in which $\hat{\pi}_0$ is non-decreasing in each individual $p$-value, and no capping is required.
Third, we seek to control the finite-sample FDR, in contrast to the null proportion estimators in \citep{SC2007,heller2021optimal,SC2009,tony2019covariate,rebafka2022powerful,jin2008proportion,cai2010optimal,carpentier2021estimating}, which establish asymptotic FDR or mFDR control. 

\subsection{Adaptation for super-uniform null $p$-values}

In the most general setting where the null $p$-values are only assumed to be super-uniform, the so-called `Storey estimator' \citep{Storey2002} is given by
\begin{equation}
    \label{equStorey}
    \hat{\pi}_{0,\lambda} := \frac{1+\sum_{i=1}^m \ind{p_i> \lambda}}{m(1-\lambda)},
\end{equation}
where $\lambda\in (0,1)$ is a free parameter. The '$1+$' above is used to ensure that BH at level $\alpha/\hat{\pi}_{0,\lambda}$, referred to here as $\lambda$-Storey procedure for short, controls the FDR at level $\alpha$ for any $m\geq 1$ and $\lambda\in (0,1)$ under independence  \citep{benjamini2006adaptive,BR2009}. In the conformal case, for which the $p$-values are discrete, the latter adaptive procedure still controls the FDR at level $\alpha$, but the definition of the estimator $\hat{\pi}_{0,\lambda} $ should be made slightly more conservative, with $\ind{p_i\geq \lambda}$ instead of $\ind{p_i> \lambda}$ in \eqref{equStorey} with $(n+1)\lambda\in [n+1]$ \cite{bates2023testing,marandon2024adaptive}. 

Choosing $\lambda$ when using $\lambda$-Storey procedure is a puzzling question, and no real consensus has been made through the years: in simulations, no choice dominates others and $\lambda=1/2$ is often presented as a good default value, while $\lambda=\alpha$ is an alternative more robust to dependence \cite{BR2009}. Nevertheless, in an independent setting where the alternative $p$-value density is decreasing, taking $\lambda$ approaching $1$ when $m$ tends to infinity appears to be a more relevant choice to balance  the bias and the variance of $\hat{\pi}_{0,\lambda}$, with a threshold converging to an oracle at a rate depending on the regularity of the smoothness of alternative density near $1$ \citep{neuvial2013asymptotic}.  

An alternative strategy is to make a dynamic choice of $\lambda$, that is to choose $\lambda=\hat{\lambda}$ driven by the data themselves. While this makes the mathematical analysis more challenging, martingale theory can be used in combination with the $\kappa$-capping trick (see (ii) above) to prove that the right-boundary estimator (essentially introduced in \cite{BH2000}) and defined by
\begin{equation}
    \label{equ:rightboundary}
    \hat{\lambda} = \min\{\lambda_j \::\:  j\in J,\: \hat{\pi}_{0,\lambda_{j}}\leq \hat{\pi}_{0,\lambda_{j+1}}\},
\end{equation}
(where $(\lambda_j,j\in J)\subset (\kappa,1)$ is some pre-specified grid) can be used into a $\hat{\lambda}$-Storey procedure to control the FDR at level $\alpha$ both in the independent and conformal settings \cite{gao2025adaptive}.  In the independence case, it is proved therein that it asymptotically mimics the best $\lambda$-Storey procedure (for $\lambda$ in the grid) under some distributional assumptions (mixture density is strictly convex), and provided that the asymptotic threshold is below the capping $\kappa$.

\paragraph{Contribution 1: Min-Storey (MS) procedure}

Our first contribution is to improve this $\hat{\lambda}$-Storey procedure, by considering the smallest Storey estimator $\hat{\pi}^{\MS}_{0}= c\cdot \min_\lambda \hat{\pi}_{0,\lambda}$ where $c\geq 1$ is a correction factor. 
Since $\hat{\pi}^{\MS}_{0,\lambda}$ is by definition nondecreasing in each $p$-value, we employ the 
monotonicity-based approach
above to show that $c=c(m)$ can be chosen in a way that the corresponding adaptive BH procedure, called the {\it min-Storey (MS) procedure}, has an FDR below $\alpha$. Interestingly, we prove that the factor $c$ is close to $1$ when $m$ gets large. 
As a consequence, we are able to show that the MS procedure has a power close to the best adaptive BH procedure when the null distribution is super-uniform, without relying on any assumption on the alternative distribution (Corollary~\ref{coroptimalMS}): from an intuitive point of view, $\lambda$-Storey estimators are suitable in this case because $\hat{\pi}_{0,\lambda}$ are always biased upwards under null super-uniformity, which guards against having a minimum too low and thus a compensating factor $c$ too large. By contrast, the right boundary is suboptimal when the  $p$-value density has two local minima, and stopping at different minima leads to different estimates, see Remark~\ref{rem:RBsuboptimal} and settings (a.2), (a.3), (b.2) and (b.3) in our numerical experiments (Section~\ref{sec:simulation}), which clearly shows the superiority of MS procedure.
In particular, our numerical experiments in Section~\ref{sec:simulation} show that MS behaves approximately as the best  Storey-procedures, across varying $\lambda$ values.

\subsection{Adaptation for uniform null $p$-values and conformal $p$-values}

Second, we deal with the case where the $p$-values are either exactly uniformly distributed under the null, or are conformal $p$-values. The alternative  distribution is still left arbitrary.
Because of possible large alternative $p$-values, the $\lambda$-Storey estimator \eqref{equStorey} is too conservative in general (and thus so are the MS and right boundary approaches). As considered in \citep{celisse2010cross,neuvial2013asymptotic,heesen2016dynamic}
 an improved version is 
\begin{equation}
    \label{equStoreytwosided}
    \hat{\pi}_{0,[\lambda,\mu]} := \frac{1+\sum_{i=1}^m \ind{p_i \in[\lambda,\mu]}}{m(\mu-\lambda)},
\end{equation}
where $0<\lambda\leq \mu\leq 1$ are free parameters.
Compared to $\hat{\pi}_{0,\lambda}$, the estimator $\hat{\pi}_{0,[\lambda,\mu]}$ is a two sided version which has the ability to focus on an interval $[\lambda,\mu]$ for which the density of the alternative $p$-values is low. In this case, since the distribution of the null $p$-values is (approximately) uniform, we have 
$$
\hat{\pi}_{0,[\lambda,\mu]}  \approx \frac{\sum_{i\in \cH_0} \ind{p_i \in[ \lambda,\mu]}}{m(\mu-\lambda)} \approx \pi_0,
$$
provided that alternative $p$-values have a small probability to belong to $[\lambda,\mu]$.
However, when $\mu<1$, it is apparent that $\hat{\pi}_{0,[\lambda,\mu]}$ is not monotone with regard to each $p$-value; this is in accordance with the fact that larger $p$-values might indicate fewer null in the possible situation where there are large alternative $p$-values. By following the 
{monotonicity-based approach}
above, using BH at level $\alpha/\hat{\pi}_{0,[\lambda,\mu]}$ and capping the threshold at $\lambda$, a procedure that one can refer to as $(\lambda,\mu)$-Storey procedure, does control the FDR both under independence \cite{heesen2016dynamic}. Provided that $(n+1)\lambda\in [n+1]$ and $(n+1)\mu \in [n+1]$, it is also true in the conformal setting (see Corollary~\ref{corcontrolIStorey} below), which is a novel result to our knowledge.
Again, the choice of the parameters is important in the $(\lambda,\mu)$-Storey procedure, as it should target an interval $[\lambda,\mu]$ where there are only few alternative $p$-values. 

\paragraph{Contribution 2: Interval-min-Storey (IMS) procedure}

Our second contribution is to achieve the choice of $[\lambda,\mu]$ in a data-driven way, by considering the smallest $(\lambda,\mu)$-Storey estimator $\hat{\pi}^{\IMS}_{0}:= c\cdot \min_{I\subset [\hat{\kappa},1]:|I|\geq \epsilon} \hat{\pi}_{0,I}$ where $c\geq 1$ is again a correction factor, $\epsilon>0$ should be chosen carefully and $\hat{\kappa}$ is a capping, that we are able to set in a data-driven way. 
We employ the strategy (i) above to show that $c=c(m,\epsilon)$ can be chosen in a way that the corresponding adaptive BH procedure (with a $\hat{\kappa}$-capping), called the {\it interval min-Storey  (IMS) procedure}, has an FDR below $\alpha$. 
We show that the IMS procedure has a power close to the best adaptive BH procedure both in the settings where the null distribution is exactly uniform and in the conformal setting, without relying on any assumption on the alternative distribution (Corollary~\ref{corgenunif}) and for $\alpha$ small enough. 
Obviously, any procedure which is suitable for super-uniform null $p$-value distribution can also be used in this context, and our power result shows that they are all dominated by the IMS procedure. Figure~\ref{fig:intro} below reports an example illustrating the superiority of IMS procedure.
In particular, our numerical experiments in Section~\ref{sec:simulation} show that IMS can outperform MS for a (exactly) uniform null distribution. 

\begin{figure}[t]
\centering
\begin{minipage}{0.35\textwidth}
  \centering
  \includegraphics[clip, trim = 0cm 0cm 19cm 0cm, height = 1\textwidth]{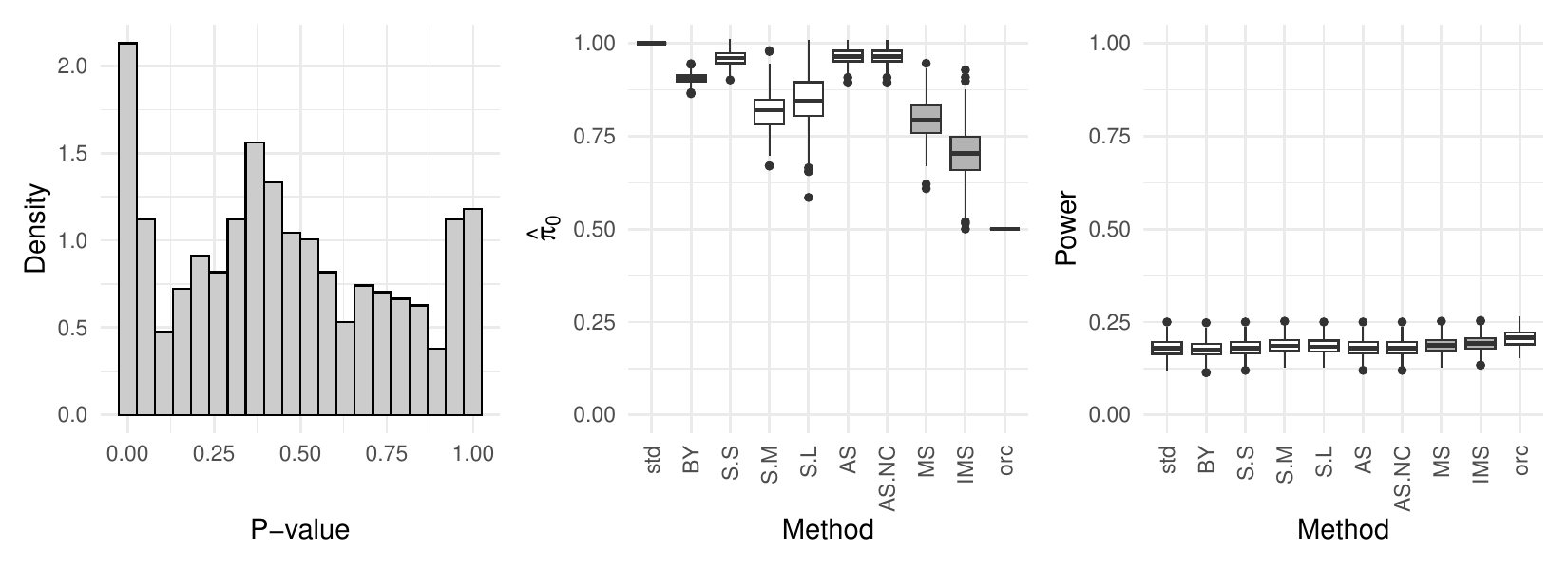}
  \label{fig:myfigure}
\end{minipage}
\begin{minipage}{0.5\textwidth}
  \centering
  \label{tab:mytable}
  \vspace{0.2em}
\begin{tabular}{ccc}
\toprule
Method & $\hat{\pi}_0$ & \# Rejections \\
\midrule
std   & 1.00 (0.0000)  & 99.7 (0.66) \\
BY    & 0.91 (0.0006) & 97.3 (0.66) \\
S.S   & 0.85 (0.0036)  & 105.6 (0.73) \\
S.M   & 0.82 (0.0026)  & 106.8 (0.74) \\
S.L   & 0.96 (0.0010) & 101.2 (0.69) \\
AS    & 0.96 (0.0011)  & 101.0 (0.69) \\ 
\textbf{MS}    & \textbf{0.79 (0.0027)}  & \textbf{107.6 (0.75)} \\
\textbf{IMS}   & \textbf{0.71 (0.0032)}  & \textbf{112.6 (0.81)} \\
\bottomrule
\end{tabular}
\end{minipage}
\caption{
Comparison of null proportion estimators for conformal $p$-values.
The left panel shows the empirical distribution of the conformal $p$-values.
The test and calibration samples both have size $1000$.
Null and calibration scores are drawn independently from $U(0,1)$; for the alternatives, we generate independent scores $S_i = \Phi(X_i)$ with $X_i \sim 0.5\mathcal N(1,9)+0.5\mathcal N(0.2,0.09)$.
The choice of the alternative mixture produces a mid-range bump (weak signals from the component $\mathcal N(0.2,0.09)$) in the $p$-value histogram and an excess of large $p$-values (the high-variance component $\mathcal N(1,9)$ yields significantly negative $X_i$ and thus considerably small $S_i$).
The true null proportion $\pi_0 = 0.5$.
The right panel shows various null proportion estimators $\hat\pi_0$ and the resulting number of rejections $R$ from the associated adaptive BH 
($400$ repeats; standard errors in parentheses). 
Our proposal MS and IMS produce the smallest $\pi_0$-estimation and the largest rejection number in average.
\label{fig:intro}}

\end{figure}

\subsection{Comparison with previous optimality results in the literature}

In the multiple testing literature, several optimality results have been established. Notably, the seminal work of \cite{SC2007} opened a line of research that identified the local false discovery rate (lfdr) as the optimal test statistics for marginal FDR (mFDR) control. Later, the picture has been completed in \cite{heller2021optimal} with the exact optimal solution for FDR control (and not mFDR control). Many following-up studies have been considered: for instance \cite{SC2009,tony2019covariate,rebafka2022powerful} among others.
These approaches work by first identifying an oracle procedure, relying on the known parameters of the underlying $p$-value distribution, and then by estimating these parameters and by using a plug-in procedure. The asymptotic optimality is proved by assuming that the parameter estimators are consistent. In particular, these results exclude regions of the parameter set where the estimation is too difficult. In addition, it does not provide a finite sample control of the FDR for the plug-in procedure.   

Our work is markedly different, because the test statistics are assumed to be fixed once for all (possibly by independent previous experiments) and the oracle procedure is the BH procedure taken at level $\alpha/\pi_0$ with these scores, or more generally any BH procedure taken at level $\alpha/\hat{\pi}_0$ with an estimator $\hat{\pi}_0$ that concentrates around some value while providing a finite sample FDR control, see Section~\ref{sec:estimclass} for the exact definition of the estimator class. 
In other words, our work deliberately ignore the optimization with regard to the test statistics (or score) and only focus on the $\pi_0$-estimation optimization. This makes the problem at stake here markedly different from the line of research described above.   

We show that this $\hat{\pi}_0$-optimization setting allows to draw a  clear picture of what is the optimal solution in the two cases (super-uniform or uniform/conformal) described in the previous section, without making any (strong) assumption on the alternative distribution of the test statistics (or scores). This is crucial for practice, for which the data or potentially complex and the alternatives are mixture of `good' and `bad' cases. Especially in the conformal case, this might come from a score function which is only partially efficiently designed, which is typically when observations are in high dimension.  
In particular, our power optimality results also hold in situations where the oracle BH makes only few rejections and can be applied indifferently in `difficult' or `easy' configurations. They are also uniform with regard to the level $\alpha$ ($\alpha$ can depends on the sample size(s)). Our power result formulation is inspired from the study in \cite{roquain2022false} which was made for another purpose (unspecified null distribution under sparsity).

Finally, we would like to mention other, perhaps less related, studies: some are optimizing the power among $p$-value weighted BH procedures  
\citep{RW2009,durand2019adaptive,ignatiadis2021covariate}, other are focusing on obtaining optimal $\pi_0$ estimator, without providing a finite sample FDR control \citep{jin2008proportion,cai2010optimal,carpentier2021estimating}. 

\section{Preliminaries}

\subsection{Settings}\label{sec:settings}

We consider a multiple testing setting for which $m$ null hypotheses are tested with $p$-values $p_1$, $\dots$, $p_m$ and $\cH_0 \subset [m]$  are the indices corresponding to the true nulls/non novelties/inliers. We denote $\cH_1=[m]\backslash \cH_0$ the indices of the false nulls/novelties/outliers and $\pi_0$ (resp. $\pi_1$) the total proportion of nulls (resp. alternatives).

We consider two  specific settings:
\begin{itemize}
    \item {\it Independent multiple testing}: $p_1$, $\dots$, $p_m$ are mutually independent, and for $i\in \cH_0$, the distribution of $p_i$ is marginally super-uniform under the null:
\begin{equation}
    \label{equ:superunif}
    \forall i\in \cH_0 ,\: \forall t\in [0,1], \:\P(p_i\leq t)\leq t.
\end{equation}
    The marginal distribution of $p_i$ when $i\in \cH_1$ is left arbitrary.
    \item {\it Conformal novelty detection}: each $p$-value $p_i$, $i\in [m]$,  is of the form
    \begin{equation}
        \label{equ-confpvalues}
        p_i = (n+1)^{-1} \Big(1+\sum_{j=1}^n \ind{S_j\geq S_{n+i}}\Big),
    \end{equation}
    where the non-conformity scores $S_1,\dots,S_{n+m}$ are random variables  such that 
    $(S_1,\dots,S_{n+m})$ have no ties almost surely and 
    $(S_j, j\in [n])\cup (S_{n+i},i\in \cH_0)$ are exchangeable conditionally on $(S_{n+i},i\in \cH_1)$.
\end{itemize}
In the conformal setting, note that the $p$-values are dependent because they are based on the same calibration sample $(S_1,\dots,S_{n})$. The dependence is however specific, and turns out to be PRDS \cite{bates2023testing,marandon2024adaptive}. 

\subsection{Adaptive BH procedures and FDR}

An estimator $\hat{\pi}_0$ of $\pi_0$ is a measurable function of the $p$-values  $p_1$, $\dots$, $p_m$, which is valued in $[0,+\infty)$. For such an estimator, the {\it adaptive BH procedure} at level $\alpha\in (0,1)$, using a capping $\kappa\in (0,1]$, and using the estimator  $\hat{\pi}_0$ is denoted $\ABH_\alpha(\hat{\pi}_0,\kappa)$ and is defined as rejecting the nulls corresponding to the indices $\{i\in [m]\::\: p_i\leq \hat{t}\}$ with the threshold $\hat{t}=\kappa \wedge (\alpha \hat{k}/ (m \hat{\pi}_0))$ (convention $1/0=+\infty$) and the number of rejections is 
\begin{equation}
    \label{equ-lhat}
    \hat{k}=\max\{k\in [0,m]\::\: p_{(k)}\leq \kappa \wedge (\alpha k/ (m \hat{\pi}_0))\},
\end{equation}
where $0=:p_{(0)}\leq p_{(1)}\leq \dots\leq p_{(m)}$ are the ordered $p$-values. When $\kappa=1$, we consider that the adaptive BH procedure does not use any capping and denote $\ABH_\alpha(\hat{\pi}_0)$ for $\ABH_\alpha(\hat{\pi}_0,1)$. Also the adaptive BH procedure $\ABH_\alpha({\pi}_0,\kappa)$ that uses the true value ${\pi}_0$ is called the  oracle adaptive BH procedure. 

The false discovery rate (FDR) of $\ABH_\alpha(\hat{\pi}_0,\kappa)$ is defined as
\begin{equation}
    \label{equ:FDR}
\FDR(\ABH_\alpha(\hat{\pi}_0,\kappa)):=\E\Bigg[\frac{\sum_{i\in \cH_0}\ind{p_i\leq \hat{t}}}{\hat{k}\vee 1}\Bigg],
\end{equation}
where $\hat{t}$ and $\hat{k}$ are defined above.
The power of $\ABH_\alpha(\hat{\pi}_0,\kappa)$ is defined as 
\begin{equation}
    \label{equ:pow}
\Pow(\ABH_\alpha(\hat{\pi}_0,\kappa)):=\E\Bigg[\frac{\sum_{i\in \cH_1}\ind{p_i\leq \hat{t}}}{m}\Bigg].
\end{equation}

\section{Extending known results}

In this section we provide a general FDR bound for adaptive BH procedure, which extends the existing literature by providing statements unifying independent/conformal settings and monotonic/capping assumptions.

\subsection{General FDR bound for adaptive BH procedures}

FDR bounds can be obtained under those conditions by considering $\tilde{\mbf{p}}^{0,i}$ a modification of the $p$-value family $\mbf{p}=(p_j)_{j\in [m]}$ defined as follows for some $i\in \cH_0$:
\begin{itemize}
    \item independence setting: $\tilde{\mbf{p}}^{0,i}$ denotes the $p$-value family $\mbf{p}$ in which $p_i$ is replaced by $0$.  
    \item conformal setting: $\tilde{\mbf{p}}^{0,i}$ denotes the $p$-value family $\mbf{p}$ in which $p_i$ is replaced by $0$ 
    and $p_j$ is replaced by
    \begin{equation}
    \label{equ:configconf}
        \tilde{p}^{0,i}_j =
      \frac{ \sum_{k\in [n]\cup\{n+i\}}  \ind{S_k \ge S_{n+j}}}{n+1}, j\in [m]\backslash\{ i\}.
\end{equation}
\end{itemize}

Let us consider the following assumptions.

\begin{assumption}\label{ass:monotone}
    \mbox{$\hat{\pi}_0(\mbf{p})$ is coordinate-wise nondecreasing in each $p$-value $p_i$.}
\end{assumption}
\begin{assumption}\label{ass:capping}
    \mbox{$\hat{\pi}_0(\mbf{p})$ does not depend on the value of $p_i$ when $p_i\leq \kappa$, which means that}\\
    \mbox{for any $\mbf{p},\mbf{p}'$, $\hat{\pi}_0(\mbf{p})=\hat{\pi}_0(\mbf{p}')$ if for any $i\in [m]$, either $p'_i\leq p_i$ if $p_i\leq \kappa$ or $p'_i=p_i$ if $p_i>\kappa$.}
\end{assumption}

The following general FDR bound holds.

\begin{theorem}\label{th-gen}
Consider a $p$-value family $\mbf{p}=(p_i)_{i\in [m]}$ following either the independent setting or the conformal setting, some estimator $\hat{\pi}_0=\hat{\pi}_0(\mbf{p})$ of $\pi_0$ and $\tilde{\mbf{p}}^{0,i}$ the $p$-value family modified as above when $i\in \cH_0$. Then the adaptive BH procedure satisfies 
\begin{equation}
    \label{equ:controlgen}
    \FDR(\ABH_\alpha(\hat{\pi}_0,\kappa))\leq (\alpha/m)\sum_{i\in \cH_0} \E[1/\hat{\pi}_0(\tilde{\mbf{p}}^{0,i})]
\end{equation}
in either of the two following cases:
\begin{itemize}
    \item[(i)] if $\hat{\pi}_0$ satisfies Assumption~\ref{ass:monotone} and $\kappa=1$ (no capping);
\item[(ii)] if $\hat{\pi}_0$ satisfies  Assumption~\ref{ass:capping} with a capping $\kappa\in (0,1)$.
\end{itemize}
\end{theorem}

Theorem~\ref{th-gen} is proved in Section~\ref{proof:th-gen}. The part (i) of Theorem~\ref{th-gen} can be considered as well known: the independence setting reduces to Theorem~11 of \cite{BR2009} and equation (5) in \cite{benjamini2006adaptive}, while the conformal setting is close to  Corollary~3.7 of \cite{marandon2024adaptive} (with a slight change in the $p$-value configuration). 
The capping part is close to arguments given in \cite{li2019multiple}, although dealing with the capping in the conformal setting is  new to our knowledge.

\subsection{Application to Storey-like estimators}

The bound \eqref{equ:controlgen} can be applied for the Storey estimator $\hat{\pi}_{0,\lambda}$ defined in \eqref{equStorey}. Since it satisfies  Assumption~\ref{ass:monotone}, no capping is needed.

\begin{corollary}\label{corcontrolStorey}
    Let $\lambda\in (0,1)$ and assume in addition $(n+1)\lambda \in [n+1]$ in the conformal setting.
    For the $\lambda$-Storey estimator $\hat{\pi}_{0,\lambda}$ defined by \eqref{equStorey}, the adaptive BH procedure controls the FDR at level $\alpha$, that is,
    $$
 \FDR(\ABH_\alpha(\hat{\pi}_{0,\lambda}))\leq \alpha,
    $$
    without any capping and in both the independent and conformal settings.
\end{corollary}

The result in Corollary~\ref{corcontrolStorey} is well known \cite{benjamini2006adaptive,bates2023testing,marandon2024adaptive}.
The short proof is given in Section~\ref{secproofcorcontrolStorey} for illustration and pedagogical purpose.

The general bound also allows to deal with the Interval-Storey estimator $\hat{\pi}_{0,[\lambda,\mu]}$ defined in \eqref{equStoreytwosided}. While it does not satisfy Assumption~\ref{ass:monotone}, a result can still be established with a capping, because it satisfies Assumption~\ref{ass:capping} with any capping $\kappa\leq \lambda$.

\begin{corollary}\label{corcontrolIStorey}
    Let $\lambda,\mu\in (0,1)$ with $\lambda\leq\mu$ and assume in addition that $p_i$ is marginally uniform when $i\in \cH_0$ in the independent setting and that $(n+1)\lambda, (n+1)\mu \in [n+1]$ in the conformal setting. Consider a capping threshold $\kappa\in (0,\lambda]$. 
    For the Interval-Storey estimator $\hat{\pi}_{0,[\lambda,\mu]}$ defined in \eqref{equStoreytwosided}, the adaptive BH procedure with capping $\kappa$ controls the FDR at level $\alpha$, that is,
    $$
 \FDR(\ABH_\alpha(\hat{\pi}_{0,[\lambda,\mu]},\kappa))\leq \alpha,
    $$
   in both the independent and conformal settings.
\end{corollary}

While the result in Corollary~\ref{corcontrolIStorey} is known  in the independent setting \cite{heesen2016dynamic}, it is new to our knowledge in the conformal case. A short proof for both cases is provided in Section~\ref{secproofcorcontrolIStorey}. 

\begin{remark}\label{rmk:conformal+1}
    Considering $p_i\geq \lambda$ and not $p_i>\lambda$ in the $\lambda$-Storey estimator \eqref{equStorey} is necessary to obtain
    that $\E[1/\hat{\pi}_0(\tilde{\mbf{p}}^{0,i})]\leq 1$ in the conformal case. For instance, consider $n = 1$ and $m = m_0 = 10$, $\lambda = 1/2$ (so that $(n+1)\lambda\in [n+1]$), then
\begin{align*}
    \mathbb{E}\left[\frac{m(1-\lambda)}{ 1+ \sum_{i=1}^m \ind{p_i> \lambda}}\right]
    = \mathbb{E}\left[\frac{5}{ 1+ \sum_{i=1}^{10} \ind{p_i=1}}\right]
    = \frac{1}{11}  \sum_{j=0}^{10}\frac{5}{1+j}
    \approx 1.37 > 1.
\end{align*}
\end{remark}

\section{New procedures with FDR control}

\subsection{Minimum-Storey adaptive procedure}

\begin{definition}\label{def:MS}
   The Minimum-Storey procedure (MS) of parameters $\epsilon,\underline{\pi}_0\in (0,1)$ and of level $\alpha\in (0,1)$ is the adaptive BH procedure $\ABH_\alpha(\hat{\pi}^{\MS}_{0,\epsilon,\underline{\pi}_0})$ (no capping) with the $\pi_0$-estimator given by 
   \begin{equation}
       \label{equ:pi0hatMS}
       \hat{\pi}^{\MS}_{0,\epsilon,\underline{\pi}_0} = \underline{\pi}_0 \vee \bigg(C(m,\epsilon,\underline{\pi}_0)  \inf_{\lambda \in [0,1-\epsilon]}
      \frac{1 \vee \sum_{i=1}^{m} \mathbf{1}\{p_i \geq \lambda\}}{m(1-\lambda)} \big)\bigg),
   \end{equation}
   where $C(m,\epsilon,\underline{\pi}_0)$ is a normalizing factor given in \eqref{equ:gMS} below. 
\end{definition}

The normalizing factor $C(m,\epsilon,\underline{\pi}_0)$ has a complex expression given in Section~\ref{sec:ghcomputing}. It can be  computed by using a Monte-Carlo method and hence can be tabulated once for all given $m,\epsilon,\underline{\pi}_0$. 
It is also proved to be close to $1$ when $m$ gets large, and we refer the reader to Section~\ref{sec:ghcomputing} for more details on these properties. Also, the infimum in \eqref{equ:pi0hatMS} is a minimum that can be easily computed as
\begin{align}
\inf_{\lambda \in [0,1-\epsilon]}
      \frac{1 \vee \sum_{i=1}^{m} \mathbf{1}\{p_i \geq \lambda\}}{m(1-\lambda)} &= 1 \wedge \min_{\lambda\in \{p_i, i\in [m]\}\cap (0,1-\epsilon)}
      \frac{1 \vee \sum_{i=1}^{m} \mathbf{1}\{p_i > \lambda\}}{m(1-\lambda)},\label{equ:simplifyinfMS}
\end{align}
so that the MS procedure is fully accessible from the data. Expression \eqref{equ:simplifyinfMS} also shows that the MS estimator can be seen alternatively as minimizing the so called quantile-estimator \cite{BKY2006,BR2009}. \et{perhaps add here a word on the exact computation in the independent case.}

An important property of $\hat{\pi}^{\MS}_{0,\epsilon,\underline{\pi}_0}$ is that it satisfies Assumption~\ref{ass:monotone}, so that Theorem~\ref{th-gen} (i) can be used to ensure FDR control.

\begin{theorem}\label{thcontrolMS}
   For any choice of parameters $\epsilon,\underline{\pi}_0\in (0,1)$, the MS procedure taken at level $\alpha\in (0,1)$ controls the FDR at level $\alpha$, that is,  $\FDR(\ABH_\alpha(\hat{\pi}^{\MS}_{0,\epsilon,\underline{\pi}_0}))\leq \alpha$ both in the independent and conformal settings.
\end{theorem}

The proof is provided in Section~\ref{sec:proofthcontrolMS}. 
In the MS procedure, the parameters $\epsilon$ and $\underline{\pi}_0$ are free and any choice provides finite sample FDR control. Hence, they can be chosen to make the power as large as possible. Our power analysis (Section~\ref{sec:optimality}) suggests to take $\epsilon$ realizing a bias variance trade-off as $\epsilon \propto m^{-1/4}$ (up to a log factor) and  $\underline{\pi}_0\leq \pi_0$. Thus taking $\underline{\pi}_0$ slowly tending to zero is an option, while $\underline{\pi}_0 = 0.5$ is generally enough in practice, unless there is a belief that the null proportion is very low (Section~\ref{sec:simulation}). 

In Section~\ref{sec:optimality} (Corollary~\ref{coroptimalMS}), we show that the MS procedure satisfies an optimality property over all (regular) distributions with a super-uniform null. 

\subsection{Interval-Minimum-Storey adaptive procedures}

When the null distribution is imposed to be (exactly) uniform, we consider the following procedure.

\begin{definition}\label{def:IMS}
   The Interval-Minimum-Storey procedure (IMS) of parameters $\epsilon,\underline{\pi}_0\in (0,1)$, with capping $\kappa\in [0,1-\epsilon]$ and of level $\alpha\in (0,1)$ is the adaptive BH procedure $\ABH_\alpha(\hat{\pi}^{\IMS}_{0,\epsilon,\underline{\pi}_0,\kappa},\kappa)$ with capping $\kappa$ and the $\pi_0$-estimator given by 
   \begin{equation}
       \label{equ:pi0hatIMS}
       \hat{\pi}^{\IMS}_{0,\epsilon,\underline{\pi}_0,\kappa} = \underline{\pi}_0 \vee \bigg(D(m,\epsilon,\underline{\pi}_0)  \inf_{I\subset [\kappa,1]: |I|\geq \epsilon}
      \frac{1 \vee \sum_{i=1}^{m} \mathbf{1}\{p_i \in I\}}{m|I|} \bigg),
   \end{equation}
   for which the infimum is taken over all open  intervals $I$ contained in $[\kappa,1]$ with size at least $\epsilon$ 
   and 
   where $D(m,\epsilon,\underline{\pi}_0)$ is given in \eqref{equ:gIMS} below.
\end{definition}
Again, the normalizing factor $D(m,\epsilon,\underline{\pi}_0)$ can be  computed by using a Monte-Carlo method and is close to $1$ when $m$ gets large, see Section~\ref{sec:ghcomputing}. The infimum in \eqref{equ:pi0hatIMS} is a minimum:
\begin{align}
      \inf_{I\subset [\kappa,1]: |I|\geq \epsilon}
      \frac{1 \vee \sum_{i=1}^{m} \mathbf{1}\{p_i \in I\}}{m|I|}  &= \min_{
      \substack{I=(a,b), \kappa\leq a\leq b-\epsilon\\ a,b\in \{p_i, i\in [m]\}\cup\{\kappa,1\}}
      }
      \frac{1 \vee \sum_{i=1}^{m} \mathbf{1}\{p_i \in I\}}{m|I|} ,\label{equ:simplifyinfIMS}
   \end{align}
so that the IMS procedure is fully computable from the data.

The estimator $\hat{\pi}^{\IMS}_{0,\epsilon,\underline{\pi}_0,\kappa}$ is not monotone anymore, but satisfies Assumption~\ref{ass:capping}. Hence, we can use Theorem~\ref{th-gen} (ii) and  $\kappa$-capping to ensure the FDR control. 

\begin{theorem}\label{thcontrolIMS}
   For any choice of parameters $\epsilon,\underline{\pi}_0\in (0,1)$, the IMS procedure with capping $\kappa\in [0,1-\epsilon]$ taken at level $\alpha\in (0,1)$ controls the FDR at level $\alpha$, that is,  $\FDR(\ABH_\alpha(\hat{\pi}^{\IMS}_{0,\epsilon,\underline{\pi}_0,\kappa},\kappa))\leq \alpha$ both in the conformal setting and in the independent setting for which the null $p$-values are additionally assumed marginally uniform.
\end{theorem}

The proof is provided in Section~\ref{sec:proofthcontrolIMS}. 

In $\ABH_\alpha(\hat{\pi}^{\IMS}_{0,\epsilon,\underline{\pi}_0,\kappa},\kappa)$, there is a trade-off when choosing $\kappa$ in $\ABH_\alpha(\hat{\pi}^{\IMS}_{0,\epsilon,\underline{\pi}_0,\kappa},\kappa)$: the capping $\kappa$ should be small enough to ensure a good $\pi_0$-estimation, while $\kappa$ should be large enough to not deteriorate too much the threshold. We propose below an automatic (data-driven) choice of $\kappa$.

\begin{definition}\label{def:IMSkappahat}
   The Interval-Minimum-Storey procedure with data-driven $\hat{\kappa}$ (IMS-$\hat{\kappa}$) of parameters $\epsilon,\underline{\pi}_0\in (0,1)$ and of level $\alpha\in (0,1)$ is the adaptive BH procedure $\ABH_\alpha(\hat{\pi}^{\IMS}_{0,\epsilon,\underline{\pi}_0,\hat{\kappa}},\hat{\kappa})$ with capping $\hat{\kappa}$
   given by 
   \begin{equation}
    \label{kappahat}
\hat{\kappa} := \sup\{\kappa \in [0,1-\epsilon]\::\: \hat{F}_m(\kappa)\geq \kappa  \hat{\pi}^{\IMS}_{0,\epsilon,\underline{\pi}_0,\kappa}/\alpha\},
\end{equation}
 for which $\hat{F}_m(t)=m^{-1}\sum_{i=1}^m \ind{p_i\leq t}$ is the $p$-value family eCDF and $\hat{\pi}^{\IMS}_{0,\epsilon,\underline{\pi}_0,\kappa}$ is defined by \eqref{equ:pi0hatIMS}.
\end{definition}
Key properties for $\hat{\kappa}$ are gathered in Lemma~\ref{lem:kappahat}. 
In words, the procedure IMS-$\hat{\kappa}$ has a rejection set containing all the rejection sets of the $\kappa$-based IMS procedures $\{\ABH_\alpha(\hat{\pi}^{\IMS}_{0,\epsilon,\underline{\pi}_0,\kappa},\kappa),\kappa\in [0,1-\epsilon]\}$, that is,
\begin{equation}
    \label{equkappahatthebest}
    \ABH_\alpha(\hat{\pi}^{\IMS}_{0,\epsilon,\underline{\pi}_0,\hat{\kappa}},\hat{\kappa})=\bigcup_{\kappa\in [0,1-\epsilon]}\ABH_\alpha(\hat{\pi}^{\IMS}_{0,\epsilon,\underline{\pi}_0,\kappa},\kappa),
\end{equation}
where the procedures in \eqref{equkappahatthebest} are identified with their rejection sets.
Also, $\ABH_\alpha(\hat{\pi}^{\IMS}_{0,\epsilon,\underline{\pi}_0,\hat{\kappa}},\hat{\kappa})$ reduces to the procedure rejecting $p$-values below $\hat{\kappa}$.
In addition, the following result shows that the FDR control is maintained.

\begin{figure}[tbp]
        \centering
        \begin{minipage}{0.32\textwidth}
                \centering
                \includegraphics[clip, trim = 0cm 0cm 0cm 0cm, width = 1\textwidth]{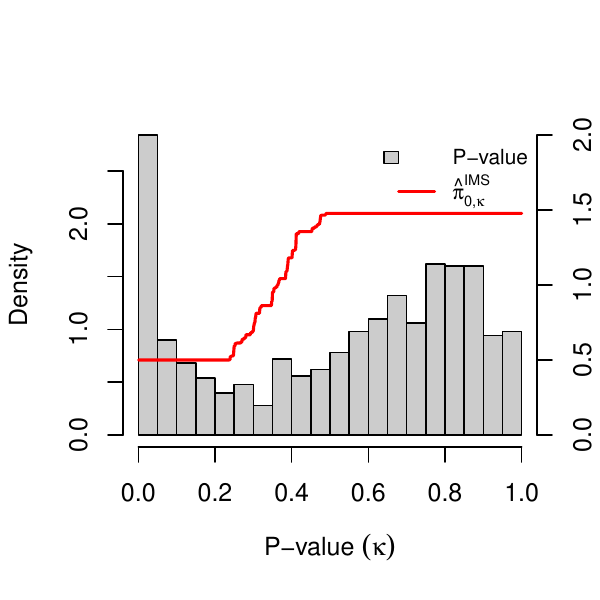}
        \end{minipage}\quad
                       \begin{minipage}{0.32\textwidth}
                \centering
                \includegraphics[clip, trim = 0cm 0cm 0cm 0cm, width = 1\textwidth]{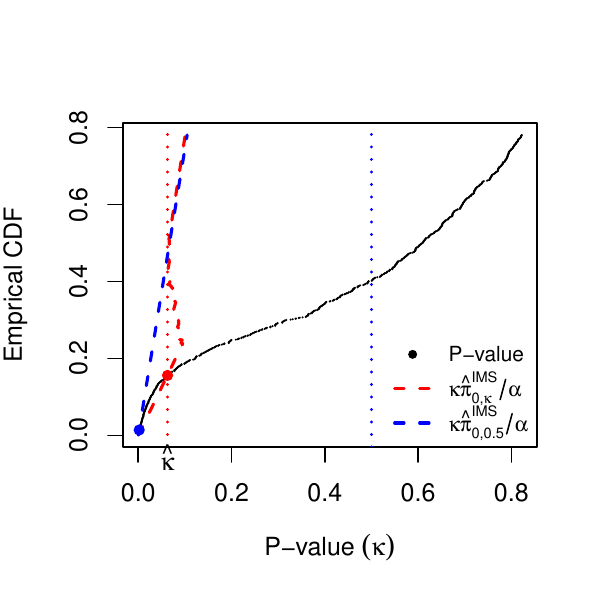}
        \end{minipage}
        \begin{minipage}{0.32\textwidth}
                \centering
                \includegraphics[clip, trim = 0cm 0cm 0cm 0cm, width = 1\textwidth]{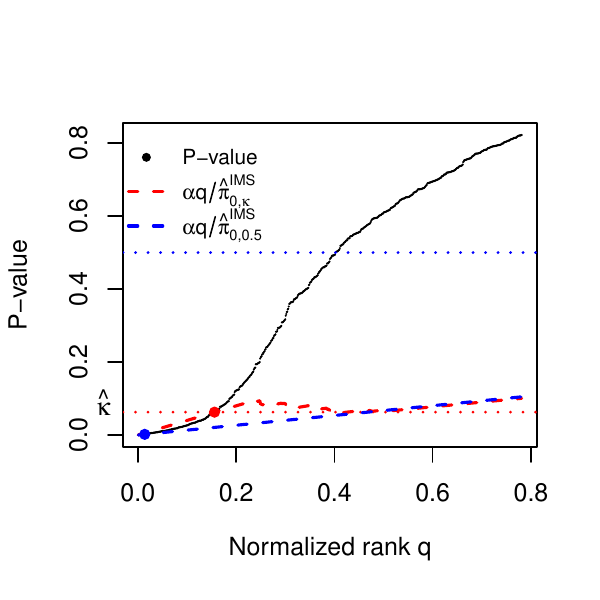}
        \end{minipage}
            \caption{Visualization of IMS with $\hat \kappa$. 
            The left panel shows the empirical distribution of the $p$-values together with $\hat{\pi}_{0,\epsilon,\underline{\pi}_0,\kappa}^{\mathrm{IMS}}$ nondecreasing in $\kappa$.
            The middle panel plots the empirical CDF against the $p$-value level or $\kappa$.
            It shows that $\hat{\kappa}$ is the largest value for which the empirical CDF at $\kappa$ exceeds $\kappa\hat{\pi}_{0,\epsilon,\underline{\pi}_0,\kappa}^{\mathrm{IMS}}/\alpha$.
            The red dot marks the rejection threshold $\hat{\kappa}$, which is larger than its blue
            counterpart under the fixed choice $\kappa=0.5$.
            The right panel is obtained by flipping the x and y axes of the middle panel,
            thereby plotting the ordered $p$-values against the normalized rank $q$, i.e., the rank of p-value divided by $m$. The red
            dashed curve denotes
            $\alpha q / \hat{\pi}_{0,\epsilon,\underline{\pi}_0,\kappa}^{\mathrm{IMS}}$ with $q = \hat F_m(\kappa)$,
            while the blue dashed curve corresponds to
            $\alpha q / \hat{\pi}_{0,\epsilon,\underline{\pi}_0,0.5}^{\mathrm{IMS}}$.}
        \label{fig:visualization.IMS.kappa}
\end{figure}

\begin{theorem}\label{thcontrolIMSkappahat}
   For any choice of parameters $\epsilon,\underline{\pi}_0\in (0,1)$, the IMS-$\hat{\kappa}$ procedure taken at level $\alpha\in (0,1)$ controls the FDR at level $\alpha$, that is,  $\FDR(\ABH_\alpha(\hat{\pi}^{\IMS}_{0,\epsilon,\underline{\pi}_0,\hat{\kappa}},\hat{\kappa}))\leq \alpha$ both in the conformal setting and in the independent setting for which the null $p$-values are additionally assumed marginally uniform. 
\end{theorem}
Theorem~\ref{thcontrolIMSkappahat} is proved in Section~\ref{proof:thcontrolIMSkappahat}. 
In Section~\ref{sec:optimality} (Corollary~\ref{corgenunif}), we show that the IMS-$\hat{\kappa}$ procedure satisfies an optimality property both in the  independent setting with a uniform null distribution and in the conformal setting. Again, the choice of the parameters $\epsilon,\underline{\pi}_0\in (0,1)$ is free for FDR control, but should be chosen wisely to ensure an optimized power. For instance, our power analysis suggests to take $\epsilon \propto m^{-1/4}$ (independence) or $\epsilon \propto m^{-1/8}$ (conformal), up to a log factor, and  $\underline{\pi}_0$ slowly tending to zero.

\subsection{Normalizing factors}\label{sec:ghcomputing}

To introduce $ C(m,\epsilon,\underline{\pi}_0)$ and $D(m,\epsilon,\underline{\pi}_0)$, let us first consider a least favorable distribution $Q_{0,s}$ for a family of $s\in [m]$ null $p$-values $\mbf{q}=(q_i)_{i\in [s]}$ defined as
\begin{itemize}
    \item independence setting:  $Q_{0,s} = \delta_0\times U(0,1)^{\otimes s-1}$, that is, $q_1=0$ and $(q_i)_{i\in [2,s]}$ are iid $U(0,1)$;
    \item conformal setting: $Q_{0,s}$ is the distribution of $(0,(\tilde{p}^{0,1}_j)_{j\in [s]})$ given by  \eqref{equ:configconf} when $\cH_0=[s]$ and when all scores $S_1,\dots,S_{n+s}$ are iid $U(0,1)$. 
\end{itemize}
Then we define in each case the normalizing factors
\begin{align}
    C(m,\epsilon,\underline{\pi}_0)&:= \max_{\underline{\pi}_0 m \leq s\leq m} c(s,\epsilon), \:\:\:c(s,\epsilon):= 
      \mathbb{E}_{\mbf{q}\sim Q_{0,s}}\left[1\vee \max_{\lambda\in \{q_i, i\in [s]\}\cap (0,1-\epsilon)}
      \frac{s(1-\lambda)}{1 \vee \sum_{i=1}^{s} \mathbf{1}\{q_i >  \lambda\}}\right]     
      ;\label{equ:gMS}\\
      D(m,\epsilon,\underline{\pi}_0)&:= \max_{\underline{\pi}_0 m \leq s\leq m} d(s,\epsilon),\:\:\: d(s,\epsilon):=
      \mathbb{E}_{\mbf{q}\sim  Q_{0,s}}\left[\max_{\substack{I=(a,b), 0\leq a\leq b-\epsilon\\ a,b\in \{q_i, i\in [s]\}\cup\{0,1\}}}
      \frac{s |I|}{1 \vee \sum_{i=1}^{s} \mathbf{1}\{q_i \in I\}}\right].\label{equ:gIMS}
\end{align}
In the conformal case, the quantities $C(m,\epsilon,\underline{\pi}_0)$ and $D(m,\epsilon,\underline{\pi}_0)$ also depend on $n$, but we remove it in the notation for short.

Quantities $c(s,\epsilon)$ in \eqref{equ:gMS} and $d(s,\epsilon)$ in \eqref{equ:gIMS} can be easily approximated by using a naive Monte-Carlo scheme. We provide a numerical evaluation of $c(s,\epsilon)$ and $d(s,\epsilon)$ over a range of $s$ and $n$ in Figure~\ref{fig:adjustment.constant} in the appendix, see Section~\ref{appe:sec:constant}.
\et{For $m$ vary large and in the independent case, we can use the Moscovich approximation (e.g. in ‘Fast calculation of boundary crossing probabilities for Poisson processes’, 2017)}

Nevertheless, evaluating $C(s,\epsilon,\underline{\pi}_0)$ and $D(s,\epsilon,\underline{\pi}_0)$ might be costly because $c(s,\epsilon)$ and $d(s,\epsilon)$ should be evaluated for all $s\in [\underline{\pi}_0 m,m]$, which increases rapidly the computation burden when $m$ gets large (cost quadratic in $m$ times the number of Monte Carlo simulations). 
Fortunately,  $c(s,\epsilon)$ and $d(s,\epsilon)$ are provably decreasing when $s$ is large enough. To see this, 
let for $\epsilon\in (0,1)$,
$$
N(\epsilon)= 1+\bigg( \frac{2\log(1/\epsilon)}{\log(1/(1-\epsilon))}\bigg) \vee \bigg(2 + (1/8)\log(1/(1-\epsilon))\bigg) ,
$$
which is such that $N(\epsilon)\sim 2 \epsilon^{-1}\log(1/\epsilon)$ and $N(e^{-\epsilon^2/8}) \sim (1/4) \log(1/\epsilon)$ when $\epsilon=o(1)$.
\begin{proposition}\label{prop:Neps}
    For any $\epsilon,\underline{\pi}_0\in (0,1)$, we have 
    \begin{equation}
        \label{equ-shortcut}
            \mbox{$C(m,\epsilon,\underline{\pi}_0)=c(\lceil m\underline{\pi}_0 \rceil,\epsilon)$ and $D(m,\epsilon,\underline{\pi}_0)=d(\lceil m\underline{\pi}_0 \rceil,\epsilon)$}
    \end{equation}
 in each of the two following cases:
    \begin{itemize}
        \item in the independent setting, if $m\geq [N(\epsilon)\vee (N(\epsilon')\vee (2/\epsilon)]/\underline{\pi}_0$.    In particular, this occurs for $m$ large enough if $\epsilon$, $\underline{\pi}_0$ are allowed to depend on $m$ in such a way  that $\epsilon\propto m^{-\beta} (\log m)^{\gamma} $ and $\underline{\pi}_0\propto 1/(\log m)^{\gamma}$ for $\beta\in (0,1)$, $\gamma\geq 0$;
        \item in the conformal setting, if $\epsilon\geq \sqrt{\frac{16\log(3(m+1)^4)+ 8 \log m}{ \underline{\pi}_0(m\wedge (n+1))}}
        $.    In particular, this occurs for $(n+1)\wedge m$ large enough if $\epsilon$, $\underline{\pi}_0$ are allowed to depend on $(n,m)$ in such a way  that $1\gg \epsilon \gg ((n+1)\wedge m)^{-1/2} (\log m)^{(1+\gamma)/2} $ and $\underline{\pi}_0\propto 1/(\log m)^{\gamma}$ for $\gamma\geq 0$.
    \end{itemize} 
 \end{proposition}
Proposition~\ref{prop:Neps} is proved in Section~\ref{proofprop:Neps}. This result significantly reduces the computation time of these factors. To compute the normalizing factors in practice, a recommendation  is therefore to check if the conditions of Proposition~\ref{prop:Neps} are satisfied by the chosen $m,\epsilon,\underline{\pi}_0$ (and $n$): if yes, use the computational shortcut  \eqref{equ-shortcut}, if not (which means that $m$ is not too large),  use the original expressions \eqref{equ:gMS}-\eqref{equ:gIMS}. 
We however note that this recommendation might be a bit theoretically-guided and overcautious because $c(s,\epsilon)$ and $d(s,\epsilon)$ seems decreasing  in $s$ even for small $s$ on the ranges we empirically tried, see Figure~\ref{fig:adjustment.constant} in Section~\ref{appe:sec:constant}.

The next issue is to theoretically bound the factors $C(m,\epsilon,\underline{\pi}_0)$ and $D(m,\epsilon,\underline{\pi}_0)$  to ensure that the new introduced procedures are not too conservative. The next result shows that these factors are close to $1$ when $m$ (and $n$) grows.

\begin{proposition}
    \label{rateCD}
     The following holds when $m$ tends to infinity while $\epsilon$, $\underline{\pi}_0$ are allowed to depend on $m$ (and $n$):
     \begin{itemize}
     \item in the independent case, for $\epsilon \gg \sqrt{(\log m)/(m\underline{\pi}_0)}$, we have for $m$ large enough 
     $$C(m,\epsilon,\underline{\pi}_0)\vee D(m,\epsilon,\underline{\pi}_0)-1 \leq 9/(\sqrt{m\underline{\pi}_0} \:\epsilon);$$
       \item  in the conformal case, 
       for $1\gg \epsilon \gg \sqrt{(\log m)/((n\wedge m)\underline{\pi}_0)}$, $n\gg \log m$, $\underline{\pi}_0\gg 1/m$, 
       we have for $n\wedge m$ large enough 
       $$C(m,\epsilon,\underline{\pi}_0)\vee D(m,\epsilon,\underline{\pi}_0)-1 \leq 34  /(\sqrt{(n\wedge m)\underline{\pi}_0} \:\epsilon^2).$$
 \end{itemize}
\end{proposition}

Proposition~\ref{rateCD} is proved in Section~\ref{proof:rateCD}, for which more accurate non-asymptotic bounds are provided (see Proposition~\ref{prop:boundingCD}).

Our numerical experiments suggest that the approximation factors are  larger in the conformal case (Figure~\ref{fig:adjustment.constant}). This is corroborated by our bounds, for which an additional $\epsilon$ appears in the denominator. While this is not an issue when choosing a constant value for $\epsilon$, the convergence rate is modified when considering $\epsilon$ tending to zero, which is typically the case to obtain an optimal power, see Section~\ref{sec:optimality}.

\et{We can reduce the computation times and the amplitude of these factors by reducing the lambda range where the infimum is taken. In practice interesting cf Appendix? Add parameters but not very sensitive at the end}

\section{Optimality with convergence rates}\label{sec:optimality}

\subsection{Models}

In this section, we show that our procedures are optimal  over specific parameter classes when $m$ gets large.
For this, let us consider the two following settings inspired from the two-group model \cite{ETST2001}:  
\begin{itemize}
    \item Independent two-group setting with parameters $(m,\pi_0,F_0,F_1)$: the $p$-values are generated from the hierarchical model where $H_i$, $i\in [m]$, are iid $\sim \mathcal{B}(1-\pi_0)$, and conditionally on $\mbf{H}=(H_i)_{i\in [m]}$, the $p$-values $\mbf{p}=(p_i)_{i\in [m]}$ are independently generated as $p_i\sim F_0$ if $H_i=0$ and $p_i\sim F_1$ if $H_i=1$. Unconditionally, the $p$-values are thus iid of common distribution given by the mixture $F=\pi_0 F_0 + \pi_1 F_1$;    
    \item Conformal two-group setting with parameters $(m,n,\pi_0,F_1)$: the scores $(S_i)_{i\in [n+m]}$ are generated as $(1-q_i)_{i\in [n+m]}$ for $(q_{i})_{i\in [n]}$ iid $\sim I$ and $(q_{n+i})_{i\in [m]}$ following the independent two-group setting with parameters $(m,\pi_0,I,F_1)$, where $I(x):=x\ind{x\in[0,1]}+\ind{x>1}$, $x\in [0,1]$, is the CDF of the uniform\footnote{Assuming that the scores are uniform under the null is not restrictive: any null score distribution (with a positive continuous density on $[0,1]$) can be described by changing $F_1$.} distribution on $(0,1)$. Then the conformal $p$-values are computed as usual according to \eqref{equ-confpvalues}. 
\end{itemize}
In both cases, we denote $\cH_0=\{i\in [m]\::\: H_i=0\}$ and $m_0=|\cH_0|$ (which are now random quantities) and we always suppose that $F_0$ and $F_1$ are continuous, so there are no ties in $p_i$ in the independent case, and no ties in the scores in the conformal case\footnote{The marginal distributions of the conformal $p$-values are not given by $F_0$ and $F_1$ in the conformal case: the conformal $p$-values \eqref{equ-confpvalues} have always discrete marginal distributions.}. Also we assume $F_0(x)\leq x$ for all $x\in [0,1]$ (super-uniformity under the null). This condition is denoted $F_0 \succeq U(0,1)$.
In this models, the distribution is indexed by the parameter $\theta=(\pi_0,F_0,F_1)$. 
To emphasize the data generation process parameter $\theta$, we denote the probabilities and expectations by $\P_\theta$ and $\E_\theta$ respectively throughout this section.

The results will be established with rates depending on $t=m$ (independent) or $t=(n,m)$ (conformal), with `$t$ large' meaning `$n\wedge m$ large' in the conformal case.

\subsection{Estimator class}\label{sec:estimclass}

The optimality  is established in this section for a class of $\pi_0$-estimators that satisfies two constraints. Given a parameter range $\Theta$ for $\theta=(\pi_0,F_0,F_1)$, the $\pi_0$-estimator $\tilde{\pi}_0$ is said 
\begin{itemize}
    \item[(i)] FDR-conservative over $\Theta$ after $\kappa$-capping (for some $\kappa\in(0,1]$) if: 
    \begin{equation}
        \label{equFDRconservative}
       \forall \alpha\in (0,1),\forall m\geq 2, \forall \theta\in \Theta, \: \FDR_\theta(\ABH_\alpha(\tilde{\pi}_0,\kappa))\leq \alpha;
    \end{equation}
      \item[(ii)] Concentrating over $\Theta$ if: for $t$ large enough,
 \begin{equation}
        \label{equconcentration}
  \forall \theta\in \Theta, 
\P_\theta(|\tilde{\pi}_0-\bar{\pi}_0(\theta)|\leq \eta_t(\theta))\geq 1-\delta_t(\theta),
    \end{equation}
    where $\bar{\pi}_0(\theta)>0$ (not necessarily equal to $\pi_0(\theta)$ and not depending on $t$) and $\eta_t(\theta)$, $\delta_t(\theta)$ are sequences converging to zero when $t$ tends to infinity ($\theta$ being fixed)\footnote{The convergence of $\tilde{\pi}_0$ does not necessarily imply the convergence of the corresponding adaptive BH's threshold, and can be strictly weaker in some settings.}. The convergence is with regard to $t=m$ in the independent case, $t=n\wedge m$ in the conformal case. 
\end{itemize}
  When the $\pi_0$-estimator $\tilde{\pi}_0$ satisfies (i) with $\kappa=1$, it is only said FDR-conservative over $\Theta$. The requirement \eqref{equconcentration} in (ii) is mild and holds for instance for any estimator $\tilde{\pi}_0$ that converges in $L^1$, that is,  with $\forall \theta\in \Theta,$ $\E_\theta(|\tilde{\pi}_0-\bar{\pi}_0(\theta)|)=o(1)$, by using $\eta_t(\theta)=\delta_t(\theta)=\E^{1/2}_\theta(|\tilde{\pi}_0-\bar{\pi}_0(\theta)|)$.


\subsection{A sufficient condition for optimality}

The quantity to target in our models is the identifiable part of $\pi_0$ which is given by\footnote{The maximum will be well defined in each of our considered cases.}
\begin{equation}\label{pi0stargen}
    \pi^*_0(\theta) := \max\{\pi_0\in (0,1]\: :\: \exists F^*_0,F^*_1 \mbox{ with $\theta^*=(\pi^*_0,F^*_0,F^*_1)\in \Theta$ and } F=F^*\},
\end{equation}
where $F=\pi_0 F_0 + (1-\pi_0)F_1 $ and $F^*= \pi^*_0 F^*_0 + (1-\pi^*_0)F^*_1$. This means that the configuration $\theta$ and $\theta^*$ are indistinguishable when observing the $p$-value family.

As Lemma~\ref{lemoptim} shows, $\pi^*_0(\theta)$ is optimal in the sense that any $\pi_0$-estimator satisfying \eqref{equFDRconservative} and \eqref{equconcentration} should concentrate around a value $\bar{\pi}_0(\theta)\geq \pi_0^*(\theta)$. 
As a result, any $\hat{\pi}_0$ estimator satisfying
\begin{equation}
        \label{sufficientcondition}
  \forall \theta\in \Theta, 
\P_\theta(\hat{\pi}_0\leq \pi_0^*(\theta)+\eta^*_t(\theta))\geq 1-\delta^*_t(\theta),
    \end{equation}
when $t$ is large enough, and for some rates $\eta^*_t(\theta)$ and $\delta^*_t(\theta)$ converging to zero when $t$ tends to infinity has an approximately maximum power over the considered $\pi_0$-estimators. 

\begin{theorem}
    \label{genthoptimal}
    Consider either the independent two-group setting or the conformal two-group setting. For any parameter space $\Theta$ for $(\pi_0,F_0,F_1)$, consider any $\pi_0$-estimator $\hat{\pi}_0$ satisfying \eqref{sufficientcondition}  with $\pi_0^*$ defined by \eqref{pi0stargen}. 
    Then the following holds: 
    \begin{itemize}
        \item[(i)] optimality without capping: 
for any  $\pi_0$-estimator $\tilde{\pi}_{0}$ that is FDR conservative  over $\Theta$ in the sense of \eqref{equFDRconservative} and  concentrating over $\Theta$  in the sense of \eqref{equconcentration}, we have  for all $\theta \in \Theta$ and for $t$ large enough
\begin{equation}\label{equoptimalnocapping}
    \sup_{\alpha\in (0,1]}\{\Pow_\theta(\ABH_{\alpha'}(\tilde{\pi}_0))-\Pow_\theta(\ABH_\alpha(\hat{\pi}_{0}))\}\leq \delta^*_t(\theta) + \delta_t(\theta),
\end{equation}
where $\alpha'=\alpha \big(1-\eta^*_t(\theta)/\pi_0^*(\theta)\big) (1-\eta_t(\theta)/(\pi_0^*(\theta)-\eta_t(\theta)))$.
        In addition, this applies to the oracle choice $\tilde{\pi}_0=\pi_{0}^*(\theta)$ with $ \eta_t(\theta)=\delta_t(\theta)=0$ and $\alpha'=\alpha \big(1-\eta^*_t(\theta)/\pi_0^*(\theta)\big) $.    
\item[(ii)] optimality with capping $\kappa\in (0,1)$:  
for any  $\pi_0$-estimator $\tilde{\pi}_{0}$ that is FDR conservative (possibly after $\kappa'$-capping  over $\Theta$ in the sense of \eqref{equFDRconservative}) and  concentrating over $\Theta$  in the sense of \eqref{equconcentration}, we have  for all $\theta \in \Theta$, and for $t$ large enough
\begin{equation}\label{equoptimalnocapping}
   \sup_{0<\alpha \leq \kappa\pi_0^*(\theta)}\{ \Pow_\theta(\ABH_{\alpha'}(\tilde{\pi}_0))-\Pow_\theta(\ABH_\alpha(\hat{\pi}_{0},\kappa))\}\leq \delta^*_t(\theta) + \delta_t(\theta),
\end{equation}
where $\alpha'=\alpha \big(1-\eta^*_t(\theta)/\pi_0^*(\theta)\big) (1-\eta_t(\theta)/(\pi_0^*(\theta)-\eta_t(\theta)))$.
         In addition, this applies to the oracle choice $\tilde{\pi}_0=\pi_{0}^*(\theta)$ with $ \eta_t(\theta)=\delta_t(\theta)=0$ and $\alpha'=\alpha \big(1-\eta^*_t(\theta)/\pi_0^*(\theta)\big) $.       
     \end{itemize}
\end{theorem}

Theorem~\ref{genthoptimal} is proved in Section~\ref{secproofgenthoptimal}. 
Case (i) is intended to be used when the $\pi_0$-estimator $\hat{\pi}_{0}$ is FDR conservative over $\Theta$ in the sense of \eqref{equFDRconservative} without any capping. It says that the power of $\ABH_\alpha(\hat{\pi}_{0})$ is approximately better than any other FDR controlling ABH procedure of the form $\ABH_{\alpha'}(\tilde{\pi}_{0})$ with $\alpha'\approx \alpha$. By contrast, if $\hat{\pi}_{0}$ is only FDR conservative over $\Theta$ with a capping $\kappa\in(0,1)$,  case (ii) should be used. It says that the power of $\ABH_\alpha(\hat{\pi}_{0},\kappa)$ is approximately better than any other FDR controlling ABH procedure of the form $\ABH_{\alpha'}(\tilde{\pi}_{0})$ with $\alpha'\approx \alpha$. For this, the $\kappa$-capping should not be effective in the threshold of $\ABH_\alpha(\hat{\pi}_{0},\kappa)$, which explains why the $\alpha$ range considered for optimality is reduced to $\alpha \leq \kappa\pi_0^*(\theta)$ in this case.

\subsection{Application to standard settings}\label{sec:applisettings}

\begin{table}[h!]
    \centering
    \begin{tabular}{|c||c|c|c|c|c|}
    \hline
     Settings   & $\Theta$ & $\pi_0^*(\theta) $ & Procedures 
     & $\eta^*_t$ & $\delta^*_t$\\ \hline\hline
   Indep super-unif  & $ \Theta_{\succeq U}$ \eqref{equThetasuperunif} & $\inf_{\lambda\in [0,1)} \frac{1-F(\lambda)}{1-\lambda} $ &MS 
   & $\propto \frac{\sqrt{\log m }}{m^{1/4}}$ & $ 1/m$\\
      Indep uniform     &  $\Theta_{U}$ \eqref{equThetaunif} & $\min_{\lambda\in [0,1]} f(\lambda) $ & 
      IMS-$\hat{\kappa}$
      & $\propto  \frac{\sqrt{\log m }}{m^{1/4}}$ & $ 1/m$\\
       Conformal     &  $\Theta_{U}$ \eqref{equThetaunif} & $\min_{\lambda\in [0,1]} f(\lambda) $ & IMS-$\hat{\kappa}$
       & $\propto \frac{\sqrt{\log m }}{(n\wedge m)^{1/8}} $&$ 1/m$
       \\\hline
    \end{tabular}
    \caption{Summary table for the optimality results when applying Theorem~\ref{genthoptimal}.
    }
    \label{tab:summary}
\end{table}

\paragraph{MS optimality in the independent super-uniform setting}

We consider the independent two-group model with a parameter $\theta=(\pi_0,F_0,F_1)$ that belongs to the set
\begin{equation}
    \label{equThetasuperunif}
    \Theta_{\succeq U} :=\bigg\{\theta=(\pi_0,F_0,F_1): \pi_0\in (0,1], F_0 \succeq U(0,1), 
    F= \pi_0 F_0+ (1-\pi_0) F_1 \mbox{ has density Lipschitz on $[0,1]$}
\bigg\} .
\end{equation}
Hence, $\Theta_{\succeq U}$ corresponds to the situation where the null is super-uniform, the alternative is left arbitrary and some density regularity is assumed. 

 \begin{lemma}
     \label{lem:pi0starsuperunif}
For the independent two-group model with parameter space $\Theta=\Theta_{\succeq U}$,       $\pi^*_0(\theta)$ in \eqref{pi0stargen} is given by 
\begin{equation}\label{pi0star}
    \pi_0^*(\theta) := \inf_{\lambda\in [0,1)} \frac{1-F(\lambda)}{1-\lambda} \in [ \pi_0,1].
\end{equation}
 \end{lemma}
The proof is provided in Section~\ref{proof:lem:pi0starsuperunif}. Hence Theorem~\ref{genthoptimal} shows that $\pi_0^*(\theta)$ is the oracle null proportion when working on $\Theta_{\succeq U}$.  Note that $\pi_0^*(\theta)=\pi_0$ is  possible if $f_0(x)=1$ for $x\in [0,1]$ and $f_1(x)=0$ for $x\in [\delta,1]$ for $0<\delta\leq 1$ for instance. The latter is related to the purity condition \cite{GW2004}. However, this assumption is not required  to apply our approach.

In \eqref{pi0star}, either the infimum is reached at some $\lambda^*\in [0,1)$, in which case it is a minimum and we have $1-F(\lambda^*)=\pi_0^*(\theta) (1-\lambda^*)$, or the infimum is reached when $\lambda\to 1$, in which case we let $\lambda^*=1$ and we have $\pi_0^*(\theta)=F'(1)=f(1)$. 
We summarize this by letting
\begin{equation}
    \label{equlambdastar}
    \lambda^* := \min\bigg\{ \lambda\in [0,1] \::\: \frac{1-F(\lambda)}{1-\lambda}=\pi_0^*(\theta)\bigg\},
\end{equation}
with some notation abuse when $\lambda^*=1$.

The form of $\pi_0^*(\theta)$ suggests that the MS estimator $ \hat{\pi}^{\MS}_{0,\epsilon,\underline{\pi}_0}$ should be close to $\pi_0^*(\theta)$. This is proved formally in the following result.

\begin{proposition}\label{propsufficientcondsuperunif}
    Choosing  $\epsilon = m^{-1/4}$ and $0.5 \wedge (162/\log(m))\leq \underline{\pi}_0\ll 1 $ provides  for any $\theta\in \Theta_{\succeq U}$, for $m$ large enough,
\begin{equation}
\label{concpi0hatindepcorboundaryoptim}
\P_\theta\bigg(   \hat{\pi}^{\MS}_{0,\epsilon,\underline{\pi}_0} \geq  \pi_0^*(\theta) +  4\frac{\sqrt{\log m }}{m^{1/4}}  \bigg)\leq 1/m.
\end{equation}
\end{proposition}

The proof of Proposition~\ref{propsufficientcondsuperunif} is provided in Section~\ref{proofpropsufficientcondsuperunif}.
Hence, $\hat{\pi}^{\MS}_{0,\epsilon,\underline{\pi}_0}$ satisfies \eqref{sufficientcondition} with $\eta^*_t(\theta)=4 (\log m )^{1/2}m^{-1/4}$ and $\delta^*_t(\theta)=1/m$. Since the MS procedure $\ABH_\alpha(\hat{\pi}^{\MS}_{0,\epsilon,\underline{\pi}_0})$ controls the FDR by Theorem~\ref{thcontrolMS}, Theorem~\ref{genthoptimal} (i) ensures that it has approximately maximum power over $\Theta_{\succeq U} $ in the following sense.

\begin{corollary}[Optimality of MS procedure for a super-uniform null distribution]
   \label{coroptimalMS} 
In the independent two-group model with parameter set $\Theta_{\succeq U}$ \eqref{equThetasuperunif}, for any  $\pi_0$-estimator that is FDR conservative over $\Theta_{\succeq U}$ in the sense of \eqref{equFDRconservative} and 
 consistent with concentration over $\Theta_{\succeq U}$  in the sense of \eqref{equconcentration}, the MS procedure with parameters $\epsilon = m^{-1/4}$ and $\underline{\pi}_0=0.5 \wedge (162/\log(m))$ is such that  for all $\theta \in \Theta_{\succeq U}$ and for $m$ large enough
\begin{equation}\label{equ:optimalityMSsuperunif}
    \sup_{\alpha\in (0,1]}\{\Pow_\theta(\ABH_{\alpha'}(\hat{\pi}_0))-\Pow_\theta(\ABH_\alpha(\hat{\pi}^{\MS}_{0,\epsilon,\underline{\pi}_0}))\}\leq 1/m + \delta_m(\theta),
\end{equation}
where $\alpha'=\alpha \big(1-(4/\pi_0^*(\theta))(\log m )^{1/2}m^{-1/4}\big) (1-\eta_m(\theta)/(\pi_0^*(\theta)-\eta_m(\theta)))$ and $\pi_0^*(\theta)$ is defined by \eqref{pi0star}.  
\end{corollary}

\begin{remark}
    \label{rem:RBsuboptimal}
    The right-boundary estimator, based on the stopping time \eqref{equ:rightboundary}, is in general sub-optimal on the parameter space $\Theta_{\succeq U}$: we can easily show that (by choosing suitably the $\lambda$-grid with $\lambda_1=\kappa$) it satisfies \eqref{equconcentration} with 
    \begin{equation}\label{equpi0barRBestimator}
   \bar{\pi}_0 := \frac{1-F( \bar\lambda)}{1- \bar\lambda} \geq \pi_0^*,\:\:\:\:\:
    \bar{\lambda} :=
    \sup\bigg\{\lambda\in [\kappa,1) \::\: \forall \lambda'< \lambda,  \frac{1-F(\lambda)}{1-\lambda}< \frac{1-F(\lambda')}{1-\lambda'}\bigg\}\leq \lambda^*.
\end{equation}
In general, $\bar{\pi}_0 > \pi_0^*$, for instance when $\lambda\mapsto (1-F(\lambda))/(1-\lambda)$ has two local minima at $\lambda_1<\lambda_2$ with $\lambda_1=\bar{\lambda}$ and $\lambda_2=\lambda^*$, or when it has a minimum occurring before $\kappa$. In particular, it does not enjoy the optimality property \eqref{equ:optimalityMSsuperunif}. This is exemplified in the simulation settings  (b.2) and (b.3) in Section~\ref{sec:simulation}, for which the right boundary is suboptimal, as the stopping time occurs too early.
\end{remark}

\paragraph{IMS optimality in the independent uniform and conformal settings}

In this section, we impose that the null distribution is (exactly) uniform and we demonstrate that the MS procedure is not optimal in general, whereas the IMS  procedure is, under a mild condition. 

We consider the independent two-group model with a parameter $\theta=(\pi_0,I,F_1)$ that belongs to the set
\begin{equation}
    \label{equThetaunif}
    \Theta_{U} :=\bigg\{\theta=(\pi_0,I,F_1): \pi_0\in (0,1), F_1 \mbox{ has a density $f_1$ Lipschitz on $[0,1]$}\bigg\} ,
\end{equation}
for which  $I(x):=x\ind{x\in[0,1]}+\ind{x>1}$, $x\in [0,1]$, is the CDF of the uniform distribution. Hence,  the alternative is left arbitrary (with a regular density) in $\Theta_{U}$. 

\begin{lemma}\label{lempi0starunif}
    In the independent two-group model with parameter space $\Theta=\Theta_{U}$,  $\pi^*_0(\theta)$ in \eqref{pi0stargen} is given by 
\begin{equation}\label{pi0starunif}
    \pi_0^*(\theta) := \min_{\lambda\in [0,1]} f(\lambda) = f(\lambda^*) \geq \pi_0,\:\:\:\mbox{ for  $\lambda^*=\min\{\lambda\in [0,1]\::\: f(\lambda)=\pi_0^*\}\in [0,1]$,}
\end{equation}
for the mixture density $f(x)=\pi_0  + \pi_1 f_1(x)$, $x\in[0,1]$.
\end{lemma}
The proof of Lemma~\ref{lempi0starunif} is provided in Section~\ref{proof:lempi0starunif}.
Note that $\pi_0^*=\pi_0$ is  possible if $f_1(x)=0$ for some $x\in [0,1]$. However, this assumption is not required  to apply our theory.
The form of $\pi_0^*(\theta)$ suggests that the IMS estimator $ \hat{\pi}^{\IMS}_{0,\epsilon,\underline{\pi}_0,\kappa}$ should be close to $\pi_0^*(\theta)$ when $\lambda^*\geq \kappa$. This is proved formally in the following result.

\begin{proposition}\label{propsufficientunifANDconformal}
For all $\theta\in \Theta_U$, letting $\bar{\pi}_0(\theta):=\min_{x\in [\kappa,1]}f(x)$, which is such that $\bar{\pi}_0(\theta)=\pi_0^*(\theta)$ when $\kappa\leq \lambda^*$, we have 
for $t$ large enough
\begin{equation}
\label{concpi0hatindepcorboundaryoptim}
\P\bigg(  \hat{\pi}^{\IMS}_{0,\epsilon,\underline{\pi}_0,\kappa}  \geq  \bar{\pi}_0(\theta) +  \eta_t   \bigg)\leq \delta_t, 
\end{equation}
\begin{itemize}
    \item  by choosing, in the independent two-group model ($t=m$),  $\epsilon = m^{-1/4}$, $\underline{\pi}_0=0.5 \wedge (41/\log(2m)) $, 
$\eta_m = 5\frac{\sqrt{2\log (2m) }}{m^{1/4}}$  and $\delta_m=1/m$;
\item by choosing , in the conformal two-group model ($t=(n, m)$),  $\epsilon = (n\wedge m)^{-1/8}$, $\underline{\pi}_0=0.5 \wedge (1/\log(m)) $,  $n\gg \log m$, $\eta_m = 10 (1+L(f))^{1/2}\frac{\sqrt{\log m }}{(n\wedge m)^{1/8}}$  and $\delta_m=1/m$;
\end{itemize}
\end{proposition}
The proof of Proposition~\ref{propsufficientunifANDconformal} is provided in Section~\ref{proofpropsufficientunifANDconformal}.

Hence, in both settings, we have that $\hat{\pi}^{\IMS}_{0,\epsilon,\underline{\pi}_0,\kappa}$ with the choice $\kappa=\lambda^*$ satisfies \eqref{sufficientcondition}. Following Theorem~\ref{genthoptimal} (ii), this leads to optimality for $\ABH_\alpha(\hat{\pi}^{\IMS}_{0,\epsilon,\underline{\pi}_0,\lambda^*},\lambda^*)$. However, $\lambda^*$ is unknown. Fortunately, the IMS-$\hat{\kappa}$ procedure dominates the latter procedure (because it dominates all capping procedures by Lemma~\ref{lem:kappahat}), which allows to deduce the following optimality result. 

\begin{corollary}[Optimality of IMS-$\hat{\kappa}$ procedure for a uniform null distribution] \label{corgenunif}
\et{recall what is $\epsilon$ here}
    Consider either the independent two-group setting  or the conformal two-group setting with the parameter space $\Theta_{U} $ in \eqref{equThetaunif} with the associated quantities $\pi^*_0$  and $\lambda^*$ given by \eqref{pi0starunif}. 
    Consider the rates $\eta^*_t=5 (2\log (2m) )^{1/2}m^{-1/4}$, $\delta^*_t=1/m$ in the independent two-group setting ($t=m$) and 
    $\eta^*_t=10 (1+L(f))^{1/2}\frac{\sqrt{\log m }}{(n\wedge m)^{1/8}} $, $\delta^*_t=1/m$ in the conformal setting ($t=(n,m)$). 
    Then the power of the IMS-$\hat{\kappa}$ procedure (with $\epsilon$ and $\underline{\pi}_0$ as in Proposition~\ref{propsufficientunifANDconformal}) is approximately better than any other FDR controlling ABH procedure of the form $\ABH_\alpha(\tilde{\pi}_{0})$ for some $\alpha$-range:
for any  $\pi_0$-estimator $\tilde{\pi}_{0}$ that is FDR conservative (possibly after $\kappa$-capping with $\kappa\in (0,1]$) over $\Theta$ in the sense of \eqref{equFDRconservative} and  concentrating over $\Theta$  in the sense of \eqref{equconcentration}, we have  for all $\theta \in \Theta$, and for $t$ large enough
\begin{equation}\label{equoptimalnocapping}
   \sup_{0<\alpha \leq \lambda^*\pi_0^*}\{ \Pow_\theta(\ABH_{\alpha'}(\tilde{\pi}_0))-\Pow_\theta(\ABH_\alpha(\hat{\pi}^{\IMS}_{0,\epsilon,\underline{\pi}_0,\hat{\kappa}},\hat{\kappa}))\}\leq \delta^*_t + \delta_t(\theta),
\end{equation}
where $\alpha'=\alpha \big(1-\eta^*_t/\pi_0^*\big) (1-\eta_t(\theta)/(\pi_0^*-\eta_t(\theta)))$.
         In addition, this applies to the oracle choice $\tilde{\pi}_0=\pi_{0}^*$ with $ \eta_t(\theta)=\delta_t(\theta)=0$ and $\alpha'=\alpha \big(1-\eta^*_t/\pi_0^*\big) $.          
\end{corollary}


\begin{remark}
    \label{rem:RBandMSsuboptimal}
    When considering the parameter space $\Theta_{U} $, both the right-boundary estimator and the MS can be suboptimal. This is because
    $$
\inf_{\lambda\in [0,1)} \frac{1-F(\lambda)}{1-\lambda} \leq  \min_{\lambda\in [0,1]} f(\lambda)  =\pi_0^*(\theta)
    $$
    with a possibly strict inequality, for instance when $F$ is strictly convex near $1$. Explicit examples are provided in our experiments in Section~\ref{sec:simulation}, see panels (a.2) and (a.3) of Figure~\ref{fig:simulation.uniform}, panels (c.2) and (c.3) of Figure~\ref{fig:simulation.conformal.1000} and panel (d.2) of Figure~\ref{fig:simulation.cifar10}.   
\end{remark}

\begin{remark}
    The choice of the parameters $(\epsilon,\underline{\pi}_0)$ in our optimality result is mainly of theoretical interest. It is done up to proportional constants that might be adjusted to increase the performance of the estimator for moderately large $m$. 
    \et{In our numerical experiments, a choice like $\underline{\pi}_0 = 0.5$ is suitable if the minimum $\pi_0^*$ is expected to occurs above $0.5$. As for $\epsilon$, taking $\epsilon$ around $0.1$ leads to satisfactory results (which is in accordance with the theoretical choice for $m=1000$).}  
\end{remark}

\section{Numerical experiments}
\label{sec:simulation}

\et{Include the grid restriction somewhere?}

\et{Report the $C,D$ values}

\subsection{Simulation settings}

Without further specification, we set the FDR level at $\alpha = 0.2$, $m = 1000$.
We consider three sets of data-generating mechanisms.  
\begin{enumerate}
    \item [(a)] Independent $p$-values, $F_0 = U(0,1)$. 
    
    \begin{enumerate}
        \item [(a.1)] Few strong signals. 
        We consider $\pi_0 = 0.9$. Null $p$-values follow $U(0,1)$. Alternative $p$-values follow $\Phi(\mathcal{N}(-2,1))$.
        In this scenario, the null proportion is close to one and the standard BH is not overly-conservative.

    \item [(a.2)] Close-to-one alternatives. 
    We consider $\pi_0 = 0.5$. Null $p$-values follow $U(0,1)$. Alternative $p$-values follow the mixture $0.5 \Phi(\mathcal{N}(-2,0.25))+0.5 \Phi(\mathcal{N}(2,0.25))$, and $f_1$ takes the minimal value at $p=0.5$. This construction yields a non-negligible fraction of alternative $p$-values close to $1$, resulting in a basin-shaped histogram of $p$-values.

    \item [(a.3)] Many signals. 
         We consider $m = 2500$, $\pi_0 = 0.2$. Null $p$-values follow $U(0,1)$. Alternatives follow the mixture $0.5\Phi(N(-1.5,1)) + 0.4\Phi(N(-1.25,1)) + 0.1\Phi(N(3,1))$. BH with the true null proportion will reject all hypotheses.
        
    \end{enumerate}

        \item [(b)] Independent $p$-values, $F_0 \succeq U(0,1)$.
    \begin{enumerate}
    \item [(b.1)] Conservative null. We consider $\pi_0 = 0.6$. Null $p$-values follow $\text{Beta}(3,1)$. Alternative $p$-values follow $\Phi(\mathcal{N}(-1.5,1))$.

    \item [(b.2)] Two local minima with conservative null.
     We consider $\pi_0=0.5$. Null $p$-values follow the mixture $0.5\Phi(\mathcal{N}(0,1)) + 0.5\Phi(\mathcal{N}(0.75,1))$.
     Alternative $p$-values are generated from the mixture $0.7\Phi(\mathcal{N}(-0.25,0.25)) + 0.3\Phi(\mathcal{N}(-2,0.25))$.

     \item [(b.3)] Many signals with conservative null.
        We consider $m = 2500$, $\pi_0 = 0.2$. Null $p$-values follow $\Phi(N(1,1))$. Alternatives follow the mixture $0.5\Phi(N(-1.5,1)) + 0.4\Phi(N(-1.25,1)) + 0.1\Phi(N(3,1))$.
    \end{enumerate}

    \item [(c)] Conformal $p$-values. We set the size of the calibration dataset at $n = 1000$.
    In settings (c.1) to (c.3), we use one minus the $p$-values from settings (a.1) to (a.3) as conformity scores for the cases of few strong signals, close-to-one alternatives, and many signals, respectively.
\end{enumerate}

We display the results with independent $p$-values in Figure~\ref{fig:simulation.uniform} and Figure~\ref{fig:simulation.conservative}, and conformal $p$-values ($n=1000$) in Figure~\ref{fig:simulation.conformal.1000} in the appendix.

\subsection{Methods for comparison}
We compare the BH procedure combined with various null proportion estimators. 
\begin{enumerate}[label=(\roman*)]
    \item Standard BH (std) which does not use a $\pi_0$ estimator.
    \item Adaptive BH with the null proportion estimator in \cite{benjamini2006adaptive} (BY). 
    Particularly, we first apply the standard BH at the reduced level $\alpha/(1+\alpha)$, denote the resulting number of rejections by $R_0$, and estimate the null proportion by $\hat{\pi}^{\mathrm{BY}} := 1 - R_0/m$. We then apply the standard BH again at level $\alpha /((1+\alpha)\hat{\pi}^{\mathrm{BY}})$ to obtain the final rejection set.
    \item Storey's estimator (S), for 
which we consider three tuning parameters: S.S with a small $\lambda=0.2$, 
S.M with a medium $\lambda=0.5$, and S.L with a large $\lambda=0.8$. No capping is required to provide FDR control (Corollary~\ref{corcontrolStorey}).

    \item Adaptive Storey (AS) procedure. We
    use stopping times in the grid $n_B/m \times \{0.2m/n_B,\ldots,0.8 m/n_B\}$ with the bin-size $n_B = 20$.
    The rejection threshold is capped at $0.2$.
    \item Minimal Storey (MS). We use $\lambda_{\max}=0.8$ for MS. We set $\underline{\pi}_0=0.5$ for all settings except (a.3), (b.3), (c.3). In (a.3), (b.3), (c.3) where there is believed to be a small proportion of nulls, we use $\underline{\pi}_0=0.2$. We experiment with a conservative $\underline{\pi}_0=0.5 >\pi_0$ in Figure~\ref{appe:fig:simulation.conservative.underline.pi0} in the appendix.
    No capping is required to provide FDR control (Theorem~\ref{thcontrolMS}).

    \item [(vi)] Interval-based minimal Storey with $\hat \kappa$ (IMS)\footnote{In the current numerical experiments, for independent $p$-values, $m$, $n$ are large enough to ensure the monotonicity of $c(s,\epsilon)$ and $d(s,\epsilon)$ by Proposition~\ref{prop:Neps}. For conformal $p$-values, we precompute $c(s,\epsilon)$ and $d(s,\epsilon)$ over a range of $s$ values and use their maximum.}. We use $\varepsilon = m^{-1/4}$ for independent $p$-values and $\varepsilon = m^{-1/8}$ for conformal $p$-values, and the same $\underline{\pi}_0$ as that of MS.  
    It uses the capping at $\hat \kappa$ to provide FDR control (Theorem~\ref{thcontrolIMSkappahat}).
    \item [(vii)] Oracle (orc). The oracle uses the true $\pi_0$.
\end{enumerate}
All procedures are guaranteed to control the FDR at level $\alpha$, except IMS for conservative nulls (setting (b.1), (b.2), (b.3)) and is thus not compared therein.

\subsection{
Results
}

For independent $p$-values with $U(0,1)$ nulls (Figure~\ref{fig:simulation.uniform}), all methods, except standard BH and BY, perform comparably to the oracle in the setting with a small
number of strong signals (panel (a.1)).
When a subset of alternatives yield large $p$-values (panel (a.2)), for example when the true effect goes opposite to the prespecified direction, IMS is the most powerful; the other Storey-type methods overestimate $\pi_0$ due to treating these large alternative $p$-values as nulls.
When signals are abundant (panel (a.3)), AS is affected by capping;
MS is free from capping and is more powerful than the Storey's method with a fixed $\lambda$ for being adaptive; IMS with $\hat{\kappa}$ is comparably powerful due to its data-driven $\hat{\kappa}$, and we show in the appendix that the capping is more pronounced for IMS with a fixed $\kappa$ (Figure~\ref{appe:fig:IMS.kappa}).
The power comparison based on conformal $p$-values is similar (Figure~\ref{fig:simulation.conformal.1000} in the appendix) except that the correction constants of IMS are more conservative.

For independent $p$-values with conservative nulls
(Figure~\ref{fig:simulation.conservative}), IMS may fail to control the FDR and is therefore excluded from comparison. Across panels (b.1) to (b.3), MS is the most powerful. In panel (b.2), when the $p$-value  histogram exhibits multiple local minima, AS may terminate early and yield a conservative estimator of
the null proportion; by contrast, MS is always able to locate the global minimum.
In panel (b.3), the signal is abundant, and AS is affected by capping, as in panel (a.3).

\begin{figure}[tbp]
        \centering
        \begin{minipage}{0.8\textwidth}
                \centering
                \includegraphics[clip, trim = 0cm 1cm 0cm 0cm, width = 1\textwidth]{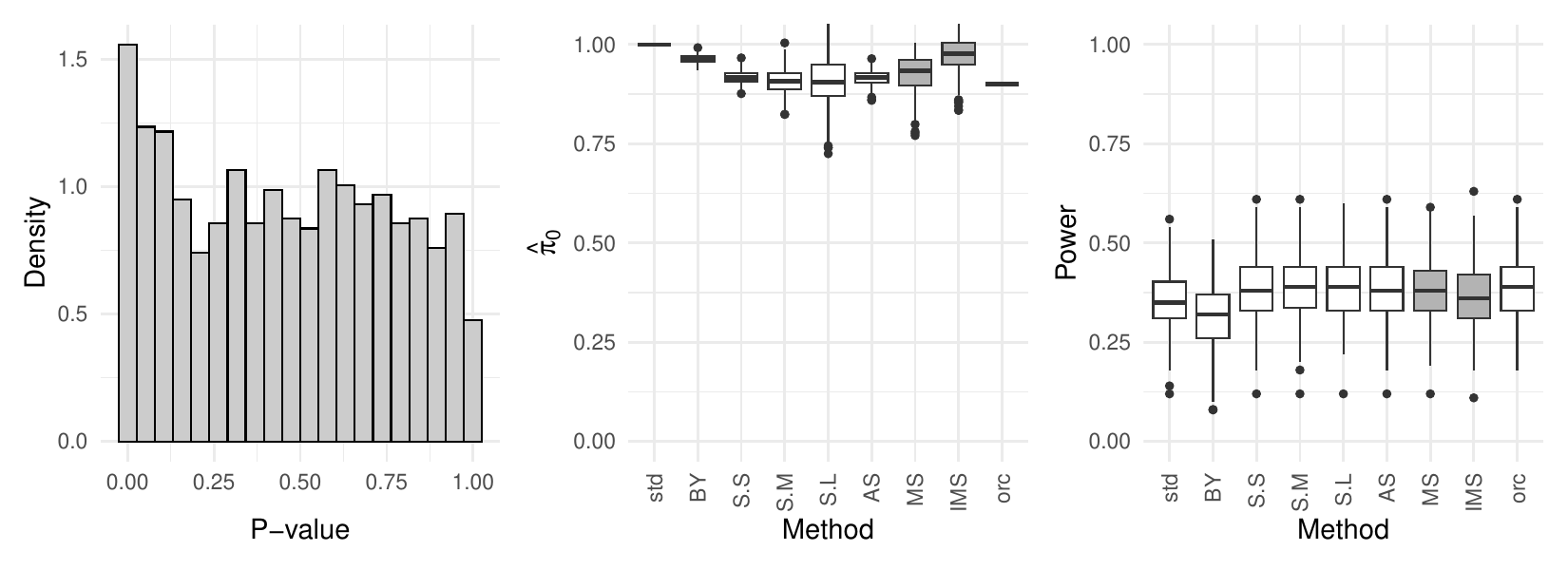}
        \subcaption*{(a.1) Few strong signals}
        \end{minipage}
        \begin{minipage}{0.8\textwidth}
                \centering
                \includegraphics[clip, trim = 0cm 1cm 0cm 0cm, width = 1\textwidth]{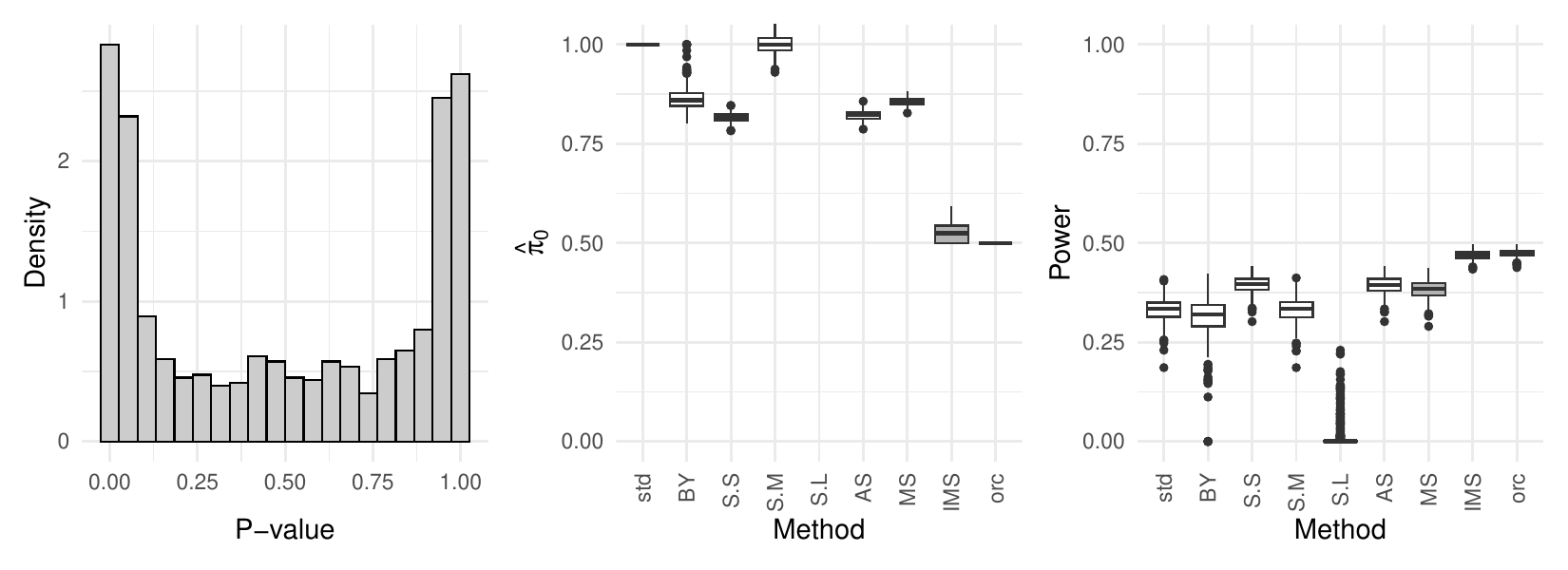}
        \subcaption*{(a.2) Close-to-one alternatives}
        \end{minipage}
     \begin{minipage}{0.8\textwidth}
                \centering
                \includegraphics[clip, trim = 0cm 1cm 0cm 0cm, width = 1\textwidth]{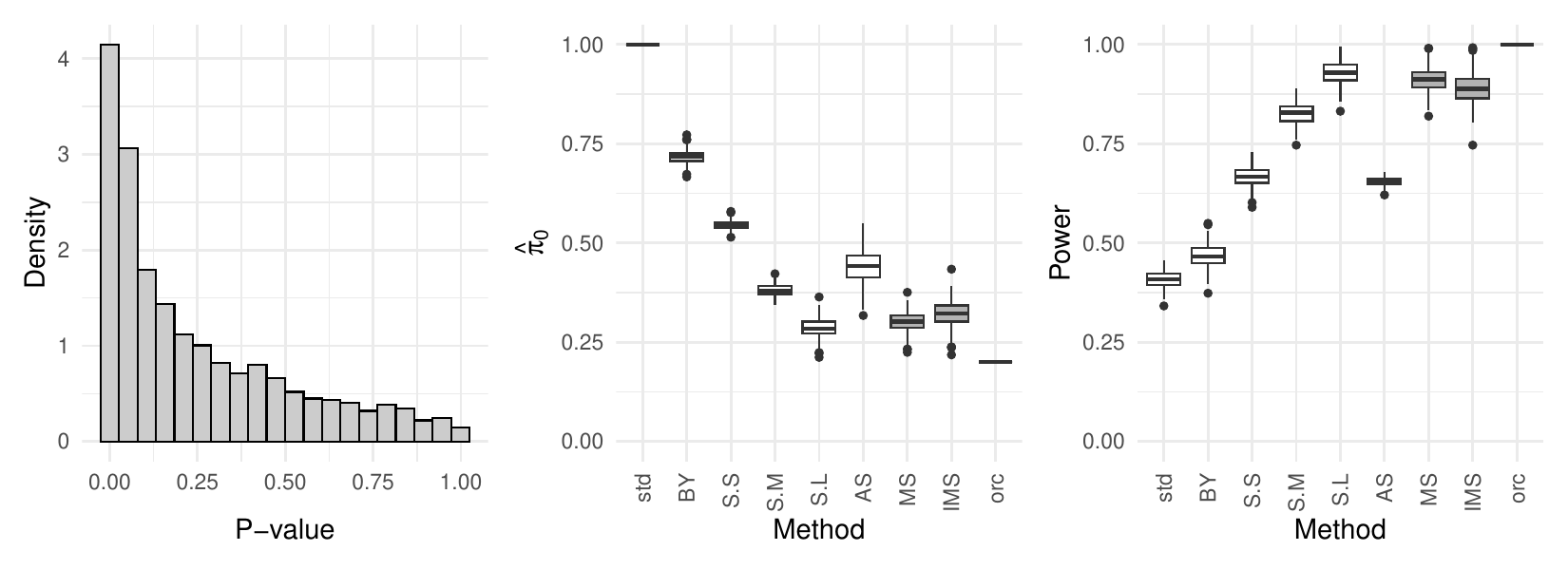}
        \subcaption*{(a.3) Many strong signals}
        \end{minipage}     

            \caption{Comparison of null proportion estimators with \emph{independent} $p$-values, $F_0 = U(0,1)$.
           Column 1 displays the histogram of the $p$-values of a random realization;
           Column 2 contains boxplots of the null proportion estimators; 
           Column 3 displays the power.
            In Column 2 and 3, MS and IMS are in grey.
            Each scenario is repeated $400$ times.}
        \label{fig:simulation.uniform}
\end{figure}

\begin{figure}[tbp]
        \centering
        \begin{minipage}{0.8\textwidth}
                \centering
                \includegraphics[clip, trim = 0cm 1cm 0cm 0cm, width = 1\textwidth]{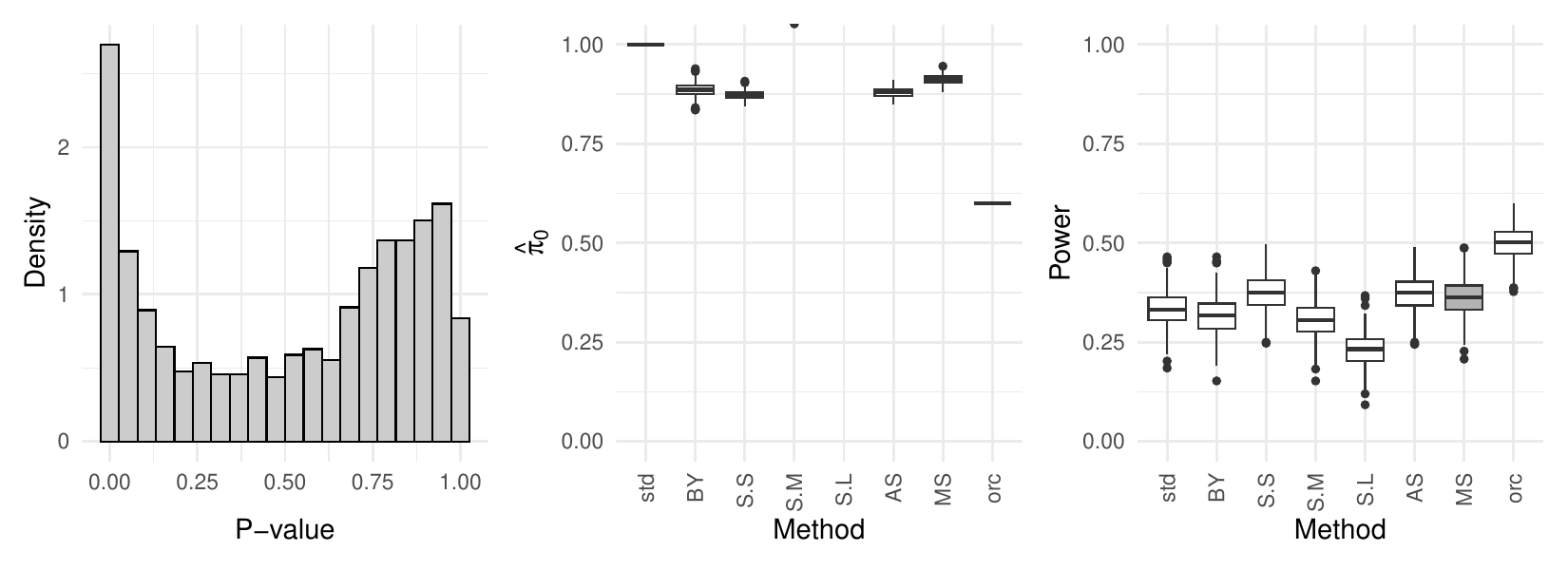}
        \subcaption*{(b.1) Conservative null}
        \end{minipage}
        \begin{minipage}{0.8\textwidth}
                \centering
                \includegraphics[clip, trim = 0cm 1cm 0cm 0cm, width = 1\textwidth]{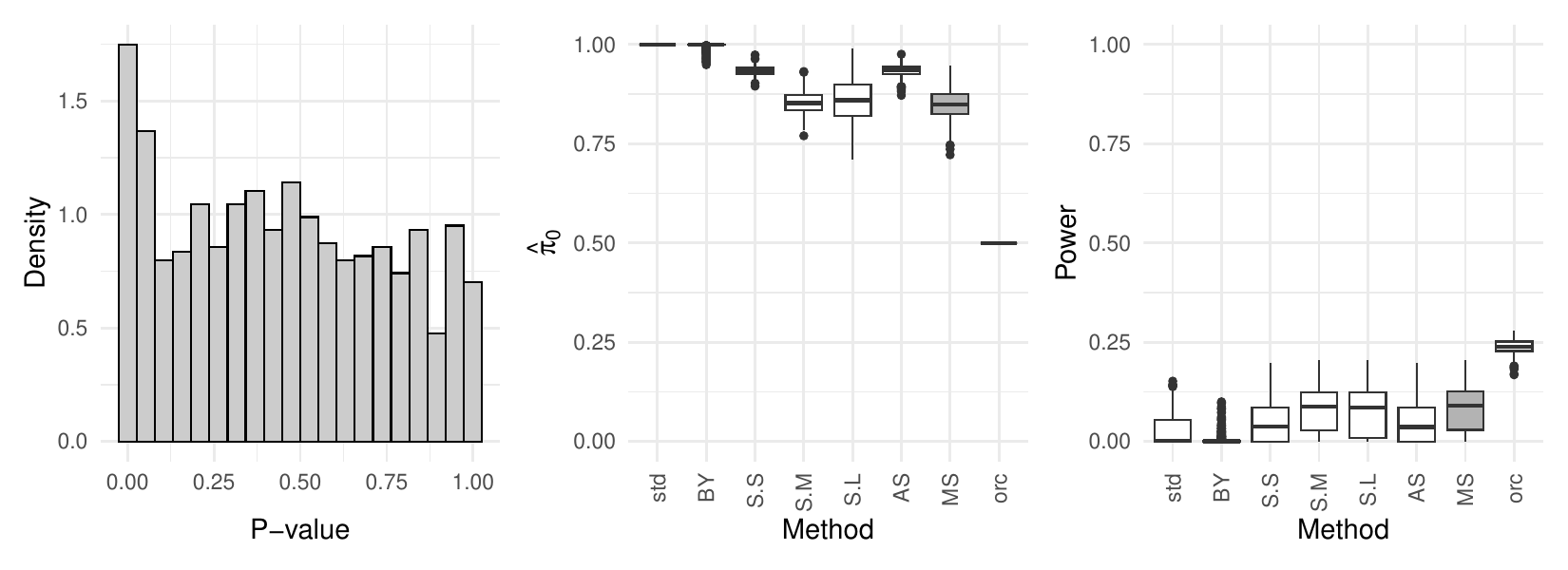}
        \subcaption*{(b.2) Two local minima with conservative null}
        \end{minipage}
        \begin{minipage}{0.8\textwidth}
                \centering
                \includegraphics[clip, trim = 0cm 1cm 0cm 0cm, width = 1\textwidth]{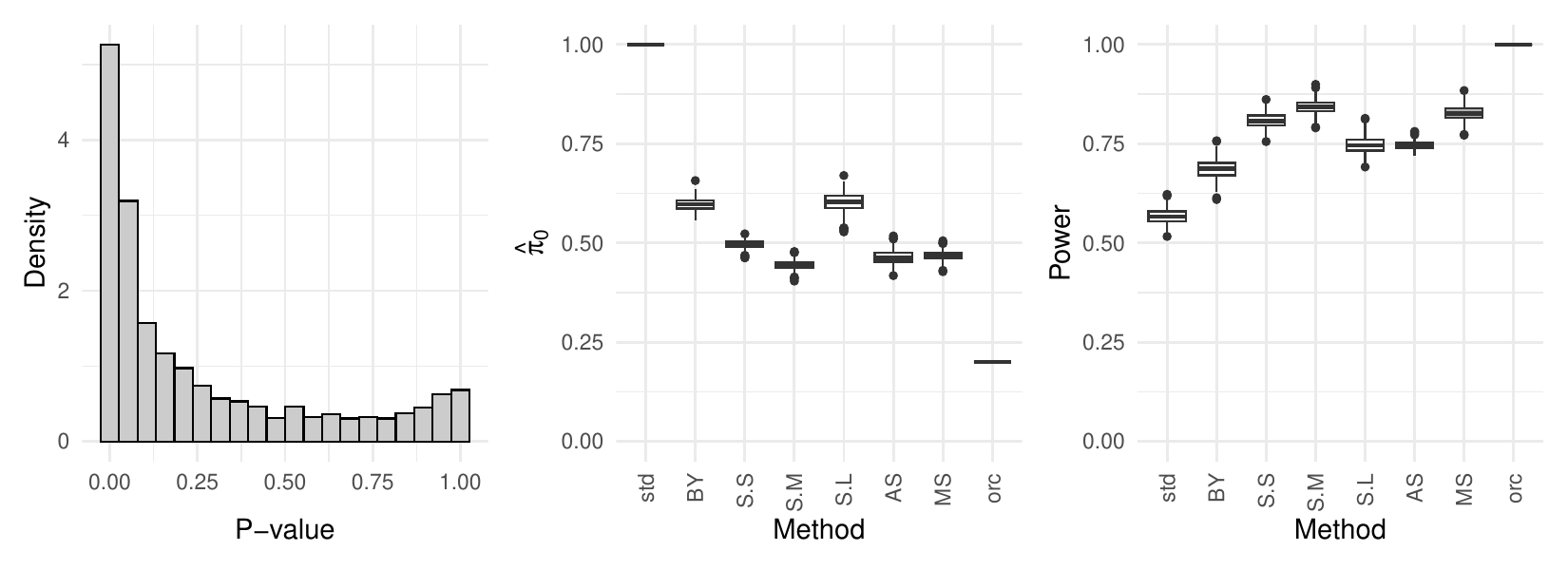}
        \subcaption*{(b.3) Many strong signals with conservative null}
        \end{minipage}

            \caption{Comparison of null proportion estimators with independent $p$-values, $F_0 \succeq U(0,1)$.
           Column 1 displays the histogram of the $p$-values of a random realization;
           Column 2 contains boxplots of the null proportion estimators; 
           Column 3 displays the power.
            In Column 2 and 3, MS is in grey. IMS does not control FDR for conservative nulls and is thus not compared.
            Each scenario is repeated $400$ times.}
        \label{fig:simulation.conservative}
\end{figure}

\section{Real data analysis: CIFAR-10 data}\label{sec:real.data}

We analyze conformal $p$-values generated from the CIFAR-10 dataset following \cite{marandon2024adaptive}. 
We consider two scenarios differing in the configuration of null, alternative, and null proportion.
In each setting, for each observation, the classifier produces predicted probabilities over all $10$ classes. The conformity score is taken to be the sum of the predicted probabilities over the alternative classes. 
Thus, larger scores provide stronger evidence against the null. 

\begin{itemize}
    \item [(d.1)] {Animal v.s. transportation, $\pi_0 = 0.8$.}
    Null classes are animals (cat, deer, dog, horse, bird, frog), and alternative classes are transportation tools (car, truck, airplane, ship). 
    The signal is strong, since pictures of animals and transportation are well separated. 
    The validation set contains $n = 2000$ observations, and the test set contains $m = 1000$ observations, of which $m_0 = 800$ are nulls.

    \item [(d.2)] {Land v.s. non-land, $\pi_0 = 0.2$.} 
    Null classes are land-based objects (car, truck, cat, deer, dog, horse), and alternative classes are objects not typically found on land (airplane, ship, bird, frog). The signal is weaker than in the animal-versus-transportation setting. The validation set contains $n = 2000$ observations, and the test set contains $m = 1000$ observations, of which $m_0 = 200$ are nulls.
\end{itemize}

We compare the null proportion estimators considered in Section~\ref{sec:simulation}. Overall, our proposal MS and IMS tend to produce the highest power (IMS is slightly less powerful because of its larger correction constant).
For AS, in the ``land vs. non-land'' setting, AS may yield a larger null proportion estimate because the $p$-value histogram is bumpy and the stopping time may terminate early. AS's rejection threshold is also influenced by capping, leading to an additional loss of power. 
For S.L., in the ``animal vs. transportation'' setting, S.L is more variable when $n$ is small; for S.M. and S.S., in the ``land vs. non-land'' setting, S.M and S.S are more conservative in estimating the null proportion as $f_1$ obtains the minimal value around $1$.

\begin{figure}[tbp]
        \centering
        \begin{minipage}{0.8\textwidth}
                \centering
                \includegraphics[clip, trim = 0cm 1cm 0cm 0cm, width = 1\textwidth]{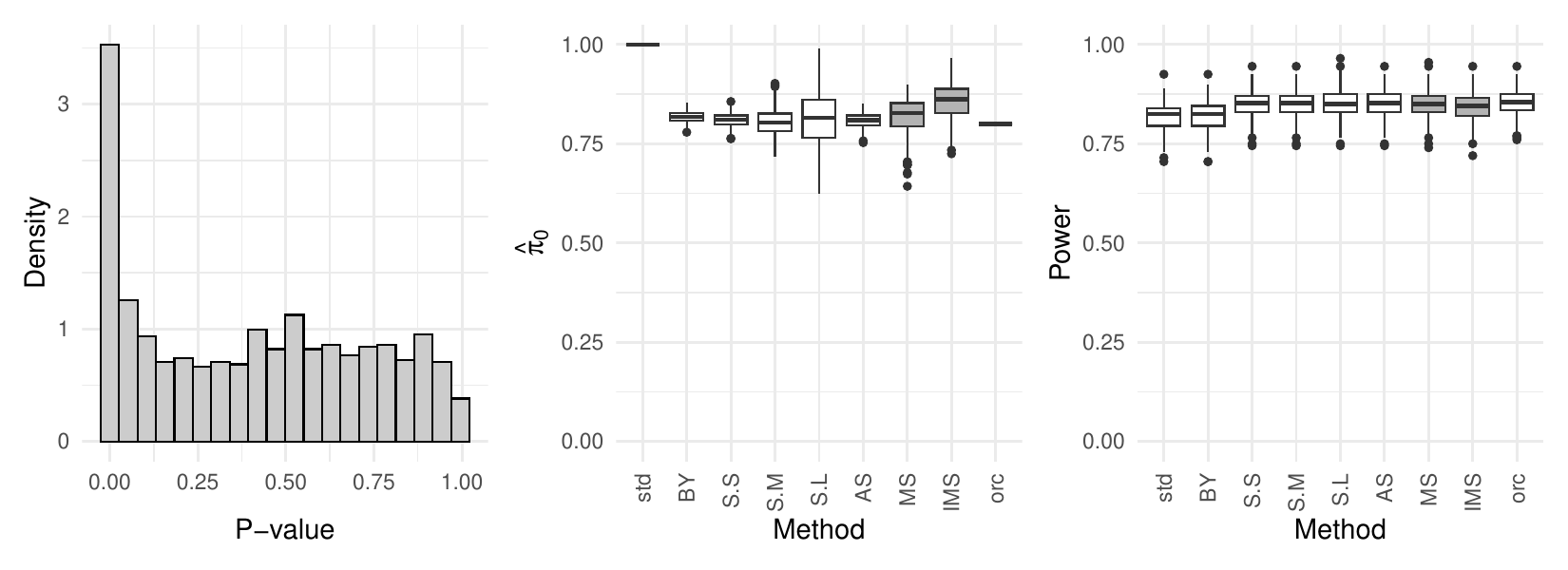}
        \subcaption*{(d.1) Animal v.s. transportation, $\pi_0 = 0.8$.}
        \end{minipage}      
        \begin{minipage}{0.8\textwidth}
                \centering
                \includegraphics[clip, trim = 0cm 1cm 0cm 0cm, width = 1\textwidth]{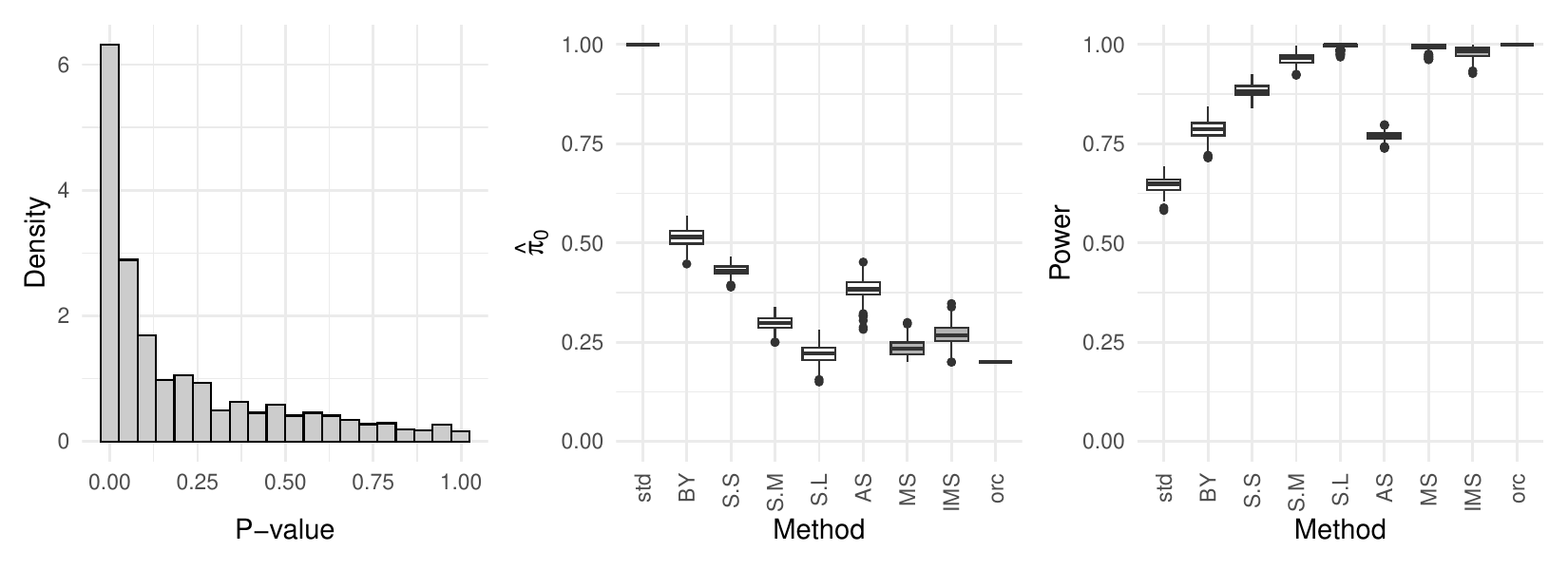}
        \subcaption*{(d.2) Land v.s. non-land,  $\pi_0 = 0.2$} 
        \end{minipage}   

            \caption{Comparison of null proportion estimators with conformal $p$-values derived from CIFAR-10 dataset. 
            For more details, see the caption of Figure~\ref{fig:simulation.uniform}.}
        \label{fig:simulation.cifar10}
\end{figure}

\section{Discussion}

This paper solved a long-standing issue in multiple testing, by identifying optimal $\pi_0$-estimators that provide finite sample FDR control when plugged into BH procedure. This was only partly addressed in previous literature, by making strong distributional 
assumption on the alternative $p$-value distribution (typically, strictly convex density \citealp{gao2025adaptive}). By contrast, we studied cases where no assumption is imposed on the alternative distribution.  
We showed that the MS procedure is suitable for the case of super-uniform null distribution while the IMS procedure is suitable for an exactly uniform null distribution. Beyond the classical multiple testing setting, we also showed that our IMS approach is suitable for a conformal setting where the scores are continuous, making our approach usable in combination with any score function produced by machine learning algorithms.  

Our IMS procedure used a capping $\hat{\kappa}$, which has the advantage to be data-driven, so adjustable from the data. However, while it is needed in our FDR controlling proof, the capping might still be a proof artifact. It impacts the quality of our optimality result:  the IMS procedure is optimal only  for $\alpha$ being in a vicinity of $0$ (Corollary~\ref{corgenunif}). 
A future challenging research direction would be to relax this capping to provide a procedure optimal on the whole range $\alpha\in (0,1)$ in the settings with uniform nulls.

\et{Reduce constant factors $C,D$ by considering a reduced grid (inspection points) while still ensuring asymptotic optimality? Also finite sample optimality?}

\section*{Acknowledgement}
{We are grateful to Sebastian D\"olher for very useful comments, including the calculation of the normalizing factors in the independent case.}
The authors acknowledge grants ANR-21-CE23-0035 (ASCAI) and ANR-23-CE40-0018-01 (BACKUP) of the French National Research Agency ANR, the Emergence project MARS of Sorbonne Universit\'e.

\bibliographystyle{plain}
\bibliography{biblio}

@article{heesen2016dynamic,
  title={Dynamic adaptive multiple tests with finite sample FDR control},
  author={Heesen, Philipp and Janssen, Arnold},
  journal={Journal of Statistical Planning and Inference},
  volume={168},
  pages={38--51},
  year={2016},
  publisher={Elsevier}
}

@article{jin2008proportion,
  title={Proportion of non-zero normal means: universal oracle equivalences and uniformly consistent estimators},
  author={Jin, Jiashun},
  journal={Journal of the Royal Statistical Society Series B: Statistical Methodology},
  volume={70},
  number={3},
  pages={461--493},
  year={2008},
  publisher={Oxford University Press}
}

@article{rebafka2022powerful,
  title={Powerful multiple testing of paired null hypotheses using a latent graph model},
  author={Rebafka, Tabea and Roquain, {\'E}tienne and Villers, Fanny},
  journal={Electronic Journal of Statistics},
  volume={16},
  number={1},
  pages={2796--2858},
  year={2022},
  publisher={The Institute of Mathematical Statistics and the Bernoulli Society}
}

@article{carpentier2021estimating,
  title={Estimating minimum effect with outlier selection},
  author={Carpentier, Alexandra and Delattre, Sylvain and Roquain, Etienne and Verzelen, Nicolas},
  journal={The Annals of Statistics},
  volume={49},
  number={1},
  pages={272--294},
  year={2021},
  publisher={JSTOR}
}

@article{cai2010optimal,
  title={Optimal Rates of Convergence for Estimating the Null Density and Proportion of Nonnull Effects in Large-Scale Multiple Testing},
  author={Cai, T Tony and Jin, Jiashun},
  journal={The Annals of Statistics},
  volume={38},
  number={304},
  pages={100--145},
  year={2010}
}

@article{ignatiadis2021covariate,
  title={Covariate powered cross-weighted multiple testing},
  author={Ignatiadis, Nikolaos and Huber, Wolfgang},
  journal={Journal of the Royal Statistical Society Series B: Statistical Methodology},
  volume={83},
  number={4},
  pages={720--751},
  year={2021},
  publisher={Oxford University Press}
}

@article{durand2019adaptive,
  title={Adaptive p-value weighting with power optimality},
  author={Durand, Guillermo},
  journal={Electronic Journal of Statistics},
  volume={13},
  pages={3336--3385},
  year={2019}
}

@article{roquain2022false,
  title={False discovery rate control with unknown null distribution: Is it possible to mimic the oracle?},
  author={Roquain, Etienne and Verzelen, Nicolas},
  journal={The Annals of Statistics},
  volume={50},
  number={2},
  pages={1095--1123},
  year={2022},
  publisher={Institute of Mathematical Statistics}
}

@article{heller2021optimal,
  title={Optimal control of false discovery criteria in the two-group model},
  author={Heller, Ruth and Rosset, Saharon},
  journal={Journal of the Royal Statistical Society Series B: Statistical Methodology},
  volume={83},
  number={1},
  pages={133--155},
  year={2021},
  publisher={Oxford University Press}
}

@article{celisse2010cross,
  title={A cross-validation based estimation of the proportion of true null hypotheses},
  author={Celisse, Alain and Robin, St{\'e}phane},
  journal={Journal of Statistical Planning and Inference},
  volume={140},
  number={11},
  pages={3132--3147},
  year={2010},
  publisher={Elsevier}
}

@article{liang2012adaptive,
  title={Adaptive and dynamic adaptive procedures for false discovery rate control and estimation},
  author={Liang, Kun and Nettleton, Dan},
  journal={Journal of the Royal Statistical Society Series B: Statistical Methodology},
  volume={74},
  number={1},
  pages={163--182},
  year={2012},
  publisher={Oxford University Press}
}

@article{li2019multiple,
  title={Multiple testing with the structure-adaptive Benjamini--Hochberg algorithm},
  author={Li, Ang and Barber, Rina Foygel},
  journal={Journal of the Royal Statistical Society Series B: Statistical Methodology},
  volume={81},
  number={1},
  pages={45--74},
  year={2019},
  publisher={Oxford University Press}
}

@misc{dohler2023unifiedclassnullproportion,
      title={A unified class of null proportion estimators with plug-in FDR control}, 
      author={Sebastian D{\"o}hler and Iqraa Meah},
      year={2023},
      eprint={2307.13557},
      archivePrefix={arXiv},
      primaryClass={stat.ME},
      url={https://arxiv.org/abs/2307.13557}, 
}

@article{guo2020adaptive,
  title={Adaptive controls of FWER and FDR under block dependence},
  author={Guo, Wenge and Sarkar, Sanat},
  journal={Journal of Statistical Planning and Inference},
  volume={208},
  pages={13--24},
  year={2020},
  publisher={Elsevier}
}

@article{heesen2015inequalities,
  title={Inequalities for the false discovery rate (FDR) under dependence},
  author={Heesen, Philipp and Janssen, Arnold},
  year={2015}
}

@misc{Ruellan2026,
  title={Conformal {DKW} inequality under covariate shift},
  author={Malo Ruellan and Etienne Roquain},
  year={2026},
  note={In preparation}
}

@article{benjamini2006adaptive,
  title={Adaptive linear step-up procedures that control the false discovery rate},
  author={Benjamini, Yoav and Krieger, Abba M and Yekutieli, Daniel},
  journal={Biometrika},
  volume={93},
  number={3},
  pages={491--507},
  year={2006},
  publisher={Oxford University Press}
}

@article{gao2025adaptive,
  title={An adaptive null proportion estimator for false discovery rate control},
  author={Gao, Zijun},
  journal={Biometrika},
  volume={112},
  number={1},
  pages={asae051},
  year={2025},
  publisher={Oxford University Press}
}

@article{tony2019covariate,
  title={Covariate-assisted ranking and screening for large-scale two-sample inference},
  author={Tony Cai, T and Sun, Wenguang and Wang, Weinan},
  journal={Journal of the Royal Statistical Society Series B: Statistical Methodology},
  volume={81},
  number={2},
  pages={187--234},
  year={2019},
  publisher={Oxford University Press}
}

@article{marandon2024adaptive,
  title={Adaptive novelty detection with false discovery rate guarantee},
  author={Marandon, Ariane and Lei, Lihua and Mary, David and Roquain, Etienne},
  journal={The Annals of Statistics},
  volume={52},
  number={1},
  pages={157--183},
  year={2024},
  publisher={Institute of Mathematical Statistics}
}

@article{bates2023testing,
  title={Testing for outliers with conformal p-values},
  author={Bates, Stephen and Cand{\`e}s, Emmanuel and Lei, Lihua and Romano, Yaniv and Sesia, Matteo},
  journal={Ann. Statist.},
  volume={51},
  number={1},
  pages={149--178},
  year={2023},
  publisher={Institute of Mathematical Statistics}
}

@Article{SC2007,
 Author = {Wenguang {Sun} and T. Tony {Cai}},
 Title = {{Oracle and adaptive compound decision rules for false discovery rate control}},
 FJournal = {{Journal of the American Statistical Association}},
 Journal = {{J. Am. Stat. Assoc.}},
 ISSN = {0162-1459; 1537-274X/e},
 Volume = {102},
 Number = {479},
 Pages = {901--912},
 Year = {2007},
 Publisher = {Taylor \& Francis, Philadelphia, PA; American Statistical Association (ASA), Alexandria, VA},
 Language = {English},
 MSC2010 = {62-XX}
}

@article{BH1995,
	author = {Benjamini, Yoav and Hochberg, Yosef},
	coden = {JSTBAJ},
	fjournal = {Journal of the Royal Statistical Society. Series B. Methodological},
	issn = {0035-9246},
	journal = {J. Roy. Statist. Soc. Ser. B},
	mrclass = {62J15},
	mrnumber = {MR1325392 (96d:62143)},
	number = {1},
	pages = {289--300},
	title = {Controlling the false discovery rate: a practical and powerful approach to multiple testing},
	volume = {57},
	year = {1995}}

@article{BH2000,
	author = {Benjamini, Yoav and Hochberg, Yosef},
	journal = {J. Behav. Educ. Statist.},
	pages = {60--83},
	title = {On the adaptive control of the false discovery rate in multiple testing with independent statistics},
	volume = {25},
	year = {2000}}

@article{BKY2006,
	author = {Benjamini, Yoav and Krieger, Abba M. and Yekutieli, Daniel},
	coden = {BIOKAX},
	fjournal = {Biometrika},
	issn = {0006-3444},
	journal = {Biometrika},
	mrclass = {62F03},
	mrnumber = {MR2261438},
	number = {3},
	pages = {491--507},
	title = {Adaptive linear step-up procedures that control the false discovery rate},
	volume = {93},
	year = {2006}}

@article{BR2009,
	author = {Blanchard, Gilles and Roquain, Etienne},
	fjournal = {Journal of Machine Learning Research (JMLR)},
	journal = {J. Mach. Learn. Res.},
	pages = {2837--2871},
	title = {Adaptive false discovery rate control under independence and dependence},
	volume = {10},
	year = {2009}}

@article{BY2001,
	author = {Benjamini, Yoav and Yekutieli, Daniel},
	coden = {ASTSC7},
	fjournal = {The Annals of Statistics},
	issn = {0090-5364},
	journal = {Ann. Statist.},
	mrclass = {62J15 (47N30 62H20)},
	mrnumber = {MR1869245 (2002i:62135)},
	mrreviewer = {S. Panchapakesan},
	number = {4},
	pages = {1165--1188},
	title = {The control of the false discovery rate in multiple testing under dependency},
	volume = {29},
	year = {2001}}

@article{Chi2007,
	author = {Chi, Zhiyi},
	coden = {ASTSC7},
	doi = {10.1214/009053607000000037},
	fjournal = {The Annals of Statistics},
	issn = {0090-5364},
	journal = {Ann. Statist.},
	mrclass = {62G10 (60G35 62H15)},
	mrnumber = {MR2351091},
	number = {4},
	pages = {1409--1431},
	title = {On the performance of {FDR} control: constraints and a partial solution},
	url = {http://dx.doi.org/10.1214/009053607000000037},
	volume = {35},
	year = {2007}}

@article{ETST2001,
	author = {Efron, Bradley and Tibshirani, Robert and Storey, John D. and Tusher, Virginia},
	coden = {JSTNAL},
	fjournal = {Journal of the American Statistical Association},
	issn = {0162-1459},
	journal = {J. Amer. Statist. Assoc.},
	mrclass = {62C12 (62P10)},
	mrnumber = {MR1946571},
	number = {456},
	pages = {1151--1160},
	title = {Empirical {B}ayes analysis of a microarray experiment},
	volume = {96},
	year = {2001}}

@article{FDR2009,
	author = {Finner, H. and Dickhaus, T. and Roters, M.},
	fjournal = {The Annals of Statistics},
	journal = {Ann. Statist.},
	number = {2},
	pages = {596--618},
	title = {On the false discovery rate and an asymptotically optimal rejection curve},
	volume = {37},
	year = {2009}}

@article{GW2002,
	author = {Genovese, Christopher and Wasserman, Larry},
	fjournal = {Journal of the Royal Statistical Society. Series B. Statistical Methodology},
	issn = {1369-7412},
	journal = {J. R. Stat. Soc. Ser. B Stat. Methodol.},
	mrclass = {62F03},
	mrnumber = {MR1924303 (2003h:62027)},
	number = {3},
	pages = {499--517},
	title = {Operating characteristics and extensions of the false discovery rate procedure},
	volume = {64},
	year = {2002}}

@article{GW2004,
	author = {Genovese, Christopher and Wasserman, Larry},
	coden = {ASTSC7},
	fjournal = {The Annals of Statistics},
	issn = {0090-5364},
	journal = {Ann. Statist.},
	mrclass = {62H15 (62G10)},
	mrnumber = {MR2065197 (2005k:62149)},
	mrreviewer = {M. Hu{\v{s}}kov{\'a}},
	number = {3},
	pages = {1035--1061},
	title = {A stochastic process approach to false discovery control},
	volume = {32},
	year = {2004}}

@article{massart1990tight,
  title={The tight constant in the Dvoretzky-Kiefer-Wolfowitz inequality},
  author={Massart, Pascal},
  journal={The annals of Probability},
  pages={1269--1283},
  year={1990},
  publisher={JSTOR}
}

@article{neuvial2013asymptotic,
  title={Asymptotic results on adaptive false discovery rate controlling procedures based on kernel estimators},
  author={Neuvial, Pierre},
  journal={The Journal of Machine Learning Research},
  volume={14},
  number={1},
  pages={1423--1459},
  year={2013},
  publisher={JMLR. org}
}

@article{Neu2013,
	author = {Neuvial, Pierre},
	fjournal = {Journal of Machine Learning Research (JMLR)},
	issn = {1532-4435},
	journal = {J. Mach. Learn. Res.},
	mrclass = {62J15 (62G07)},
	mrnumber = {3081930},
	pages = {1423--1459},
	title = {Asymptotic results on adaptive false discovery rate controlling procedures based on kernel estimators},
	volume = {14},
	year = {2013}}

@article{RW2009,
	author = {Roquain, E. and van de Wiel, M.},
	fjournal = {Electronic Journal of Statistics},
	journal = {Electron. J. Stat.},
	pages = {678--711},
	title = {Optimal weighting for false discovery rate control},
	volume = {3},
	year = {2009}}

@article{Sar2008,
	author = {Sarkar, S.~K.},
	journal = {Sankhya, Ser. A},
	numner = {2},
	pages = {135--168},
	title = {On Methods Controlling the False Discovery Rate},
	volume = {70},
	year = {2008}}

@article{SC2009,
	author = {Sun, Wenguang and Cai, T. Tony},
	doi = {10.1111/j.1467-9868.2008.00694.x},
	fjournal = {Journal of the Royal Statistical Society. Series B. Statistical Methodology},
	issn = {1369-7412},
	journal = {J. R. Stat. Soc. Ser. B Stat. Methodol.},
	mrclass = {Database Expansion Item},
	mrnumber = {2649603},
	number = {2},
	pages = {393--424},
	title = {Large-scale multiple testing under dependence},
	url = {http://dx.doi.org/10.1111/j.1467-9868.2008.00694.x},
	volume = {71},
	year = {2009}}

@article{Storey2002,
	author = {Storey, John D.},
	fjournal = {Journal of the Royal Statistical Society. Series B. Statistical Methodology},
	issn = {1369-7412},
	journal = {J. R. Stat. Soc. Ser. B Stat. Methodol.},
	mrclass = {62F03},
	mrnumber = {MR1924302 (2003h:62029)},
	number = {3},
	pages = {479--498},
	title = {A direct approach to false discovery rates},
	volume = {64},
	year = {2002}}

@article{gazin2023transductive,
  title={Transductive conformal inference with adaptive scores},
  author={Gazin, Ulysse and Blanchard, Gilles and Roquain, Etienne},
  journal={arXiv preprint arXiv:2310.18108},
  year={2023}
}

\appendix

\section{Proofs}

\subsection{Proof of Theorem~\ref{th-gen}}\label{proof:th-gen}.

Let us denote $\hat{k}$ in \eqref{equ-lhat} by $\hat{k}=\hat{k}(\mbf{p},\hat{\pi}_0(\mbf{p}))$ where
$$
    \hat{k}(\mbf{p},q)=\max\{k\in [0,m]\::\: p_{(k)}\leq \kappa \wedge (\alpha k/ (m q))\}, \:\:\:q>0,
$$
so that the two sources of dependence with regard to $\mbf{p}$ are separated in the notation $\hat{k}=\hat{k}(\mbf{p},\hat{\pi}_0(\mbf{p}))$.

Let us first consider the independent case, for which $\tilde{\mbf{p}}^{0,i}$ is independent of $p_i$ for all $i\in\cH_0$. We have
\begin{align}
\FDR(\ABH_\alpha(\hat{\pi}_0,\kappa))&=\sum_{i\in \cH_0}\E\Bigg[\frac{\ind{p_i\leq \tau_{\hat{k}}}}{\hat{k}\vee 1}\Bigg]\nonumber\\    &\leq \sum_{i\in \cH_0}\E\Bigg[\frac{\ind{p_i\leq \kappa \wedge (\alpha [1\vee \hat{k}(\mbf{p},\hat{\pi}_0(\mbf{p}))]/ (m\hat{\pi}_0(\mbf{p})))}}{1\vee \hat{k}(\mbf{p},\hat{\pi}_0(\mbf{p}))}\Bigg]\label{equintermfirstproof}
\end{align}
Now, if $\hat{\pi}_0$ satisfies Assumption~\ref{ass:monotone} and $\kappa=1$ (case (i)), the last display is at most
\begin{align*} 
(\alpha/m)\sum_{i\in \cH_0}\E\Bigg[\E\Bigg[\frac{\ind{p_i\leq g(p_i;\tilde{\mbf{p}}^{0,i})}}{g(p_i;\tilde{\mbf{p}}^{0,i})}\:\Big|\: \tilde{\mbf{p}}^{0,i}\Bigg]\Bigg],
\end{align*}
where we let $g(p_i;\tilde{\mbf{p}}^{0,i})=\alpha [1\vee \hat{k}(\mbf{p},\hat{\pi}_0(\mbf{p}))]/ m$ which is a nonincreasing function of $p_i$ by Assumption~\ref{ass:monotone}. Since $p_i$ is independent of $\tilde{\mbf{p}}^{0,i}$, Lemma~\ref{lem:BR2009} gives that
\begin{align*}
\FDR(\ABH_\alpha(\hat{\pi}_0,\kappa))&\leq (\alpha/m)\sum_{i\in \cH_0}\E\Bigg[\frac{g(p_i;\tilde{\mbf{p}}^{0,i})/\hat{\pi}_0(\tilde{\mbf{p}}^{0,i})}{g(p_i;\tilde{\mbf{p}}^{0,i})}\Bigg] = (\alpha/m)\sum_{i\in \cH_0}\E\Bigg[\frac{1}{\hat{\pi}_0(\tilde{\mbf{p}}^{0,i})}\Bigg],
\end{align*}
which conclude case (i). If $\hat{\pi}_0$ satisfies  Assumption~\ref{ass:capping} with a capping $\kappa\in (0,1)$ (case (ii)), we have $\hat{\pi}_0(\mbf{p})=\hat{\pi}_0(\tilde{\mbf{p}}^{0,i})$ when $p_i\leq \kappa$. Also, by applying Lemma~\ref{lemmaSU} with $\tau_k= \kappa \wedge (\alpha k/ (m\hat{\pi}_0(\mbf{p})))$, $k\in [m]$, \eqref{equintermfirstproof} can be written as
\begin{align*}
 &\sum_{i\in \cH_0}\E\Bigg[\E\Bigg[\frac{\ind{p_i\leq \kappa \wedge (\alpha [1\vee \hat{k}(\tilde{\mbf{p}}^{0,i},\hat{\pi}_0(\mbf{p}))]/ (m\hat{\pi}_0(\mbf{p})))}}{1\vee \hat{k}(\tilde{\mbf{p}}^{0,i},\hat{\pi}_0(\mbf{p}))}\:\Big|\: \tilde{\mbf{p}}^{0,i}\Bigg]\Bigg]\\
 &=
 \sum_{i\in \cH_0}\E\Bigg[\frac{\ind{p_i\leq \kappa \wedge (\alpha [1\vee \hat{k}(\tilde{\mbf{p}}^{0,i},\hat{\pi}_0(\tilde{\mbf{p}}^{0,i}))]/ (m\hat{\pi}_0(\tilde{\mbf{p}}^{0,i})))}}{1\vee \hat{k}(\tilde{\mbf{p}}^{0,i},\hat{\pi}_0(\tilde{\mbf{p}}^{0,i}))}\:\Big|\: \tilde{\mbf{p}}^{0,i}\Bigg]\Bigg]\\
 &= \sum_{i\in \cH_0}\E\Bigg[\E\Bigg[ \frac{\kappa \wedge (\alpha [1\vee \hat{k}(\tilde{\mbf{p}}^{0,i},\hat{\pi}_0(\tilde{\mbf{p}}^{0,i}))]/ (m\hat{\pi}_0(\tilde{\mbf{p}}^{0,i})))}{1\vee \hat{k}(\tilde{\mbf{p}}^{0,i},\hat{\pi}_0(\tilde{\mbf{p}}^{0,i}))}\Bigg]\leq  (\alpha/m)\sum_{i\in \cH_0}\E\Bigg[\frac{1}{\hat{\pi}_0(\tilde{\mbf{p}}^{0,i})}\Bigg],
\end{align*}
where we used again the independence between $p_i$ and $\tilde{\mbf{p}}^{0,i}$. This finishes the proof for the independent case.

Let us now turn to the conformal case. Let for any $i\in \cH_0$ 
   \begin{equation}
       \label{equ:wi}
W_i = ( \{S_k,k\in [n]\cup\{n+i\}\}, (S_{n+j}, j \in [m]\backslash\{i\})).
   \end{equation}
   By the no-ties assumption and exchangeability of $(S_k,k\in [n]\cup\{n+i\})$ given $(S_{n+j}, j \in [m]\backslash\{i\})$, Lemma~\ref{lem:conformal} provides that $p_i$ is independent of $W_i$, $p_i$ is marginally super-uniform and $(p_j,j\in [m]\backslash\{i\})=(\Psi_{ij}(p_i;W_i))_{j\in [m]\backslash\{i\}}$ with $\Psi_{ij}(p_i;W_i)$ being nondecreasing with regard to $p_i$.  
In addition,  we note that $\tilde{\mbf{p}}^{0,i}$ defined in \eqref{equ:configconf} is such that, for $j\in [m]\backslash\{ i\}$,
   \begin{align*}
\tilde{p}^{0,i}_j &=
      \frac{ \sum_{k\in [n]\cup\{n+i\}}  \ind{S_k \ge S_{n+j}}}{n+1} \leq p_j,
   \end{align*}
   with equality if $S_{n+i}\geq S_{n+j}$. 
   In particular, for any $j\in [m]$ with $p_j>p_i$, we must have $S_{n+i}\geq  S_{n+j}$ and thus  $\tilde{p}^{0,i}_j=p_j$. The requirement of Lemma~\ref{lemmaSU} is thus satisfied with $\mbf{p}'=\tilde{\mbf{p}}^{0,i}$ and we have 
   \begin{equation}
       \label{eqveryusefulindeed}
       p_i\leq \hat{k}(\mbf{p},\hat{\pi}_0(\mbf{p})) \mbox{ entails } \hat{k}(\mbf{p},\hat{\pi}_0(\mbf{p}))=\hat{k}(\tilde{\mbf{p}}^{0,i},\hat{\pi}_0(\mbf{p})).
   \end{equation}
   Note also that $\tilde{\mbf{p}}^{0,i}$ is measurable only with regard to $W_i$.
   
   Now, we consider the case where $\hat{\pi}_0$ satisfies Assumption~\ref{ass:monotone} and $\kappa=1$ (case (i)). 
Coming back to \eqref{equintermfirstproof}, and letting $g(p_i;W_i)=\alpha [1\vee \hat{k}(\mbf{p},\hat{\pi}_0(\mbf{p}))]/ m$ which is a nonincreasing function of $p_i$ by Assumption~\ref{ass:monotone} and because each function $\Psi_{ij}(p_i;W_i)$ is nondecreasing in $p_i$, \eqref{equintermfirstproof} is at most
\begin{align*} 
(\alpha/m)\sum_{i\in \cH_0}\E\Bigg[\E\Bigg[\frac{\ind{p_i\leq g(p_i;W_i)/\hat{\pi}_0(\mbf{p})}}{g(p_i;W_i)}\:\Big|\: W_i\Bigg]\Bigg]\leq (\alpha/m)\sum_{i\in \cH_0}\E\Bigg[\E\Bigg[\frac{\ind{p_i\leq g(p_i;W_i)/\hat{\pi}_0(\tilde{\mbf{p}}^{0,i})}}{g(p_i;W_i)}\:\Big|\: W_i\Bigg]\Bigg]
\end{align*}
by using again Assumption~\ref{ass:monotone}. Now, since $p_i$ is independent of $W_i$, 
Lemma~\ref{lem:BR2009} gives that
\begin{align*}
\FDR(\ABH_\alpha(\hat{\pi}_0,\kappa))&\leq (\alpha/m)\sum_{i\in \cH_0}\E\Bigg[\frac{g(p_i;W_i)/\hat{\pi}_0(\tilde{\mbf{p}}^{0,i})}{g(p_i;W_i)}\Bigg] = (\alpha/m)\sum_{i\in \cH_0}\E\Bigg[\frac{1}{\hat{\pi}_0(\tilde{\mbf{p}}^{0,i})}\Bigg],
\end{align*}
which conclude case (i).
Now, if $\hat{\pi}_0$ satisfies  Assumption~\ref{ass:capping} with a capping $\kappa\in (0,1)$ (case (ii)), we have $\hat{\pi}_0(\mbf{p})=\hat{\pi}_0(\tilde{\mbf{p}}^{0,i})$ when $p_i\leq \kappa$ and thus we follow a reasoning similar to the independent case, with \eqref{equintermfirstproof} and \eqref{eqveryusefulindeed} leading to
\begin{align*}
 \FDR(\ABH_\alpha(\hat{\pi}_0,\kappa))&\leq \sum_{i\in \cH_0}\E\Bigg[\E\Bigg[\frac{\ind{p_i\leq \kappa \wedge (\alpha [1\vee \hat{k}(\tilde{\mbf{p}}^{0,i},\hat{\pi}_0(\mbf{p}))]/ (m\hat{\pi}_0(\mbf{p})))}}{1\vee \hat{k}(\tilde{\mbf{p}}^{0,i},\hat{\pi}_0(\mbf{p}))}\:\Big|\: W_i\Bigg]\Bigg]\\
 &=
 \sum_{i\in \cH_0}\E\Bigg[\E\Bigg[\frac{\ind{p_i\leq \kappa \wedge (\alpha [1\vee \hat{k}(\tilde{\mbf{p}}^{0,i},\hat{\pi}_0(\tilde{\mbf{p}}^{0,i}))]/ (m\hat{\pi}_0(\tilde{\mbf{p}}^{0,i})))}}{1\vee \hat{k}(\tilde{\mbf{p}}^{0,i},\hat{\pi}_0(\tilde{\mbf{p}}^{0,i}))}\:\Big|\: W_i\Bigg]\Bigg]\\
 &= \sum_{i\in \cH_0}\E\Bigg[ \frac{\kappa \wedge (\alpha [1\vee \hat{k}(\tilde{\mbf{p}}^{0,i},\hat{\pi}_0(\tilde{\mbf{p}}^{0,i}))]/ (m\hat{\pi}_0(\tilde{\mbf{p}}^{0,i})))}{1\vee \hat{k}(\tilde{\mbf{p}}^{0,i},\hat{\pi}_0(\tilde{\mbf{p}}^{0,i}))}\Bigg]\leq  (\alpha/m)\sum_{i\in \cH_0}\E\Bigg[\frac{1}{\hat{\pi}_0(\tilde{\mbf{p}}^{0,i})}\Bigg].
\end{align*}

\subsection{Proof of Corollary~\ref{corcontrolStorey}}\label{secproofcorcontrolStorey}

 By Theorem~\ref{th-gen} (i), we only have to show that for $i\in \cH_0$,
$
    \E[1/\hat{\pi}_{0,\lambda}(\tilde{\mbf{p}}^{0,i})] \leq m/m_0
$
in both settings. In the independent case, the latter holds by stochastic domination:
\begin{align*}
     \E[1/\hat{\pi}_{0,\lambda}(\tilde{\mbf{p}}^{0,i})] \leq m (1-\lambda) \E\bigg[\frac{1}{1+\sum_{j\in \cH_0 \backslash\{i\}} \ind{p_j>\lambda}}\bigg] = m (1-\lambda) \E\bigg[\frac{1}{1+\mathcal{B}(m_0-1,1-\lambda)}\bigg] \leq m/m_0,
\end{align*}
by using a classical relation on the binomial random variable (see Lemma~1 in \cite{benjamini2006adaptive}). 
In the conformal setting, by the no ties and exchangeability assumptions, it can be assumed that the scores $(S_j, j\in [n])\cup (S_{n+i},i\in \cH_0)$ are iid uniform when computing quantities involving the distribution of the variables $\tilde{p}^{0,i}_j,j\in \cH_0\backslash\{i\}$. Hence, if we assume so, they are iid conditionally on  $(S_1,\ldots,S_n,S_{n+i})$, and it can be checked that $\P(\tilde{p}^{0,i}_j\geq \lambda\:|\:(S_1,\ldots,S_n,S_{n+i})) = \P(S_{(\lceil \lambda(n+1)\rceil)}\geq S_{n+j}\:|\:(S_1,\ldots,S_n,S_{n+i})) = S_{(\lceil \lambda(n+1)\rceil)}$, where $S_{(1)}\geq \dots\geq S_{(n+1)}$ are the ordered elements of $(S_1,\ldots,S_n,S_{n+i})$. Hence, by using a reasoning similar as above, and that $S_{(\lceil \lambda(n+1)\rceil)}\sim\beta(n+2-\lceil \lambda(n+1)\rceil,\lceil \lambda(n+1)\rceil)$, we have 
\begin{align*}
     \E[1/\hat{\pi}_{0,\lambda}(\tilde{\mbf{p}}^{0,i})] &\leq  m (1-\lambda) \E\bigg[\E\bigg[\frac{1}{1+\mathcal{B}(m_0-1,S_{(\lceil \lambda(n+1)\rceil)})}\bigg] \bigg| (S_1,\ldots,S_n,S_{n+i}) \bigg] \\
     &\leq (1-\lambda) m/m_0 \E\bigg[\frac{1}{\beta(n+2-\lceil \lambda(n+1)\rceil,\lceil \lambda(n+1)\rceil)}\bigg]\leq  \frac{(n+1)(1-\lambda) m/m_0}{n+1-\lceil \lambda(n+1)\rceil} = m/m_0,
\end{align*}
because $(n+1)\lambda\in [n+1]$ and $\E[1/\beta(a,b)]=(a+b-1)/(a-1)$.

\subsection{Proof of Corollary~\ref{corcontrolIStorey}}\label{secproofcorcontrolIStorey}
     By Theorem~\ref{th-gen} (ii), we only have to show that for $i\in \cH_0$,
$
    \E[1/\hat{\pi}_{0,[\lambda,\mu]}(\tilde{\mbf{p}}^{0,i})] \leq m/m_0
$
in both settings. In the independent case, similarly to the previous proof
\begin{align*}
     \E[1/\hat{\pi}_{0,[\lambda,\mu]}(\tilde{\mbf{p}}^{0,i})] =  \E\bigg[\frac{m (\mu-\lambda)}{1+\sum_{j\in \cH_0 \backslash\{i\}} \ind{p_j\in[\lambda, \mu]}}\bigg] = \E\bigg[\frac{m (\mu-\lambda)}{1+\mathcal{B}(m_0-1,\mu-\lambda)}\bigg] \leq m/m_0.
\end{align*}
In the conformal setting, assume that the scores $(S_j, j\in [n])\cup (S_{n+i},i\in \cH_0)$ are iid uniform without loss of generality. The variables $\tilde{p}^{0,i}_j,j\in \cH_0\backslash\{i\}$ are iid conditionally on  $(S_1,\ldots,S_n,S_{n+i})$, and we have similarly to above $\P(\tilde{p}^{0,i}_j\in [\lambda,\mu] \:|\:(S_1,\ldots,S_n,S_{n+i})) = 
\P((n+1)\tilde{p}^{0,i}_j\in [(n+1)\lambda,(n+1)\mu+1) \:|\:(S_1,\ldots,S_n,S_{n+i})) 
= S_{( (n+1)\lambda}-S_{( (n+1)\mu+1)}$, where $S_{(1)}\geq \dots\geq S_{(n+1)}$ are the ordered elements of $(S_1,\ldots,S_n,S_{n+i})$. Since $S_{((n+1)\lambda)}-S_{( (n+1)\mu+1)}\sim\beta( (n+1)(\mu-\lambda)+1, n+1- (n+1)(\mu-\lambda))$, we obtain 
\begin{align*}
     \E[1/\hat{\pi}_{0,[\lambda,\mu]}(\tilde{\mbf{p}}^{0,i})] &\leq   \E\bigg[\E\bigg[\frac{m (\mu-\lambda)}{1+\mathcal{B}(m_0-1,S_{((n+1)\lambda)}-S_{( (n+1)\mu+1)})}\bigg] \bigg| (S_1,\ldots,S_n,S_{n+i}) \bigg] \\
     &\leq  \E\bigg[\frac{(\mu-\lambda) m/m_0}{\beta( (n+1)(\mu-\lambda)+1, n+1- (n+1)(\mu-\lambda))}\bigg]\leq  \frac{(n+1)(\mu-\lambda) m/m_0}{(n+1)(\mu-\lambda)} = m/m_0,
\end{align*}
because $\E[1/\beta(a,b)]=(a+b-1)/(a-1)$.

\subsection{Proof of Theorem~\ref{thcontrolMS}}\label{sec:proofthcontrolMS}

First observe that $\hat{\pi}^{\MS}_{0,\epsilon,\underline{\pi}_0}$ in \eqref{equ:pi0hatMS} satisfies Assumption~\ref{ass:monotone}, so that we can apply Theorem~\ref{th-gen} (i). Hence, we have 
\begin{align*}
    \FDR(\ABH_\alpha(\hat{\pi}^{\MS}_{0,\epsilon,\underline{\pi}_0}))&\leq (\alpha/m)\sum_{i\in \cH_0} \E[1/\hat{\pi}^{\MS}_{0,\epsilon,\underline{\pi}_0}(\tilde{\mbf{p}}^{0,i})].
\end{align*}
Since $\hat{\pi}^{\MS}_{0,\epsilon,\underline{\pi}_0}(\tilde{\mbf{p}}^{0,i})\geq \underline{\pi}_0$, this directly leads to the bound $\alpha (m_0/m)/\underline{\pi}_0\leq \alpha$ when $\pi_0\leq \underline{\pi}_0$. Assume now $\pi_0\geq \underline{\pi}_0$. In that case, the above display is at most 
\begin{align*}
 & \frac{\alpha/m}{C(m,\epsilon,\underline{\pi}_0)}\sum_{i\in \cH_0} \E  \bigg(  \sup_{\lambda \in [0,1-\epsilon]}
      \frac{m(1-\lambda)}{1 \vee \sum_{j\in [m]} \mathbf{1}\{\tilde{{p}}^{0,i}_j \geq \lambda\}} \bigg)\\
&\leq \frac{\alpha/m_0}{C(m,\epsilon,\underline{\pi}_0)}\sum_{i\in \cH_0} \E  \bigg(  \sup_{\lambda \in [0,1-\epsilon]}
      \frac{m_0(1-\lambda)}{1 \vee \sum_{j\in \cH_0} \mathbf{1}\{\tilde{{p}}^{0,i}_j \geq \lambda\}} \bigg)\\
      &\leq \frac{\alpha}{C(m,\epsilon,\underline{\pi}_0)} \E_{\mbf{q}\sim Q_0}  \bigg(  \sup_{\lambda \in [0,1-\epsilon]}
      \frac{m_0(1-\lambda)}{1 \vee \sum_{j\in [m_0]} \mathbf{1}\{q_j \geq \lambda\}} \bigg),
\end{align*}
where the last step is obtained by stochastic domination in the independence setting (because all $p$-values under the null are super-uniform), and by using that the distribution of $(\tilde{{p}}^{0,i}_j)_{j\in \cH_0\backslash\{i\}}$ is the same as if $i=1$ and $\cH_0=[m_0]$ (by exchangeability of the scores $S_{j}, j\in [n]$, $S_{n+i}, i\in \cH_0$ in the conformal setting). Taking the maximum over all possible $m_0\in [\underline{\pi_0}m, m]$ concludes the proof by definition \eqref{equ:gMS} of $C(m,\epsilon,\underline{\pi}_0)$.

\subsection{Proof of Theorem~\ref{thcontrolIMS}}\label{sec:proofthcontrolIMS}
First observe that $\hat{\pi}^{\IMS}_{0,\epsilon,\underline{\pi}_0,\kappa}$ in \eqref{equ:pi0hatMS} satisfies Assumption~\ref{ass:capping}, so that we can apply Theorem~\ref{th-gen} (ii).
Hence, we have 
\begin{align*}
    \FDR(\ABH_\alpha(\hat{\pi}^{\IMS}_{0,\epsilon,\underline{\pi}_0,\kappa},\kappa))&\leq (\alpha/m)\sum_{i\in \cH_0} \E[1/\hat{\pi}^{\IMS}_{0,\epsilon,\underline{\pi}_0,\kappa}(\tilde{\mbf{p}}^{0,i})].
\end{align*}
We conclude as in Section~\ref{sec:proofthcontrolMS} (by using that the null $p$-value are marginally uniform in the independence setting).

\subsection{Proof of Theorem~\ref{thcontrolIMSkappahat}}\label{proof:thcontrolIMSkappahat}

Denote $\hat{\kappa}=\hat{\kappa}(\mbf{p})$ and $\hat{\pi}^{\IMS}_{0,\epsilon,\underline{\pi}_0,\hat{\kappa}}=\hat{\pi}^{\IMS}_{0,\epsilon,\underline{\pi}_0,\hat{\kappa}(\mbf{p})}(\mbf{p})$ to underline the dependence with regard to the $p$-value family $\mbf{p}$. We have both in the independence and conformal setting,
\begin{align}
\FDR(\ABH_\alpha(\hat{\pi}^{\IMS}_{0,\epsilon,\underline{\pi}_0,\hat{\kappa}},\hat{\kappa}))
&=\sum_{i\in \cH_0}\E\Bigg[\frac{\ind{p_i\leq \hat{\kappa}(\mbf{p})}}{ 1\vee \sum_{j\in [m]}\ind{p_j\leq \hat{\kappa}(\mbf{p})}}\Bigg]\nonumber\\   
&\leq \sum_{i\in \cH_0}\E\Bigg[\frac{\ind{p_i\leq \hat{\kappa}(\mbf{p})}}{ 1\vee [m \hat{\kappa} (\mbf{p}) \hat{\pi}^{\IMS}_{0,\epsilon,\underline{\pi}_0,\hat{\kappa}(\mbf{p})}(\mbf{p})/\alpha]}\Bigg]\nonumber\\   
&\leq (\alpha/m) \sum_{i\in \cH_0}\E\Bigg[ \frac{1}{\hat{\pi}^{\IMS}_{0,\epsilon,\underline{\pi}_0,\hat{\kappa}(\mbf{p})}(\mbf{p})} \frac{\ind{p_i\leq \hat{\kappa}(\mbf{p})}}{ (\alpha/(m\hat{\pi}^{\IMS}_{0,\epsilon,\underline{\pi}_0,\hat{\kappa}(\mbf{p})}(\mbf{p}) ))\vee \hat{\kappa}(\mbf{p})}\Bigg]\nonumber,
\end{align}
where we used successively that the threshold of $\ABH_\alpha(\hat{\pi}^{\IMS}_{0,\epsilon,\underline{\pi}_0,\hat{\kappa}},\hat{\kappa})$ is $\hat{\kappa})$ (Lemma~\ref{lem:kappahat} (iii)) and that $\sum_{j\in [m]}\ind{p_j\leq \hat{\kappa}(\mbf{p})}\geq m \hat{\kappa} (\mbf{p}) \hat{\pi}^{\IMS}_{0,\epsilon,\underline{\pi}_0,\hat{\kappa}(\mbf{p})}/\alpha$ (Lemma~\ref{lem:kappahat} (ii)). 
Now we use Lemma~\ref{lem:kappahat2} to write (the requirement of the lemma is satisfied with $\mbf{p}'=\tilde{\mbf{p}}^{0,i}$ as in the proof of Theorem~\ref{th-gen})
\begin{align}
\FDR(\ABH_\alpha(\hat{\pi}^{\IMS}_{0,\epsilon,\underline{\pi}_0,\hat{\kappa}},\hat{\kappa}))
&\leq (\alpha/m) \sum_{i\in \cH_0}\E\Bigg[ \frac{1}{\hat{\pi}^{\IMS}_{0,\epsilon,\underline{\pi}_0,\hat{\kappa}(\tilde{\mbf{p}}^{0,i})}(\tilde{\mbf{p}}^{0,i})} \frac{\ind{p_i\leq \hat{\kappa}(\tilde{\mbf{p}}^{0,i})}}{ (\alpha/(m\hat{\pi}^{\IMS}_{0,\epsilon,\underline{\pi}_0,\hat{\kappa}(\tilde{\mbf{p}}^{0,i})}(\tilde{\mbf{p}}^{0,i}) ))\vee \hat{\kappa}(\tilde{\mbf{p}}^{0,i})}\Bigg]\nonumber\\
&\leq (\alpha/m) \sum_{i\in \cH_0}\E\Bigg[ \frac{1}{\hat{\pi}^{\IMS}_{0,\epsilon,\underline{\pi}_0,\hat{\kappa}(\tilde{\mbf{p}}^{0,i})}(\tilde{\mbf{p}}^{0,i})} \frac{\ind{p_i\leq (\alpha/(m\hat{\pi}^{\IMS}_{0,\epsilon,\underline{\pi}_0,\hat{\kappa}(\tilde{\mbf{p}}^{0,i})}(\tilde{\mbf{p}}^{0,i}) ))\vee\hat{\kappa}(\tilde{\mbf{p}}^{0,i})}}{ (\alpha/(m\hat{\pi}^{\IMS}_{0,\epsilon,\underline{\pi}_0,\hat{\kappa}(\tilde{\mbf{p}}^{0,i})}(\tilde{\mbf{p}}^{0,i}) ))\vee \hat{\kappa}(\tilde{\mbf{p}}^{0,i})}\Bigg]\nonumber
\end{align}
Now using that $\tilde{\mbf{p}}^{0,i}$ is independent of $p_i$ (in both settings), we obtain
\begin{align*}
\FDR(\ABH_\alpha(\hat{\pi}^{\IMS}_{0,\epsilon,\underline{\pi}_0,\hat{\kappa}},\hat{\kappa}))
&\leq (\alpha/m) \sum_{i\in \cH_0}\E\Bigg[ \frac{1}{\hat{\pi}^{\IMS}_{0,\epsilon,\underline{\pi}_0,\hat{\kappa}(\tilde{\mbf{p}}^{0,i})}(\tilde{\mbf{p}}^{0,i})} \Bigg]\\
&\leq \bigg(\frac{\alpha m_0}{ m\underline{\pi}_0} \bigg)\wedge  \bigg[\frac{\alpha/m_0}{D(m,\epsilon,\underline{\pi}_0)}\sum_{i\in \cH_0} \E  \bigg(  \sup_{I\subset [0,1] : |I|\geq \epsilon}
      \frac{m_0 |I|}{1 \vee \sum_{j\in \cH_0} \mathbf{1}\{\tilde{{p}}^{0,i}_j \in I\}} \bigg)\bigg]\\
&\leq \bigg(\frac{\alpha m_0}{ m\underline{\pi}_0} \bigg)\wedge  \bigg[\frac{\alpha d(m_0,\epsilon)}{D(m,\epsilon,\underline{\pi}_0)}\bigg] \leq \alpha ,
\end{align*}
because $(\tilde{p}^{0,i}_j)_{j\in \cH_0}\sim Q_{0,m_0}$ by exchangeability and by the definitions in \eqref{equ:gIMS}.

\subsection{Proof of Proposition~\ref{prop:Neps}}\label{proofprop:Neps}

Let us define
\begin{equation}
    \label{equNeps}
  \mathcal{M}(\epsilon) := 
  \{s\geq 2\::\:  (s+1)^4 (1-\epsilon)^{s-1} \leq \epsilon \}. 
\end{equation}

\begin{proposition}\label{prop:Nepsproof}
    The following holds:
    \begin{itemize}
        \item in the independent setting, $c(s+1,\epsilon)\leq c(s,\epsilon)$ provided that $s\in \mathcal{M}(\epsilon)$ and $d(s+1,\epsilon)\leq d(s,\epsilon)$ if $s\geq \mathcal{M}(\epsilon')\cap [2/\epsilon,+\infty)$ for $\epsilon'=e^{-\epsilon^2/8}$. 
        \item  in  the conformal setting, both $c(s+1,\epsilon)\leq c(s,\epsilon)$  and $d(s+1,\epsilon,\kappa)\leq d(s,\epsilon,\kappa)$ hold if 
        $$\epsilon\geq \sqrt{\frac{16\log(3(s+1)^4) + 8 \log((n+1)\wedge s)}{(n+1)\wedge s}}
        .$$
    \end{itemize} 
\end{proposition}

By concavity of $\log$, we have  $\log(s+1)\leq \log(a)-1+2/a + (s-1)/a$. Choosing $a=16/\log(1/(1-\epsilon))$,  we have  in the independent case that $[N(\epsilon),+\infty)\subset \mathcal{M}(\epsilon)$.
Hence, Proposition~\ref{prop:Nepsproof} implies Proposition~\ref{prop:Neps} and we now prove Proposition~\ref{prop:Nepsproof}.

For all $s\geq 2$, we have 
\begin{align*}
c(s+1,\epsilon)&= \mathbb{E}\left[1\vee \sup_{\lambda \in (0,1-\epsilon]}
      \frac{(s+1)(s-1)(1-\lambda)}{(s-1)\big(1 \vee \sum_{i=2}^{s+1} \mathbf{1}\{q_i \geq  \lambda\}\big)}\right] 
\end{align*}
Note that if $\sum_{i=2}^{s+1} \mathbf{1}\{q_i \geq  1-\epsilon\}\geq 2$, then we have for all $\lambda \in (0,1-\epsilon]$,  $(s-1)\big(1 \vee \sum_{i=2}^{s+1} \mathbf{1}\{q_i \geq  \lambda\}\big)\geq \sum_{i=2}^{s+1} \big(1 \vee \sum_{j=2}^{s+1} \ind{j\neq i}\mathbf{1}\{q_j \geq  \lambda\}\big)$ (because none of the `$\vee 1$' are effective in that case). 
Hence, we have 
\begin{align*}
c(s+1,\epsilon)&\leq  \mathbb{E}\left[1\vee \sup_{\lambda \in (0,1-\epsilon]}
      \frac{(s+1)(s-1)(1-\lambda)}{\sum_{i=2}^{s+1} \big(1 \vee \sum_{j=2}^{s+1} \ind{j\neq i}\mathbf{1}\{q_j \geq  \lambda\}\big)}\right] + (s+1)\: \P\big(\sum_{i=2}^{s+1} \mathbf{1}\{q_i \geq  1-\epsilon\}\leq 1\big)\\
      &=  \mathbb{E}\left[\sup_{\lambda \in (0,1-\epsilon]}
      \frac{(s+1)(s-1)(1-\lambda)}{\sum_{i=2}^{s+1} \big(1 \vee \sum_{j=2}^{s+1} \ind{j\neq i}\mathbf{1}\{q_j \geq  \lambda\}\big)}\right] + \frac{\epsilon}{s(s-1)},
\end{align*}
provided that $s$ satisfies 
\begin{equation}
    \label{equ:conds}
    (s+1)\: \P\big(\sum_{i=2}^{s+1} \mathbf{1}\{q_i \geq  1-\epsilon\}\leq 1\big)\leq \frac{\epsilon}{s(s-1)}.
\end{equation}
Hence, if $s$ satisfies \eqref{equ:conds}, the last display is at most
\begin{align*}
   \mathbb{E}\left[\sup_{\lambda \in (0,1-\epsilon]}
      \frac{s^2(1-\lambda)}{\sum_{i=2}^{s+1} \big(1 \vee \sum_{j=2}^{s+1} \ind{j\neq i}\mathbf{1}\{q_j \geq  \lambda\}\big)}\right] &\leq \mathbb{E}\left[\sup_{\lambda \in (0,1-\epsilon]} \sum_{i=2}^{s+1} 
      \frac{(1-\lambda)}{1 \vee \sum_{j=2}^{s+1} \ind{j\neq i}\mathbf{1}\{q_j \geq  \lambda\}}\right]\\
      &\leq  \mathbb{E}\left[\sup_{\lambda \in (0,1-\epsilon]} 
      \frac{s(1-\lambda)}{1 \vee \sum_{j=3}^{s+1} \mathbf{1}\{q_j \geq  \lambda\}}\right]\\
      &\leq  \mathbb{E}\left[1\vee \sup_{\lambda \in (0,1-\epsilon]} 
      \frac{s(1-\lambda)}{1 \vee \sum_{j=3}^{s+1} \mathbf{1}\{q_j \geq  \lambda\}}\right] = c(s,\epsilon),
\end{align*}
where we used in the first inequality that for any positive sequence $(a_i)_{i\in [k]}$, we have  
$
k^2/\big(\sum_{i=1}^{k} a_i\big)\leq \sum_{i=1}^{k} 1/a_i
$
(by convexity of $x\mapsto 1/x$ on $(0,+\infty)$).
This prove that if $s$ satisfies \eqref{equ:conds}, we have $c(s,\epsilon)\geq c(s+1,\epsilon)$. Let us now analyze  \eqref{equ:conds} in each setting.

In the independent setting, we have $\P\big(\sum_{i=2}^{s+1} \mathbf{1}\{q_i \geq  1-\epsilon\}\leq 1\big)= (1-\epsilon)^{s}+ s \epsilon (1-\epsilon)^{s-1}\leq (s+1)(1-\epsilon)^{s-1}$ hence $s$ satisfies  \eqref{equ:conds} provided that $s\in \mathcal{M}(\epsilon)$ in \eqref{equNeps}. In the conformal case, $s$ satisfies  \eqref{equ:conds} under the more constrained condition \eqref{equ:condsford} below.

Let us now turn to studying the factor $d(\cdot)$. 
For all $s\geq 2$, we have 
\begin{align*}
d(s+1,\epsilon)&= 
\mathbb{E}\left[\sup_{I\subset [0,1]:|I|\geq \epsilon}
      \frac{(s+1) (s-1)|I|}{(s-1) \big(1 \vee \sum_{i=1}^{s+1} \mathbf{1}\{q_i \in I\}\big)}\right].
\end{align*}
Hence, we can follow a similar approach and show that $d(s+1,\epsilon)\leq d(s,\epsilon)$ provided that
$s$ satisfies 
\begin{equation}
    \label{equ:condsford}
    (s+1)\: \P\bigg(\exists I \subset [0,1]:|I|\geq \epsilon,\: \sum_{i=2}^{s+1} \mathbf{1}\{q_i \in I\}\leq 1\bigg)\leq \frac{\epsilon}{s(s-1)}.
\end{equation}
To evaluate \eqref{equ:condsford} in the independent case, we use the DKW inequality \cite{massart1990tight} (see Theorem~\ref{thDKW}) to obtain
\begin{align*}
 &\P\bigg(\exists I \subset [0,1]:|I|\geq \epsilon,\: \sum_{i=2}^{s+1} \mathbf{1}\{q_i \in I\}\leq 1\bigg)\\
 &\leq \P\bigg(\exists I \subset [0,1]:|I|= \epsilon,\: \sum_{i=2}^{s+1} \mathbf{1}\{q_i \in I\}\leq 1\bigg)\\
 &\leq    \P\bigg(\exists t\in [0,1-\epsilon]: s^{-1}\sum_{i=2}^{s+1} \mathbf{1}\{t< q_i \leq t+\epsilon\} - \epsilon\leq 1/s-\epsilon\bigg) \\
 &\leq \P\bigg(\exists t\in [0,1-\epsilon]: s^{-1}\sum_{i=2}^{s+1} \mathbf{1}\{t< q_i \leq t+\epsilon\} - \epsilon\leq -\epsilon/2\bigg)\\
 &\leq e^{-2 s (\epsilon/4)^2} + e^{-2 s (\epsilon/4)^2} = 2 e^{- s \epsilon^2/8},
\end{align*}
by assuming $s\geq 2/\epsilon$. This leads to the conditions $2(s+1)^3 e^{- s \epsilon^2/8} \leq \epsilon $ and $s\geq 2/\epsilon$ and gives the result. For the conformal case, we use instead the conformal DKW inequality of \cite{gazin2023transductive} (see Theorem~\ref{thDKW}). 
A subtle modification however here is that we deal with $q_i=\tilde{p}^{0,1}_i$, $i\in [2,s+1]$ as in  \eqref{equ:configconf} (with $m=s+1$ and $i=1$), and not strictly to conformal $p$-values $p_i,$ $i\in [1,s]$ \eqref{equ-confpvalues} (with $m=s$). Lemma~\ref{lem:changetildeconf} shows that the conformal DKW inequality can be easily adapted to that case and provides
\begin{align*}
& \P\bigg(\exists I \subset [0,1]:|I|\leq \epsilon,\: \sum_{i=2}^{s+1} \mathbf{1}\{q_i \in I\}\leq 1\bigg) \\
&\leq    \P\bigg(\exists t\in [0,1-\epsilon]: s^{-1}\sum_{i=2}^{s+1} \mathbf{1}\{t< q_i \leq t+\epsilon\} - \frac{\lfloor (t+\epsilon)(n+1)\rfloor - \lfloor t(n+1)\rfloor}{n+2}\leq 1/s-\epsilon + 2/(n+2)\bigg) \\ 
&\leq 2\bigg(1+\frac{\sqrt{2\pi} \epsilon \tau_{n+1,s}}{2(n+1+s)^{1/2}}\bigg)e^{- \tau_{n+1,s} \epsilon^2/8},
\end{align*}
by assuming $4/(n+2)+2/s\leq \epsilon$ (so that $1/s-\epsilon + 2/(n+2)\leq -\epsilon/2$), with $\tau_{n+1,s}=(n+1)s/(n+1+s)$.
This leads to the conditions $2(s+1)^3 \bigg(1+\sqrt{\pi/2}  \tau^{1/2}_{n+1,s}\epsilon\bigg)e^{- \tau_{n+1,s} \epsilon^2/8} \leq \epsilon $ and $6 \leq \epsilon ((n+1)\wedge s)$. This holds in particular if  $3(s+1)^3  (\tau^{1/2}_{n+1,s}\epsilon) e^{- \tau_{n+1,s} \epsilon^2/8} \leq \epsilon $ and $6 \leq \epsilon ((n+1)\wedge s)$, and thus if $3(s+1)^4  \tau^{1/2}_{n+1,s} e^{- \tau_{n+1,s} \epsilon^2/8} \leq 1 $ and $6 \leq \epsilon ((n+1)\wedge s)$. This leads to the result because $((n+1)\wedge s)/2 \leq \tau_{n+1,s}\leq (n+1)\wedge s$.

\subsection{Proof of Proposition~\ref{rateCD}}\label{proof:rateCD}

It is a direct consequence of the following results.

\begin{proposition}
    \label{prop:boundingCD}
 The following holds for any $\epsilon\in (0,1)$ and $s\geq 2$,
 \begin{itemize}
     \item in the independent case,
     \begin{align}
         c(s,\epsilon) \vee d(s,\epsilon,\kappa) &\leq 1+\inf_{c>0}\big\{4(c+1) /(\sqrt{s-1} \epsilon)+ 3s e^{-(c/(c+1))^2 (s-1) \epsilon^2 /2}\big\}\label{equ-indepcd}\\
         &\leq 1+ 8(c_s+1) /(\sqrt{s} \epsilon)+1/s,
     \end{align}
     by choosing $c_s>0$ with $c_s/(c_s+1)= \epsilon^{-1}\sqrt{2\log(3s^2)/(s-1)}$.
     \item  in the conformal case, for $\epsilon >4/n$, 
      \begin{align}
         c(s,\epsilon) \vee d(s,\epsilon,\kappa) 
         &\leq 
         1+ \frac{6}{\epsilon n-4}+\inf_{c>0}\big\{\frac{ 33 (c+1)^2}{\epsilon^2 ((n+1)\wedge (s-1))^{1/2} } + 17 s^4  e^{- ((n+1)\wedge (s-1)) \epsilon^2 ( c/(c+1))^2 /4}\big\}.
        \label{equ-indepcd}
     \end{align}
 In particular, when $n>4/\epsilon$, and by choosing $c_s>0$ with $c_s/(c_s+1)= \epsilon^{-1}\sqrt{4\log(17 s^5)/(((n+1)\wedge (s-1))}$, we have 
 \begin{equation}
         \label{equ-indepcdsimple}
         c(s,\epsilon) \vee d(s,\epsilon,\kappa) \leq 1+  \frac{6}{\epsilon n-4}+1/s+
         \frac{ 33 (c_s+1)^2}{\epsilon^2 ((n+1)\wedge(s-1))^{1/2} }.
     \end{equation}
 \end{itemize}
\end{proposition}

\begin{proof}
Take $s\geq 1$ and let us start with the independent setting. We have by definition \eqref{equ:gMS}-\eqref{equ:gIMS},  for $U_1,\dots,U_s$ iid $U(0,1)$, 
\begin{align*}
&c(s+1,\epsilon)\vee d(s+1,\epsilon) -1\\
&\leq \mathbb{E}\left[\sup_{I\subset [0,1]:|I|\geq \epsilon}
      \frac{s |I|}{1 \vee \sum_{i=1}^{s} \ind{U_i \in I}}\right] -1\\
      &= (1+1/s)\int_0^{s+1} \P\bigg(\exists I\subset [0,1]:|I|\geq \epsilon : s |I| /(x+1) > \sum_{i=1}^{s} \ind{U_i \in I} \bigg) dx\\
&=(1+1/s)\int_0^{s+1} \P\bigg(\exists I\subset [0,1]:|I|\geq \epsilon : - |I|  x /(x+1) >s^{-1} \sum_{i=1}^{s} \ind{U_i \in I} - |I| \bigg) dx \\
&= (1+1/s)\int_0^{s+1} \P\bigg(\exists u,v\in [0,1], u<v-\epsilon : 
(\hat{F}_s(v)-v) - (\hat{F}_s(u)-u) < -(v-u) x/(x+1) \bigg) dx,
\end{align*}
by denoting $\hat{F}_s$ for the eCDF of $U_1,\dots,U_s$. By using the DKW inequality (see Theorem~\ref{thDKW}), the last display is at most
\begin{align*}
& (1+1/s)\int_0^{s+1} \P\bigg(\exists u,v\in [0,1], u<v-\epsilon : 
(\hat{F}_s(v)-v) - (\hat{F}_s(u)-u) < -\epsilon x/(x+1) \bigg) dx\\
&\leq (1+1/s)\int_0^{s+1} \bigg(\P\bigg(\exists v\in [0,1] : 
\hat{F}_s(v)-v  < - (\epsilon/2) x/(x+1) \bigg) \\
&\:\:\:\:\:\:\:\:\:\:\:\:\:\:\:+ \P\bigg(\exists u\in [0,1] : 
\hat{F}_s(u)-u  >  (\epsilon/2) x/(x+1) \bigg) \bigg)dx\\
&\leq 2(1+1/s)\int_0^{s+1} e^{-(x/(x+1))^2 s \epsilon^2 /2} dx.
\end{align*}
Now, for any $c>0$, we have
\begin{align*}
\int_0^{s+1} e^{-(x/(x+1))^2 s \epsilon^2 /2} dx&=\int_0^{c} e^{-(x/(x+1))^2 s \epsilon^2 /2} dx+\int_c^{s+1} e^{-(x/(x+1))^2 s \epsilon^2 /2} dx\\
&\leq\int_0^{+\infty} e^{-x^2 s \epsilon^2 /(2(c+1)^2)} dx+(s+1) e^{-(c/(c+1))^2 s \epsilon^2 /2} \\
&= (c+1)\sqrt{\pi/2} /(\sqrt{s} \epsilon)+(s+1) e^{-(c/(c+1))^2 s \epsilon^2 /2},
\end{align*}
because $x/(x+1)\geq c/(c+1)$ for $x\geq c$ and $1+x\leq 1+c$ for $x\leq c$ and 
$\int_0^\infty e^{-x^2/(2\sigma^2)}dx= \sqrt{\pi\sigma^2/2}$.  This leads to the conclusion in the independent case.

The conformal setting can be handled similarly, by using the modified conformal DKW inequality of Lemma~\ref{lem:changetildeconf}: this gives 
\begin{align*}
&c(s+1,\epsilon)\vee d(s+1,\epsilon) -1\\
&= (1+1/s)\int_0^{s+1} \P\bigg(\exists u,v\in [0,1], u<v-\epsilon : 
\bigg(s^{-1} \sum_{j=2}^{s+1} \ind{q_j \leq v}- \frac{\lceil (n+1)v \rceil+1}{n+2}\bigg)\\
&\:\:\:\:\:\:\:\:- \bigg(s^{-1} \sum_{j=2}^{s+1} \ind{q_j \leq u} - \frac{\lceil (n+1)u \rceil+1}{n+2}\bigg)< -\epsilon x/(x+1) + 2/(n+2) \bigg) dx\\
&\leq 6/(\epsilon n-4) + 2(1+1/s)\int_{4/(\epsilon n-4)}^{s+1}\bigg(1+\frac{2\sqrt{2\pi}x \tau_{n+1,s}}{(n+s+1)^{1/2}}\bigg)e^{- \tau_{n+1,s} \epsilon^2 ( x/(x+1))^2 /2} dx\\
&\leq 6/(\epsilon n-4) + 2(1+1/s)\int_{4/(\epsilon n-4)}^{s+1}\bigg(1+2\sqrt{2\pi}x \tau^{1/2}_{n+1,s}\bigg)e^{- \tau_{n+1,s} \epsilon^2 ( x/(x+1))^2 /2} dx
\end{align*}
by using $-\frac{\lceil (n+1)v \rceil+1}{n+2} + \frac{\lceil (n+1)u \rceil+1}{n+2} \leq -(v-u) + 2/(n+2)\leq -\epsilon + 2/(n+2)$ and the fact that $\epsilon x/(x+1)>4/(n+2)$ for $x>4/(\epsilon n-4)$.

Now, similarly to above, we obtain for all $c>0$,
\begin{align*}
&c(s+1,\epsilon)\vee d(s+1,\epsilon) -1-6/(\epsilon n-4) \\
&\leq  2(1+1/s)\int_{0}^{c}\bigg(1+2\sqrt{2\pi}x \tau^{1/2}_{n+1,s}\bigg)e^{- \tau_{n+1,s} \epsilon^2 ( x/(x+1))^2 /2} dx\\
 &+ 2(1+1/s)\int_{c}^{s+1}\bigg(1+2\sqrt{2\pi}x \tau^{1/2}_{n+1,s}\bigg)e^{- \tau_{n+1,s} \epsilon^2 ( x/(x+1))^2 /2} dx\\
&\leq  2(1+1/s)\int_{0}^{c}\bigg(1+2\sqrt{2\pi}x \tau^{1/2}_{n+1,s}\bigg)e^{- \tau_{n+1,s} \epsilon^2 ( x/(c+1))^2 /2} dx\\
&+2(1+1/s)(s+1) \bigg(1+2\sqrt{2\pi}(s+1) \tau^{1/2}_{n+1,s}\bigg) e^{-(c/(c+1))^2 \tau_{n+1,s} \epsilon^2 /2}. 
\end{align*}
 Now we use $\int_0^\infty x e^{-x^2/(2\sigma^2)}dx= \sigma^2$ to obtain the desired bound. 
 \end{proof}

\section{Proof for Section~\ref{sec:optimality}}

\subsection{Optimality of $\pi_0^*$}

\begin{lemma}\label{lemoptim}
    Consider either the independent two-group setting or the conformal two-group setting. For any parameter space $\Theta$ for $(\pi_0,F_0,F_1)$, consider any $\pi_0$-estimator that is FDR conservative over $\Theta$ after $\kappa$-capping in the sense of \eqref{equFDRconservative} (for some $\kappa\in(0,1]$) and concentrating over $\Theta$ around $\bar{\pi}_0(\theta)>0$ in the sense of \eqref{equconcentration}.
For any configuration $\theta=(\pi_0,F_0,F_1)\in \Theta$, consider $\pi^*_0(\theta)$ in \eqref{pi0stargen} (and assume that the maximum is well defined).
 Then we have $\bar{\pi}_0(\theta)\geq \pi^*_0$ in the two following cases:
 \begin{itemize}
 \item $\kappa=1$ (no capping in \eqref{equFDRconservative});
 \item $\kappa<1$ and for all $\theta = (\pi_0,F_0,F_1)\in \Theta$, $F_0=I$.
 \end{itemize}
\end{lemma}
 
\begin{proof}
    Assume $\bar{\pi}_0:=\bar{\pi}_0(\theta)< \pi^*_0$ for some configurations $\theta=(\pi_0,F_0,F_1)\in \Theta$ and $\theta^*=(\pi^*_0,F^*_0,F^*_1)\in \Theta$. 
Note that by \eqref{equconcentration} and since the distribution of the $p$-value family is the same under $\theta$ and $\theta^*$, we have
$$
\P_{\theta^*}(\tilde{\pi}_0\leq \bar{\pi}_0+ \eta_m)=\P_\theta(\tilde{\pi}_0\leq \bar{\pi}_0+ \eta_m)\geq 1-\delta_m,
$$
by denoting $\eta_m:=\eta_m(\theta)$ and $\delta_m(\theta):=\delta_m$.

Now first consider the case (i) where $\kappa=1$ (no capping).  Let $\alpha=(\bar{\pi}_0+\pi_0^*)/2\in (\bar{\pi}_0,\pi_0^*)$.
By definition of the ABH rejection number $\hat{k}$, we have 
\begin{align*}
\P_{\theta^*}(\hat{k}=m) 
&\geq \P_{\theta^*}(p_{(m)} \leq \alpha /\tilde{\pi}_0)\\
&\geq \P_{\theta^*}( 1 \leq \alpha /(\bar{\pi}_0+\eta_m))- \delta_m\\
&= \ind{\alpha\geq \bar{\pi}_0+\eta_m}- \delta_m\\
&\geq 1-\delta_m
\end{align*}
for $m$ large enough because $\alpha>\bar{\pi}_0$. In other words, the procedure $\ABH_\alpha(\tilde{\pi}_0)$ rejects all null hypotheses with probability at least $1-\delta_m$. Hence, we have
\begin{align*}
    \FDR_{\theta^*}(\ABH_\alpha(\tilde{\pi}_0)) &\geq \E_{\theta^*}[m_0/m] - \delta_m = \pi_0^*- \delta_m >\alpha
\end{align*}
for $m$ large enough because $\pi_0^*>\alpha$. This contradicts \eqref{equFDRconservative} for $\kappa=1$.

Now consider the case $\kappa<1$ and thus $F_0(t)=F^*_0(t)=t$ for all $t\in [0,1]$ by assumption. Hence,  both in the independent and conformal two-group models, we have both $\|\hat{F}_m-F\|_\infty$ and $\|\hat{F}_{0,m}-I\|_\infty$ converges in probability to $0$ for $\hat{F}_m (t):= m^{-1} \sum_{i=1}^m \ind{p_i\leq t}$ and $\hat{F}_{0,m} (t):= m_0^{-1} \sum_{i\in \cH_0} \ind{p_i\leq t}$. The convergence is with regard to $m$ in the independent case and with regard to both $n,m$ in the conformal case (see e.g. Theorem~\ref{thDKW}).

 Let $\alpha= (\bar{\pi}_0/2+\pi_0^*/2)\kappa/F(\kappa)$ which lies in $(0,1)$ because $F(\kappa)>0$ and $\pi_0^*\kappa/F(\kappa)\leq 1$. 
Let $t^*=\kappa \bar{\pi}_0/\alpha <\kappa$ and consider the threshold $\hat{t}=t_{\hat{k}}$ of $\ABH_\alpha(\tilde{\pi}_0,\kappa)$ for $t_k=\kappa \wedge (\alpha k/ (m \tilde{\pi}_0))$, $1\leq k\leq m$. 
Note that $\hat{t}$ is equivalently given by
$$
\hat{t}=\max\{t_k\in [0,1]\::\: \hat{F}_m (t_k) \geq t_k \tilde{\pi}_0/\alpha, 1\leq k\leq m\}.
$$
Let $\epsilon\in (0,\kappa)$ and let $\hat{t}_\epsilon= \min\{t_k \::\: t_k\geq \kappa  -\epsilon,1\leq k\leq m\}$. For $m$ large enough, this minimum exists on the event $\Omega_m=\{\tilde{\pi}_0\leq \bar{\pi}_0+\eta_m\}$ because it yields $\alpha/ \tilde{\pi}_0 \geq \alpha/(\bar{\pi}_0+\eta_m) \geq \kappa -\epsilon$ because $\alpha\geq \bar{\pi}_0$. Moreover, $\hat{t}_\epsilon\geq \kappa  -\epsilon$ on $\Omega_m$. Hence,
\begin{align*}
\P_{\theta^*}(\hat{t} < \hat{t}_\epsilon,\Omega_m )
&\leq \P_{\theta^*}(  \hat{F}_m(\hat{t}_\epsilon ) < \hat{t}_\epsilon  \tilde{\pi}_0/\alpha,\Omega_m) \\
&\leq \P_{\theta^*}(  \hat{F}_m(\hat{t}_\epsilon ) < \hat{t}_\epsilon  (\bar{\pi}_0+\eta_m)/\alpha,\Omega_m)\\
&= \P_{\theta^*}(  \|\hat{F}_m - F\|_\infty >  F(\hat{t}_\epsilon)-\hat{t}_\epsilon (\bar{\pi}_0+\eta_m)/\alpha   ,\Omega_m).
\end{align*}
Now, by definition of $\alpha$, we have $F(\hat{t}_\epsilon)-\hat{t}_\epsilon (\bar{\pi}_0+\eta_m)/\alpha\geq F(\kappa-\epsilon)-\kappa (\bar{\pi}_0+\eta_m)/\alpha>c$ on $\Omega_m$ for $m$ large enough and $c>0$ being some constant depending on $\epsilon$.
Hence, we obtain
\begin{align*}
\P_{\theta^*}(\hat{t} < \hat{t}_\epsilon)
&\leq \P_{\theta^*}(  \|\hat{F}_m - F\|_\infty > c)- \delta_m=o(1).
\end{align*}
Hence, by using in addition that $\hat{t}\leq \kappa$ and $\hat{t}_\epsilon\geq \kappa-\epsilon$ (a.s.), we have 
\begin{align*}
    \FDR_{\theta^*}(\ABH_\alpha(\tilde{\pi}_0,\kappa)) &= \E_{\theta^*}\bigg[\frac{m_0}{m} \frac{\hat{F}_{0,m} (\hat{t})}{ (1/m)\vee \hat{F}_{m} (\hat{t})} \bigg]\\
    &\geq \E_{\theta^*}\bigg[\frac{m_0}{m} \frac{\hat{F}_{0,m} (\hat{t}_\epsilon)}{ (1/m)\vee \hat{F}_{m} (\kappa)} \ind{\hat{t} \geq \hat{t}_\epsilon}\ind{\hat{F}_{m} (\kappa)\leq F(\kappa)/2}\bigg] \\
    &\geq \E_{\theta^*}\bigg[\frac{m_0}{m} \frac{\hat{F}_{0,m} (\kappa-\epsilon)}{ (1/m)\vee \hat{F}_{m} (\kappa)} \ind{\hat{t} \geq \hat{t}_\epsilon}\ind{\hat{F}_{m} (\kappa)\leq F(\kappa)/2}\bigg] .
   \end{align*} 
    Now let 
    \begin{align*}
\upsilon_m &:= \E_{\theta^*}\bigg[\frac{m_0}{m} \frac{(\kappa-\epsilon)}{F(\kappa)} \ind{\hat{t} \geq \hat{t}_\epsilon}\ind{\hat{F}_{m} (\kappa)\leq F(\kappa)/2}\bigg]\\
&- \E_{\theta^*}\bigg[\frac{m_0}{m} \frac{\hat{F}_{0,m} (\kappa-\epsilon)}{ (1/m)\vee \hat{F}_{m} (\kappa)} \ind{\hat{t} \geq \hat{t}_\epsilon}\ind{\hat{F}_{m} (\kappa)\leq F(\kappa)/2}\bigg]\\
&\leq \E_{\theta^*}\bigg[ \bigg|\frac{(\kappa-\epsilon)}{F(\kappa)} - \frac{\hat{F}_{0,m} (\kappa-\epsilon)}{ (1/m)\vee \hat{F}_{m} (\kappa)} \bigg|\ind{\hat{F}_{m} (\kappa)\leq F(\kappa)/2}\bigg]\\
&\leq \E_{\theta^*}\bigg[ \bigg|\frac{(\kappa-\epsilon) - \hat{F}_{0,m} (\kappa-\epsilon)}{F(\kappa)}\bigg|\bigg] + \E_{\theta^*}\bigg[ \hat{F}_{0,m} (\kappa-\epsilon) \bigg|\frac{{F} (\kappa)-(1/m)\vee \hat{F}_{m} (\kappa)}{ {F} (\kappa)\cdot [ (1/m)\vee \hat{F}_{m} (\kappa)]} \bigg|\ind{\hat{F}_{m} (\kappa)\leq F(\kappa)/2}\bigg]\\
&\leq \frac{\E_{\theta^*}\|\hat{F}_{0,m}-I\|_\infty}{F(\kappa)}  + \frac{\E_{\theta^*}\|\hat{F}_{m}-F\|_\infty}{F(\kappa)^2/2}=o(1).
\end{align*}
In addition, by the Markov inequality, we have $\P_{\theta^*}(\hat{F}_{m} (\kappa)\leq F(\kappa)/2)\leq (F(\kappa)/2)^{-1}\E_{\theta^*}\|\hat{F}_{m}-F\|_\infty=o(1)$. Combining the equations above leads to 
\begin{align*}
    \FDR_{\theta^*}(\ABH_\alpha(\tilde{\pi}_0,\kappa)) 
    &\geq \E_{\theta^*}\bigg[\frac{m_0}{m} \frac{\hat{F}_{0,m} (\kappa-\epsilon)}{ (1/m)\vee \hat{F}_{m} (\kappa)} \ind{\hat{t} \geq \hat{t}_\epsilon}\ind{\hat{F}_{m} (\kappa)\leq F(\kappa)/2}\bigg] \\
 &\geq \E_{\theta^*}\bigg[\frac{m_0}{m} \frac{(\kappa-\epsilon)}{F(\kappa)} \ind{\hat{t} \geq \hat{t}_\epsilon}\ind{\hat{F}_{m} (\kappa)\leq F(\kappa)/2}\bigg] - \upsilon_m\\
    &\geq \E_{\theta^*}\bigg[\frac{m_0}{m} \frac{(\kappa-\epsilon)}{F(\kappa)} \bigg] - \upsilon_m -\frac{(\kappa-\epsilon)}{F(\kappa)}\P_{\theta^*}(\hat{t} < \hat{t}_\epsilon)- \frac{(\kappa-\epsilon)}{F(\kappa)}\P_{\theta^*}(\hat{F}_{m} (\kappa)\leq F(\kappa)/2)\\
    &\geq \pi^*_0 \frac{(\kappa-\epsilon)}{F(\kappa)}-o(1) > \pi^*_0 \frac{(\kappa-2\epsilon)}{F(\kappa)}>\alpha
\end{align*}
for $m$ large enough and $\epsilon>0$ chosen small enough  (which is possible because  $\alpha<\pi_0^*\kappa/F(\kappa)$). This contradicts \eqref{equFDRconservative}.
Hence, $\bar{\pi}_0(\theta)\geq \pi^*_0$ in both cases.

\end{proof}

\subsection{Proof of Theorem~\ref{genthoptimal}}\label{secproofgenthoptimal}

Let us prove first (i). By \eqref{sufficientcondition}, with probability at least $1-\delta_t^*(\theta)$, we have 
$\hat{\pi}_0\leq \pi_0^*(\theta)+\eta^*_t(\theta)$ and thus on this event $\alpha/\hat{\pi}_{0}\geq \alpha/(\pi_0^*(\theta)+\eta^*_t(\theta))\geq  \alpha''/\pi_0^*$ for 
$$
\alpha'' = \alpha \big(1-\eta^*_t(\theta)/\pi_0^*(\theta)\big) \leq \pi_0^*(\theta) \alpha/(\pi_0^*(\theta)+\eta^*_t(\theta)),
$$
by using that $1-t\leq 1/(1+t)$ for all $t\in[0,1)$. Since the power is nondecreasing in the threshold, this gives the domination for the oracle ABH part.
As for the estimator $\tilde{\pi}_0$, we have by combining Lemma~\ref{lemoptim} and \eqref{equconcentration} we have with probability at least $1-\delta_t(\theta)$, $\tilde{\pi}_0\geq \pi_0^*(\theta)-\eta_t(\theta)$. This means on this event 
 $\alpha''/\pi^*_0(\theta)\geq \alpha''/(\tilde{\pi}_0+ \eta_t(\theta))\geq  \alpha'/\tilde{\pi}_0$ because 
$$
\alpha'=\alpha''(1-\eta_t(\theta)/(\pi_0^*(\theta)-\eta_t(\theta)))\leq \alpha''\tilde{\pi}_0 /(\tilde{\pi}_0+\eta_t(\theta))
$$
when $\tilde{\pi}_0\geq \pi_0^*(\theta)-\eta_t(\theta)$.

For (ii), let us first note that $\Pow_\theta(\ABH_\alpha(\hat{\pi}_{0},\kappa))\geq \Pow_\theta(\ABH_\alpha(\hat{\pi}'_{0},\kappa))$ for the random variable $\hat{\pi}'_{0}=\hat{\pi}_{0} \vee \pi_0^*(\theta)$.
Now, provided that $\alpha\leq \kappa \pi_0^*$, we have for $m$ large enough that for all $k\in [m]$, $\alpha k/(m\hat{\pi}'_0)\leq \alpha k /(m \pi_0^*)\leq \kappa$, which means that the capping has no effect: $\ABH_\alpha(\hat{\pi}'_{0},\kappa)$ reduces to $\ABH_\alpha(\hat{\pi}'_{0})$. As a result, we can apply the reasoning in (i) because $\hat{\pi}'_{0}$ also satisfies \eqref{sufficientcondition}.

\subsection{Proof of Lemma~\ref{lem:pi0starsuperunif}}\label{proof:lem:pi0starsuperunif}
 
We define the parameters corresponding  to $\pi_0^*(\theta) $ by
\begin{align}\label{eq:counter.example}
F^*_0(x)=\bigg(1-\dfrac{1-F(x)}{\pi_0^*(\theta)}\bigg)_+ ;\qquad
F^*_1(x):=1\wedge (F(x)/(1-\pi_0^*(\theta))),
\end{align}
with the convention that $F_1^*(x)=x$ for all $x\in [0,1]$ if  $\pi^*_0(\theta)=1$. We easily check that $\theta^*$ is a valid configuration, that is, $\theta^*\in \Theta_{\succeq U}$: $F^*_0$ is nondecreasing, $F^*_0(0)= 0$ and $F^*_0(x)\leq 1-(1-\lambda^*)(1-F(x))/(1-F(\lambda^*))\leq x$ (with $F^*_0(\lambda^*)=\lambda^*$) and $F^*_0(1)=1$ (if $\pi^*_0=1$, this is still well defined because $F_1^*$ is not needed). 

\subsection{Proof of Proposition~\ref{propsufficientcondsuperunif}}\label{proofpropsufficientcondsuperunif}

Proposition~\ref{propsufficientcondsuperunif} is a direct consequence of the following result.

\begin{proposition}[Concentration of $\hat{\pi}^{\MS}_{0,\epsilon,\underline{\pi}_0}$ below $\pi_0^*$]\label{propi0hatindepinf}
Consider $\theta\in \Theta_{\succeq U}$ with a corresponding $\pi_0^*\in [0,1]$ and $L(f)$ the Lipschitz constant of the corresponding mixture density $f=\pi_0 f_0 + \pi_1 f_1$.
Assume $\pi_0^*> \underline{\pi}_0$. Then the following holds:
\begin{itemize}
    \item[(i)] For $\epsilon>0$ and $a=\ind{\lambda^*> 1-\epsilon}\in \{0,1\}$, for any $\Delta_m\in [0,1]$ with $C(m,\epsilon,\underline{\pi}_0)-1\leq \Delta_m$ and $\pi_0^* \epsilon  > (1+\Delta_m)/m$, we have
\begin{equation}
\label{concpi0hatindepboundary}
\P_\theta(   \hat{\pi}^{\MS}_{0,\epsilon,\underline{\pi}_0} \geq  \pi_0^* + 4\pi_0^*\Delta_m+2 a L(f) \epsilon  )\leq e^{-2 m \epsilon^2\Delta^2_m (\pi_0^*)^2  }.
\end{equation}
\item[(ii)] Choosing an estimator $\hat{\pi}^{\MS}_{0,\epsilon,\underline{\pi}_0}$ with any parameter $(\epsilon,\underline{\pi}_0)$ such that $\epsilon \gg \sqrt{(\log m)/m}$
and $162/\log(m)\leq \underline{\pi}_0\ll 1 $  provides for $m$ large enough, 
\begin{equation}
\label{concpi0hatindepcorboundary}
\P_\theta\bigg(   \hat{\pi}^{\MS}_{0,\epsilon,\underline{\pi}_0} \geq  \pi_0^* + 4 \sqrt{\frac{\log m }{2 m   \epsilon^2}}+2 a L(f) \epsilon  \bigg)\leq 1/m.
\end{equation}
 \end{itemize}
 \end{proposition}
\begin{proof}
First consider the case where $\lambda^*\leq 1-\epsilon$.
    We have for $x>0$, provided that $\pi_0^*  (1-\lambda^*) > (1+\Delta_m)/m$,
\begin{align*}
\P(   \hat{\pi}^{\MS}_{0,\epsilon,\underline{\pi}_0} \geq \pi_0^* + x )&= \P\bigg(\underline{\pi}_0 \vee \bigg(C(m,\epsilon,\underline{\pi}_0)  \inf_{\lambda \in [0,1-\epsilon]}
      \frac{1 \vee \sum_{i=1}^{m} \mathbf{1}\{p_i \geq \lambda\}}{m(1-\lambda)} \big)\bigg)\geq \pi_0^* + x\bigg)  \\ 
      &\leq \P\bigg(  
     (1+\Delta_m) \frac{1 \vee\sum_{i=1}^{m} \mathbf{1}\{p_i \geq \lambda^*\}}{m(1-\lambda^*)} \geq \pi_0^* + x\bigg) \\
   &\leq \P\bigg(  
     (1+\Delta_m) \frac{\sum_{i=1}^{m} \mathbf{1}\{p_i \geq \lambda^*\}}{m(1-\lambda^*)} \geq \pi_0^* + x\bigg) \\
        &\leq \P\bigg(  
     (1+\Delta_m) m^{-1}\sum_{i=1}^{m} \mathbf{1}\{p_i \geq \lambda^*\} \geq \pi_0^*(1-\lambda^*) + x (1-\lambda^*)\bigg) 
\end{align*}
 Now by the DKW inequality (Theorem~\ref{thDKW}), for any $\delta>0$,
\begin{align*}
\P(   \hat{\pi}^{\MS}_{0,\epsilon,\underline{\pi}_0} \geq \pi_0^* + x )
        &\leq e^{-2 m \delta^2}
\end{align*}
provided that $\delta$ is chosen so that 
$
     (1+\Delta_m) (1-F(\lambda^*) +\delta) \leq \pi_0^*(1-\lambda^*) + x (1-\lambda^*)$. 
This gives the condition
$$
\delta\leq \frac{\pi_0^*(1-\lambda^*) +  x (1-\lambda^*)}{1+\Delta_m} - (1-F(\lambda^*))= (1-\lambda^*)\frac{x-\Delta_m \pi_0^*  }{1+\Delta_m} = (1-\lambda^*)\frac{2\Delta_m \pi_0^*  }{1+\Delta_m},
$$
by choosing $x=3\Delta_m \pi_0^*$.
Since $\Delta_m\leq 1$, the latter condition is satisfied if
$
\delta=  (1-\lambda^*)\Delta_m \pi_0^*,
$
which gives \eqref{concpi0hatindepboundary} in the case $\lambda^*\leq 1-\epsilon$ (for which $a=0$).

Now 
consider the case where $\lambda^*> 1-\epsilon$.
 Similarly to above, we have (by taking $\lambda=1-\epsilon$) for $x>0$ provided that $\pi_0^* \epsilon  > (1+\Delta_m)/m$,
\begin{align*}
\P(   \hat{\pi}^{\MS}_{0,\epsilon,\underline{\pi}_0} \geq \pi_0^* + x )
        &\leq \P\bigg(
     (1+\Delta_m) m^{-1}\sum_{i=1}^{m} \mathbf{1}\{p_i \geq 1-\epsilon\} \geq \pi_0^*\epsilon + x \epsilon)\bigg) 
\end{align*}
 Now 
by the DKW inequality (Theorem~\ref{thDKW}), for any $\delta>0$,
\begin{align*}
\P(   \hat{\pi}^{\MS}_{0,\epsilon,\underline{\pi}_0} \geq \pi_0^* + x )
        &\leq e^{-2 m \delta^2}
\end{align*}
provided that $\delta>0$ is chosen so that $
     (1+\Delta_m) (1-F(1-\epsilon) +\delta) \leq \pi_0^*\epsilon + x \epsilon$ that is
$$
\delta\leq \frac{\pi_0^*\epsilon + x \epsilon}{1+\Delta_m} - (1-F(1-\epsilon)).
$$
Now, we have by the mean value theorem that there exists $x_\epsilon\in (1-\epsilon,1)$ such that $1-F(1-\epsilon) =f(c_\epsilon) \epsilon$. 
Also, both when $\lambda^*<1$ and $\lambda^*=1$, there exists $x'_\epsilon\in (\lambda^*,1]\subset (1-\epsilon,1]$ such that $\pi_0^* =f(x'_\epsilon)$. This entails
$$
|\pi_0^*\epsilon - (1-F(1-\epsilon))| =| f(x'_\epsilon) - f(x_\epsilon)| \epsilon \leq L(f) |x'_\epsilon-x_\epsilon|\epsilon \leq L(f) \epsilon^2
$$
because $f$ is $L(f)$-Lipschitz. As a result, we can choose
$$
\delta= \epsilon(x/2 -\Delta_m\pi_0^*) - L(f) \epsilon^2 = \epsilon\Delta_m\pi_0^*  
$$
by choosing $x=4\Delta_m \pi_0^* + 2 L(f) \epsilon$, which gives \eqref{concpi0hatindepboundary} and point (i).

As for (ii), choosing $ \Delta^2_m = \log m / ( 2 m  (\pi_0^*)^2 \epsilon^2)$ gives $e^{-2 m (\epsilon\Delta_m\pi_0^*)^2}=1/m$. We have $\Delta_m\leq 1$ for large $m$ by the assumption on $\epsilon$. Also,  by Proposition~\ref{rateCD}, we have
$$
\Delta^2_m= \log m / ( 2 m  (\pi_0^*)^2 \epsilon^2) \geq 9^2/(m \underline{\pi}_0 \epsilon^2)\geq (C(m,\epsilon,\underline{\pi}_0)-1)^2,
$$
by using the condition on $\underline{\pi}_0$.
This 
gives 
\eqref{concpi0hatindepcorboundary} and (ii).
\end{proof}

\subsection{Proof of Lemma~\ref{lempi0starunif}}\label{proof:lempi0starunif}

Consider the alternative configuration
\begin{align}\label{eq:counter.exampleunif}
F^*_1:=\dfrac{F-\pi_0^*I}{1-\pi_0^*} ,
\end{align}
with the convention $F^*_1=I$ if $\pi_0^*=1$. 
It is clear that $\theta^*=(\pi^*_0,I,F^*_1)$ is in $\Theta_{U}$, because $F^*_1$ has the density 
$
f_1^*  = (f-\pi_0^*)/(1-\pi_0^*) .
$

\subsection{Proof of Proposition~\ref{propsufficientunifANDconformal}}\label{proofpropsufficientunifANDconformal}

We split the proof of Proposition~\ref{propsufficientunifANDconformal} in the two following propositions, that we prove separately.

\begin{proposition}\label{propsufficientunif}
 In the independent two-group model, by choosing  $\epsilon = m^{-1/4}$, $\underline{\pi}_0=0.5 \wedge (41/\log(2m)) $, we have 
for $m$ large enough,
\begin{equation}
\label{concpi0hatindepcorboundaryoptimuniform}
\P\bigg(  \hat{\pi}^{\IMS}_{0,\epsilon,\underline{\pi}_0,\kappa}  \geq  \bar{\pi}_0(\theta) +  5\frac{\sqrt{2\log (2m) }}{m^{1/4}}  \bigg)\leq 1/m, 
\end{equation}
with $\bar{\pi}_0(\theta):=\min_{x\in [\kappa,1]}f(x)$, which is such that $\bar{\pi}_0(\theta)=\pi_0^*(\theta)$ when $\kappa\leq \lambda^*$.
\end{proposition}
The proof of Proposition~\ref{propsufficientunif} is in Section~\ref{proofpropsufficientunif}.

\begin{proposition}\label{propsufficientconformal}
 In the conformal two-group model, by choosing  $\epsilon = (n\wedge m)^{-1/8}$, $\underline{\pi}_0=0.5 \wedge (1/\log(m)) $, we have 
for $n\wedge m$ large enough, and $n\gg \log m$,
\begin{equation}
\label{concpi0hatindepcorboundaryoptimuniform}
\P\bigg(  \hat{\pi}^{\IMS}_{0,\epsilon,\underline{\pi}_0,\kappa}  \geq  \bar{\pi}_0(\theta) +  10 (1+L(f))^{1/2}\frac{\sqrt{\log m }}{(n\wedge m)^{1/8}}  \bigg)\leq 1/m, 
\end{equation}
with $\bar{\pi}_0(\theta):=\min_{x\in [\kappa,1]}f(x)$, which is such that $\bar{\pi}_0(\theta)=\pi_0^*(\theta)$ when $\kappa\leq \lambda^*$.
\end{proposition}
The proof of Proposition~\ref{propsufficientconformal} is in Section~\ref{proofpropsufficientconformal}.

\subsection{Proof of Proposition~\ref{propsufficientunif}}\label{proofpropsufficientunif}

Proposition~\ref{propsufficientunif} is a direct consequence of the following result.

\begin{proposition}[Concentration of $\hat{\pi}^{\IMS}_{0,\epsilon,\underline{\pi}_0,\kappa}$]\label{propi0hatindepinfuniform}
Consider $\theta\in \Theta_{U}$ with a corresponding $\lambda^*\in [0,1]$ as in \eqref{pi0starunif}. 
Then the following holds for $\bar{\pi}_0:=\min_{x\in [\kappa,1]}f(x)\geq \pi_0^*$:
\begin{itemize}
    \item[(i)] Assume $\pi_0^*> \underline{\pi}_0$. Then for any $\Delta_m\in [0,1]$ with $D(m,\epsilon,\underline{\pi}_0) -1\leq \Delta_m$ and $3\pi_0^*\Delta_m+2 L(f) \epsilon<\bar{\pi}_0-\underline{\pi}_0$, we have
    
\begin{equation}
\label{concpi0hatindepboundaryuniform}
\P(   \hat{\pi}^{\IMS}_{0,\epsilon,\underline{\pi}_0,\kappa} \geq  \bar{\pi}_0 + 4\bar{\pi}_0 \Delta_m+2 L(f) \epsilon )\leq 2 e^{- m \epsilon^2\Delta^2_m (\bar{\pi}_0 )^2/2  },
\end{equation}
where $L(f)$ denotes the Lipschitz constant of the mixture density $f=\pi_0 f_0 + \pi_1 f_1$.

\item[(ii)] By choosing the parameter $(\epsilon,\underline{\pi}_0,\kappa)$ such that $\epsilon \gg \sqrt{(\log m)/m}$, $0.5 \vee (41/\log(2m))\leq \underline{\pi}_0\ll 1 $, we have for $m$ large enough
\begin{equation}
\label{concpi0hatindepcorboundaryuniform}
\P\bigg(  \hat{\pi}^{\IMS}_{0,\epsilon,\underline{\pi}_0,\kappa}  \geq  \bar{\pi}_0+ 4 \sqrt{\frac{2\log (2m) }{ m   \epsilon^2}}+2 L(f) \epsilon  \bigg)\leq 1/m.
\end{equation}
\end{itemize}
 \end{proposition}

\begin{proof}
        We have for $x>0$ provided that $\bar{\pi}_0 \epsilon> (1+\Delta_m)/m$ and by denoting $I^*=[\lambda^*,\lambda^*+\epsilon]$,
\begin{align*}
\P(  \hat{\pi}^{\IMS}_{0,\epsilon,\underline{\pi}_0,\kappa} \geq \bar{\pi}_0 + x )&= \P\bigg(\underline{\pi}_0 \vee \bigg(D(m,\epsilon,\underline{\pi}_0)  \inf_{I\subset [\kappa,1]: |I|\geq \epsilon}
      \frac{1 \vee \sum_{i=1}^{m} \mathbf{1}\{p_i \in I\}}{m|I|} \bigg)\geq \bar{\pi}_0 + x\bigg)  \\ 
      &\leq \P\bigg( 
     (1+\Delta_m) \frac{1 \vee \sum_{i=1}^{m} \mathbf{1}\{p_i \in I^*\}}{m\epsilon} \geq \bar{\pi}_0 + x\bigg) \\
     &\leq \P\bigg( 
     (1+\Delta_m) \frac{ \sum_{i=1}^{m} \mathbf{1}\{p_i \in I^*\}}{m\epsilon} \geq \bar{\pi}_0 + x\bigg) \\
        &\leq \P\bigg(  (1+\Delta_m) m^{-1}\sum_{i=1}^{m} \mathbf{1}\{p_i \in [\lambda^*,\lambda^*+\epsilon]\} \geq \bar{\pi}_0\epsilon + x \epsilon\bigg) .
\end{align*}
  Now by the DKW inequality (Theorem~\ref{thDKW}), for any $\delta>0$,
\begin{align*}
\P(   \hat{\pi}^{\IMS}_{0,\epsilon,\underline{\pi}_0,\kappa} \geq \bar{\pi}_0 + x )
        &\leq 2e^{- m \delta^2/2}
\end{align*}
provided that $\delta>0$ is chosen so that $
     (1+\Delta_m) (F(\lambda^*+\epsilon)-F(\lambda^*) +\delta) \leq \bar{\pi}_0\epsilon + x \epsilon$ that is
$$
\delta \leq \frac{\bar{\pi}_0\epsilon + x \epsilon}{1+\Delta_m} - (F(\lambda^*+\epsilon)-F(\lambda^*)).
$$
Now, by the mean value theorem, there exists $x_\epsilon \in [\lambda^*,\lambda^*+\epsilon]$ such that $F(\lambda^*+\epsilon)-F(\lambda^*)=f(x_\epsilon) \epsilon \leq \bar{\pi}_0\epsilon + (f(x_\epsilon)-f(\lambda^*)) \epsilon\leq \bar{\pi}_0\epsilon + L(f) \epsilon^2 $ by definition of $\bar{\pi}_0$ and because $f$ is $L(f)$-Lipschitz. 
As a result, we can choose
$$
\frac{\bar{\pi}_0\epsilon + x \epsilon}{1+\Delta_m} - (\bar{\pi}_0\epsilon + L(f) \epsilon^2)\geq \delta:= \epsilon(x/2 -\Delta_m\bar{\pi}_0) - L(f) \epsilon^2 = \epsilon\Delta_m\bar{\pi}_0  
$$
by letting $x=4\Delta_m \bar{\pi}_0 + 2 L(f) \epsilon$, which gives \eqref{concpi0hatindepboundaryuniform} and (i).

Now we choose 
$$ \Delta^2_m = 2\log (2m) / (  m  (\pi_0^*)^2 \epsilon^2)\geq 9^2/(m \underline{\pi}_0 \epsilon^2)\geq (D(m,\epsilon,\underline{\pi}_0)-1)^2,$$ by applying Proposition~\ref{rateCD} and for this choice of $\underline{\pi}_0$ (and with $\epsilon>0$ providing $\Delta_m\leq 1$). 
This leads to \eqref{concpi0hatindepcorboundary} and (ii). 
\end{proof}

\subsection{Proof of Proposition~\ref{propsufficientconformal}}\label{proofpropsufficientconformal}

Proposition~\ref{propsufficientconformal} is a direct consequence of the following result.

\begin{proposition}[Concentration of $\hat{\pi}^{\IMS}_{0,\epsilon,\underline{\pi}_0,\kappa}$ in the conformal case]\label{propi0hatconformal}
Consider $\theta\in \Theta_{U}$ with a corresponding $\lambda^*\in [0,1]$ as in \eqref{pi0starunif}. 
The the following holds for $\bar{\pi}_0:=\min_{x\in [\kappa,1]}f(x)\geq \pi_0^*$:
\begin{itemize}
    \item[(i)] Assume $\pi_0^*> \underline{\pi}_0$. Then for any $\Delta\in [0,1]$ with $D(m,\epsilon,\underline{\pi}_0) -1\leq \Delta$ and $3\pi_0^*\Delta+2 L(f) \epsilon<\bar{\pi}_0-\underline{\pi}_0$, we have
\begin{equation}
\label{concpi0hatindepboundaryuniform}
\P(   \hat{\pi}^{\IMS}_{0,\epsilon,\underline{\pi}_0,\kappa} \geq  \bar{\pi}_0 + 4\bar{\pi}_0 \Delta+2 L(f) \epsilon )\leq B(\tau^{1/2}_{n',m} \epsilon \Delta \bar\pi_0), 
\end{equation}
where $L(f)$ denotes the Lipschitz constant of the mixture density $f=\pi_0 f_0 + \pi_1 f_1$, $\tau_{n',m}=n'm/(n'+m)$ for $n'=n/(1+L(f))^2$ and $B(u)=2(1+\sqrt{2\pi}u)e^{-u^2/2}$ for $u>0$.
\item[(ii)] By choosing the parameter $(\epsilon,\underline{\pi}_0,\kappa)$ such that $1\gg \epsilon \geq 
 (2/\bar\pi_0) (\log m)^{1/2}(n\wedge m)^{-1/4}$, $ \underline{\pi}_0=1/\log(m) $, $n\gg \log m$, we have for $n\wedge m$ large enough
\begin{equation}
\label{concpi0hatindepcorboundaryuniform}
\P\bigg(  \hat{\pi}^{\IMS}_{0,\epsilon,\underline{\pi}_0,\kappa}  \geq  \bar{\pi}_0+ 9 (1+L(f))^{1/2} 
\frac{(\log m)^{1/2}}{(n\wedge m)^{1/4}\epsilon}  +2 L(f) \epsilon  \bigg)\leq 1/m.
\end{equation}
\end{itemize}
 \end{proposition}

\begin{proof}
The item (i) is obtained similarly to Proposition~\ref{propsufficientunif}: instead of using the classical DKW inequality, we should use the conformal DKW inequality with covariate shift in Theorem~\ref{thDKW} with $n'=n/(1+L(f))^2$. For this we use that the density $f$ is bounded by $1+L(f)$ and the fact that since $\tau_{n',m}\leq n'+m$, for all $x>0$,
$$
2\bigg(1+\frac{2\sqrt{2\pi} (x/2) \tau_{n',m}}{(n'+m)^{1/2}}\bigg)e^{-2 \tau_{n',m} (x/2)^2} \leq B(\tau^{1/2}_{n',m}x).
$$
Now prove (ii). Note that $B((3 \log u)^{1/2})\leq 1/u$ for $u>0$ large enough. Hence, Proposition~\ref{rateCD} suggests to choose 
$$
\Delta = \tau^{-\gamma}_{n',m} (3\log m)^{1/2} /(\epsilon \bar\pi_0) \geq \frac{34}{\sqrt{(n\wedge m)\underline{\pi}_0} \epsilon^2}, 
$$
for $n\wedge m$ large enough and some $\gamma<1/2$. 
Since 
$\tau^{-1/2}_{n',m} \Delta \epsilon \bar\pi_0\geq (3 \log m)^{1/2}$
, we have $B(\tau^{1/2}_{n',m} \epsilon \Delta \underline{\pi}_0  )\leq 1/m$ for  $n\wedge m$ large enough.
Choosing $\underline{\pi}_0=1/\log(m)$, this gives the following conditions for choosing $\epsilon$:
 $1\gg \epsilon \gg \sqrt{(\log m)/((n\wedge m)\underline{\pi}_0)}$, $n\gg \log m$ 
(Proposition~\ref{rateCD}), $\epsilon\geq \bar\pi_0^{-1}(3\log m)^{1/2}((n\wedge m)/2)^{-\gamma}$ ($\Delta\leq 1$, $ \tau_{n',m}\geq (n'\wedge m)/2\geq (n\wedge m)/2$) and $\epsilon \underline{\pi}^{1/2}_0 \geq 
34 (3 \log m)^{-1/2} (n\wedge m)^{-(1/2-\gamma)}$ that is $\epsilon  \geq (34/\sqrt{3})(n\wedge m)^{-(1/2-\gamma)}$ ($D(m,\epsilon,\underline{\pi}_0)-1\leq \Delta$). Hence, an appropriate choice is $\gamma=1/4$ and 
$\bar (2/\pi_0) (\log m)^{1/2}(n\wedge m)^{-1/4}$
is enough to satisfies all conditions. 
\end{proof}

\section{Auxiliary results}

\begin{lemma}[Variation of Lemma~D.6 in \cite{marandon2024adaptive}]\label{lemmaSU}
    Let $(\tau_k)_{k\geq 1}$ be any nondecreasing sequence of nonnegative values, a $p$-value family $\mbf{p}=(p_i)_{i\in [m]}$ and 
    $$
\hat{k}(\mbf{p}) = \max\{k\in [0,m]\::\: p_{(k)}\leq \tau_k\}
    $$
     be the rejection number of the corresponding step-up procedure, where $p_{(0)}:=0\leq p_{(1)}\leq \dots\leq p_{(m)}$ are the ordered $p$-values. Fix $i\in [m]$ and consider  $\mbf{p}'=(p'_i)_{i\in [m]}$ any $p$-value family with 
    $$
\forall j \in [m], \left\{ \begin{array}{cc}
     p'_j \leq p_j &  \mbox{ if $p_j\leq p_i$}\\
     p'_j = p_j& \mbox{ if $p_j> p_i$}
\end{array}\right. 
    $$
    Then $p_i\leq \tau_{\hat{k}(\mbf{p})}$ iff $p_i\leq \tau_{\hat{k}(\mbf{p}')}$ and both assertions  implies $\hat{k}(\mbf{p})=\hat{k}(\mbf{p}')$.
\end{lemma}

\begin{proof}
    First, since $p'_j \leq p_j$ for all $j$, we clearly have $\hat{k}(\mbf{p})\leq\hat{k}(\mbf{p}')$. This gives that $p_i\leq \tau_{\hat{k}(\mbf{p})}$ implies $p_i\leq \tau_{\hat{k}(\mbf{p}')}$. Now prove the other direction by assuming $p_i\leq \tau_{\hat{k}(\mbf{p}')}$ and by showing $\hat{k}(\mbf{p})\geq\hat{k}(\mbf{p}')$. For this, we have
    \begin{align*}
        \sum_{j=1}^m \ind{p_j \leq \tau_{\hat{k}(\mbf{p}')}} &=  \sum_{j=1}^m \ind{p_j \leq \tau_{\hat{k}(\mbf{p}')}}\ind{p_j\leq p_i} +  \sum_{j=1}^m \ind{p_j \leq \tau_{\hat{k}(\mbf{p}')}}\ind{p_j> p_i}\\
        &= \sum_{j=1}^m \ind{p_j\leq p_i} +  \sum_{j=1}^m \ind{p'_j \leq \tau_{\hat{k}(\mbf{p}')}}\ind{p_j> p_i}\\
        &= \sum_{j=1}^m \ind{p_j\leq p_i}\ind{p'_j\leq \tau_{\hat{k}(\mbf{p}')}} +  \sum_{j=1}^m \ind{p'_j \leq \tau_{\hat{k}(\mbf{p}')}}\ind{p_j> p_i}\\
        &= \sum_{j=1}^m \ind{p'_j \leq \tau_{\hat{k}(\mbf{p}')}} \geq \hat{k}(\mbf{p}').
    \end{align*}
    Hence, by definition of $\hat{k}(\mbf{p})$ we have $\hat{k}(\mbf{p}')\leq \hat{k}(\mbf{p})$ and the result is proved.
\end{proof}

\begin{theorem}[DKW inequalities]\label{thDKW}
The following inequalities holds for the process $\hat{F}_{m}(t) =m^{-1}\sum_{i=1}^{m} \mathbf{1}\{p_i \leq t\}$, $t\in [0,1]$, and any  $x>0$:
    \begin{itemize}
        \item Classical \citep{massart1990tight}: if $p_i$, $i\in [m]$, are iid of common CDF $ F$, we have 
        \begin{align}
             \P\bigg(\exists t\in [0,1]:  \hat{F}_{m}(t) -F(t)\leq -x\bigg)&\leq e^{-2m x^2}\label{equDKWleft}\\
             \P\bigg(\exists t\in [0,1]: \hat{F}_{m}(t) -  F(t)\geq x\bigg)&\leq e^{-2m x^2}\label{equDKWleft};
        \end{align}
        \item Conformal exchangeable \citep{gazin2023transductive}: in the conformal setting of Section~\ref{sec:settings}, when $\cH_0=[m]$,
         \begin{align}
             \P\bigg(\exists t\in [0,1]: \hat{F}_{m}(t)  -  \frac{\lceil (n+1)t \rceil}{n+1}\leq -x\bigg)&\leq \bigg(1+\frac{2\sqrt{2\pi}x \tau_{n,m}}{(n+m)^{1/2}}\bigg)e^{-2 \tau_{n,m} x^2}\label{equDKWleft}\\
             \P\bigg(\exists t\in [0,1]: \hat{F}_{m}(t)  - \frac{\lceil (n+1)t \rceil}{n+1}\geq x\bigg)&\leq \bigg(1+\frac{2\sqrt{2\pi}x \tau_{n,m}}{(n+m)^{1/2}}\bigg)e^{-2 \tau_{n,m} x^2}\label{equDKWleft};
        \end{align}
        for $\tau_{n,m}=nm/(n+m)\in [(n\wedge m)/2,n\wedge m]$.
        \item Conformal with covariate shift \citep{Ruellan2026}: in the conformal setting of Section~\ref{sec:optimality}, assuming that $F$ has a density $f$ uniformly bounded, that is, $f(x)\leq C<+\infty$, $x\in [0,1]$,
\begin{align}
             \P\bigg(\exists t\in [0,1]: \hat{F}_{m}(t) -  F(t)< -x\bigg)&\leq \bigg(1+\frac{2\sqrt{2\pi}x \tau_{n',m}}{(n'+m)^{1/2}}\bigg)e^{-2 \tau_{n',m} x^2}\label{equDKWleft}\\
             \P\bigg(\exists t\in [0,1]: \hat{F}_{m}(t) - F(t)> x\bigg)&\leq \bigg(1+\frac{2\sqrt{2\pi}x \tau_{n',m}}{(n'+m)^{1/2}}\bigg)e^{-2 \tau_{n',m} x^2}\label{equDKWright};
        \end{align}
        for $\tau_{n',m}=n'm/(n'+m)\in [(n'\wedge m)/2,n'\wedge m]$ and $n'=n/C^2$.
    \end{itemize}
\end{theorem}

\begin{lemma}
    \label{lem:changetildeconf}
    Denote the joint distribution of conformal $p$-values $p_i,$ $i\in [m]$ \eqref{equ-confpvalues} by $P_{n,m}$. Consider    
    their modification $q_j=\tilde{p}^{0,i}_j$, $j\in [m]$ as in  \eqref{equ:configconf} for some arbitrary $i\in [m]$. Then the following holds:
    \begin{itemize}
        \item[(i)] the joint distribution of $(((n+1)q_j+1)/(n+2))_{j\in [m]\backslash\{i\}}$ is $P_{n+1,m-1}$;
        \item[(ii)] We have
         \begin{align*}
             &\P\bigg(\exists t\in [0,1]: (m-1)^{-1}\sum_{j\in [m]\backslash\{i\}} \ind{q_j\leq t}  -  \frac{\lceil (n+1)t \rceil+1}{n+2}\leq -x\bigg)\\
             &\leq \bigg(1+\frac{2\sqrt{2\pi}x \tau_{n+1,m-1}}{(n+m)^{1/2}}\bigg)e^{-2 \tau_{n+1,m-1} x^2}\\
             &\P\bigg(\exists t\in [0,1]: (m-1)^{-1}\sum_{j\in [m]\backslash\{i\}} \ind{q_j\leq t}  - \frac{\lceil (n+1)t \rceil+1}{n+2}\geq x\bigg)\\&\leq \bigg(1+\frac{2\sqrt{2\pi}x \tau_{n+1,m-1}}{(n+m)^{1/2}}\bigg)e^{-2 \tau_{n+1,m-1} x^2};
        \end{align*}
        for $\tau_{n,m}=nm/(n+m)\in [(n\wedge m)/2,n\wedge m]$.
    \end{itemize}
\end{lemma}

\begin{proof}
    Item (i) is clear from the definition  \eqref{equ:configconf}. Item (ii) comes from the fact that letting $p'_j=((n+1)q_j+1)/(n+2)$, $j\in [m]\backslash\{i\}$, we have  
    $$
(m-1)^{-1}\sum_{j\in [m]\backslash\{i\}} \ind{q_j\leq t} 
=(m-1)^{-1}\sum_{j\in [m]\backslash\{i\}} \ind{p'_j\leq (t(n+1)+1)/(n+2)}.
$$
We conclude  by applying Theorem~\ref{thDKW}. 
\end{proof}

\begin{lemma}[Lemma~27 in \cite{BR2009}]
    \label{lem:BR2009}
    Let $g:[0,1]\to (0,\infty)$ a nonincreasing function and $U$ be a super-uniform random variable, then for any $c>0$, we have
    $$
\E\bigg[\frac{\ind{U\leq c g(U)}}{g(U)}\bigg]\leq c.
    $$
\end{lemma}

\begin{lemma}[\cite{marandon2024adaptive}]\label{lem:conformal}
 In the conformal setting (including the no-ties assumption and exchangeability of $(S_k,k\in [n]\cup\{n+i\})$ given $(S_{n+j}, j \in [m]\backslash\{i\})$), let for any fixed $i\in \cH_{0}$, $W_i$ as in \eqref{equ:wi}.  Then we have
\begin{itemize}
    \item $(p_j)_{j\in\range{m}\backslash\{i\}}$ is equal to $(\Psi_{ij}(p_i,W_i))_{j\in\range{m}\backslash\{i\}}$ for some nondecreasing functions $u\in [0,1] \mapsto \Psi_{ij}(u,W_i)\in \R$, $j\in\range{m}\backslash\{i\}$ ;
    \item $(n+1)p_i $ is uniformly distributed on $\range{n+1}$ and independent of $W_i$.
\end{itemize}  
\end{lemma}

\begin{lemma}
    \label{lem:kappahat}
    The quantity $\hat{\kappa}=\hat{\kappa}(\mbf{p})$ defined by \eqref{kappahat} satisfies the following properties almost surely:
    \begin{itemize}
        \item[(i)] $\kappa\in [0,1]\mapsto \hat{\pi}^{\IMS}_{0,\epsilon,\underline{\pi}_0,\kappa}$ is non-decreasing anf left-continuous;
        \item[(ii)] $\hat{\kappa}$ in \eqref{kappahat} is a maximum, that is, $\hat{F}_m(\hat{\kappa})\geq \hat{\kappa}  \hat{\pi}^{\IMS}_{0,\epsilon,\underline{\pi}_0,\hat{\kappa}}/\alpha$;
        \item[(iii)] $\ABH_\alpha(\hat{\pi}^{\IMS}_{0,\epsilon,\underline{\pi}_0,\hat{\kappa}},\hat{\kappa})$ corresponds to the procedure rejecting $p$-values at most  $\hat{\kappa}$;
        \item[(iv)] All procedures $\ABH_\alpha(\hat{\pi}^{\IMS}_{0,\epsilon,\underline{\pi}_0,\kappa},\kappa)$, $\kappa\in [0,1]$, have rejection sets contained in the rejection set of $\ABH_\alpha(\hat{\pi}^{\IMS}_{0,\epsilon,\underline{\pi}_0,\hat{\kappa}},\hat{\kappa})$;
    \end{itemize}
\end{lemma}

\begin{proof}
    For point (i), the monotonicity is obvious because $\hat{\pi}^{\IMS}_{0,\epsilon,\underline{\pi}_0,\kappa}$ is an infimum on $[\kappa,1]$. The left-continuity at some $\kappa_0$ comes from \eqref{equ:simplifyinfIMS}: if $\kappa_0$ is not a $p$-value, then it is clear that  $\kappa\in [0,1]\mapsto \hat{\pi}^{\IMS}_{0,\epsilon,\underline{\pi}_0,\kappa}$ is continuous at $\kappa_0$. Else, $\kappa_0$ is a $p$-value $p_{i_0}$ and the minimum in \eqref{equ:simplifyinfIMS} is attained at some $J_0=(p_{j_0},b)$. Denote also $I_0=(p_{i_0},b')$ the best interval starting at $p_{i_0}$ (which can be different than $J_0$). 
    Considering $\upsilon\in (0,\epsilon)$ smaller than the gap between two different $p$-values, we have
    $$
    \frac{ \sum_{i=1}^{m} \mathbf{1}\{p_i \in (\kappa_0-\upsilon,b)\}}{|I_0|+\upsilon} =\frac{ 1+\sum_{i=1}^{m} \mathbf{1}\{p_i \in I_0\}}{|I_0|+\upsilon}> \frac{ \sum_{i=1}^{m} \mathbf{1}\{p_i \in I_0\}}{|I_0|}\geq \frac{ \sum_{i=1}^{m} \mathbf{1}\{p_i \in J_0\}}{|J_0|},
    $$
    so the minimum in \eqref{equ:simplifyinfIMS} when $\kappa=\kappa_0-\upsilon$ still in $J_0$, and  $\hat{\pi}^{\IMS}_{0,\epsilon,\underline{\pi}_0,\kappa-\upsilon}=\hat{\pi}^{\IMS}_{0,\epsilon,\underline{\pi}_0,\kappa}$ which proves the left-continuity at $\kappa_0$.

    Point (ii) comes from the fact that, by definition of the supremum, for all $\upsilon>0$ there exists $\kappa_\upsilon\in (\hat{\kappa}-\upsilon,\hat{\kappa}]$ such that $\hat{F}_m(\kappa_\upsilon)\geq \kappa_\upsilon  \hat{\pi}^{\IMS}_{0,\epsilon,\underline{\pi}_0,\kappa_\upsilon}/\alpha$. This entails
    $$
\hat{F}_m(\hat{\kappa})\geq \hat{F}_m(\kappa_\upsilon)\geq \kappa_\upsilon  \hat{\pi}^{\IMS}_{0,\epsilon,\underline{\pi}_0,\kappa_\upsilon}/\alpha .
    $$
    Now,  the right-hand-side of the above relation tends to $\hat{\kappa}  \hat{\pi}^{\IMS}_{0,\epsilon,\underline{\pi}_0,\hat{\kappa}}/\alpha $ when $\upsilon$ tends to $0$ by the left continuity proved in point (i). This gives the result.

Let us prove point (iii). By point (ii), we have $\hat{F}_m(\hat{\kappa})\geq \hat{\kappa}  \hat{\pi}^{\IMS}_{0,\epsilon,\underline{\pi}_0,\hat{\kappa}}/\alpha$. Letting $k_0=m\hat{F}_m(\hat{\kappa})$. This gives 
$$
p_{(k_0)} \leq \hat{\kappa} = \hat{\kappa}\wedge (\alpha k_0/ (m \hat{\pi}^{\IMS}_{0,\epsilon,\underline{\pi}_0,\hat{\kappa}}))
,$$
because of the definition of $k_0$. By \eqref{equ-lhat}, the rejection number of $\ABH_\alpha(\hat{\pi}^{\IMS}_{0,\epsilon,\underline{\pi}_0,\hat{\kappa}},\hat{\kappa})$ is then $\hat{k}\geq k_0$ and its threshold is given by $\alpha \hat{k}/ (m \hat{\pi}^{\IMS}_{0,\epsilon,\underline{\pi}_0,\hat{\kappa}})$ with
$$
\alpha \hat{k}/ (m \hat{\pi}^{\IMS}_{0,\epsilon,\underline{\pi}_0,\hat{\kappa}}) \geq \alpha k_0/ (m \hat{\pi}^{\IMS}_{0,\epsilon,\underline{\pi}_0,\hat{\kappa}}) \geq \hat{\kappa},
$$
and thus is $\hat{\kappa}$. 

For (iv), let us consider any $\kappa\in [0,1]$, and the procedure $\ABH_\alpha(\hat{\pi}^{\IMS}_{0,\epsilon,\underline{\pi}_0,\kappa},\kappa)$.  If $\kappa\leq \hat{\kappa}$, it is clear that the threshold of $\ABH_\alpha(\hat{\pi}^{\IMS}_{0,\epsilon,\underline{\pi}_0,\kappa},\kappa)$ is at most $\kappa\leq \hat{\kappa}$ and thus of the one of $\ABH_\alpha(\hat{\pi}^{\IMS}_{0,\epsilon,\underline{\pi}_0,\hat{\kappa}},\hat{\kappa})$ by point (iii). If $\kappa> \hat{\kappa}$, we have $\hat{F}_m(\kappa)<\kappa \hat{\pi}^{\IMS}_{0,\epsilon,\underline{\pi}_0,\kappa}/\alpha$ by definition \eqref{kappahat} of $\hat{\kappa}$. 
Now consider the rejection threshold $\hat{t}=(\alpha \hat{k}/ (m\hat{\pi}^{\IMS}_{0,\epsilon,\underline{\pi}_0,\kappa}))\wedge \kappa$ of $\ABH_\alpha(\hat{\pi}^{\IMS}_{0,\epsilon,\underline{\pi}_0,\kappa},\kappa)$ as in \eqref{equ-lhat} and prove that $\hat{t}\leq \hat{\kappa}$. We have 
$
\hat{k}/m= \hat{F}_m(\hat{t}) \leq \hat{F}_m(\kappa) < \kappa \hat{\pi}^{\IMS}_{0,\epsilon,\underline{\pi}_0,\kappa}/\alpha,
$
which means $\alpha \hat{k}/ (m\hat{\pi}^{\IMS}_{0,\epsilon,\underline{\pi}_0,\kappa})\leq \kappa$ and thus 
$$
\hat{t}=\alpha \hat{k}/ (m\hat{\pi}^{\IMS}_{0,\epsilon,\underline{\pi}_0,\kappa})\leq \kappa.
$$
Since $\hat{\pi}^{\IMS}_{0,\epsilon,\underline{\pi}_0,\kappa}\geq \hat{\pi}^{\IMS}_{0,\epsilon,\underline{\pi}_0,\hat{t}}$ by point (i), this entails
$$
\hat{F}_m(\hat{t}) = \hat{k}/m = \hat{t} \hat{\pi}^{\IMS}_{0,\epsilon,\underline{\pi}_0,\kappa}/\alpha \geq \hat{t} \hat{\pi}^{\IMS}_{0,\epsilon,\underline{\pi}_0,\hat{t}}/\alpha,
$$
and thus $\hat{t}\leq \hat{\kappa}$ by  definition \eqref{kappahat}. This proves (iv).

\end{proof}

\begin{lemma}\label{lem:kappahat2}
Consider $\hat{\kappa}(\mbf{p})$ defined by \eqref{kappahat} for some $\alpha\in(0,1)$.
    Fix $i\in [m]$ and consider  $\mbf{p}'=(p'_i)_{i\in [m]}$ any $p$-value family with 
    $$
\forall j \in [m], \left\{ \begin{array}{cc}
     p'_j \leq p_j &  \mbox{ if $p_j\leq p_i$}\\
     p'_j = p_j& \mbox{ if $p_j> p_i$}
\end{array}\right. 
    $$
    Then $p_i\leq \hat{\kappa}(\mbf{p})$ implies $\hat{\kappa}(\mbf{p})=\hat{\kappa}(\mbf{p}')$ and $\hat{\pi}^{\IMS}_{0,\epsilon,\underline{\pi}_0,\hat{\kappa}(\mbf{p})}(\mbf{p})=\hat{\pi}^{\IMS}_{0,\epsilon,\underline{\pi}_0,\hat{\kappa}(\mbf{p}')}(\mbf{p}')$.
\end{lemma}

\begin{proof}
Let us first rewrite \eqref{kappahat} as 
$$\hat{\kappa}(\mbf{p}) = \max\{\kappa \in [0,1-\epsilon]\::\: \hat{F}_m(\kappa,\mbf{p})\geq \kappa  \hat{\pi}^{\IMS}_{0,\epsilon,\underline{\pi}_0,\kappa}(\mbf{p})/\alpha\},
$$
for which the dependence with regard to the $p$-value vector $\mbf{p}$ is made explicit in the notation (we also use Lemma~\ref{lem:kappahat} (ii) to set $\hat{\kappa}(\mbf{p})$ as a maximum). 

Now assume $p_i\leq \hat{\kappa}(\mbf{p})$. Let us establish the following
\begin{equation}
    \label{equ:usefulpi0hat}
    \forall \kappa\geq \hat{\kappa}(\mbf{p}), \:\: \hat{\pi}^{\IMS}_{0,\epsilon,\underline{\pi}_0,\kappa}(\mbf{p}')=\hat{\pi}^{\IMS}_{0,\epsilon,\underline{\pi}_0,\kappa}(\mbf{p}).  
\end{equation}
Indeed, by definition
\begin{align*}
\hat{\pi}^{\IMS}_{0,\epsilon,\underline{\pi}_0,\kappa}(\mbf{p}') &= \underline{\pi}_0 \vee \bigg(D(m,\epsilon,\underline{\pi}_0)  \inf_{I\subset [\kappa,1]: |I|\geq \epsilon}
      \frac{1 \vee \sum_{j=1}^{m} \mathbf{1}\{p'_j \in I\}}{m|I|} \bigg)\\
      &=\underline{\pi}_0 \vee \bigg(D(m,\epsilon,\underline{\pi}_0)  \inf_{I\subset [\kappa,1]: |I|\geq \epsilon}
      \frac{1 \vee \sum_{j=1}^{m} \mathbf{1}\{p_j \in I\}}{m|I|} \bigg) \\
      &= \hat{\pi}^{\IMS}_{0,\epsilon,\underline{\pi}_0,\kappa}(\mbf{p}),
     \end{align*}
because for all $j\in [m]$, $\mathbf{1}\{p_j \in I\}=\mathbf{1}\{p'_j \in I\}$ for all intervals $I\subset [\kappa,1]$: if $p_j \in I$ this means that $p_j>\kappa\geq p_i$ and thus $p'_j=p_j$ and $p'_j \in I$. Conversely, if $p'_j \in I$ then $p'_j>\kappa\geq p_i$ which means also $p'_j=p_j$ and $p_j \in I$. This gives \eqref{equ:usefulpi0hat}.

Now we prove $\hat{\kappa}(\mbf{p})\leq\hat{\kappa}(\mbf{p}')$. For this, we note that
\begin{align*}
    \hat{F}_m(\hat{\kappa}(\mbf{p}),\mbf{p}') &\geq    \hat{F}_m(\hat{\kappa}(\mbf{p}),\mbf{p}) \geq \hat{\kappa}(\mbf{p})  \hat{\pi}^{\IMS}_{0,\epsilon,\underline{\pi}_0,\hat{\kappa}(\mbf{p})}(\mbf{p})/\alpha =\hat{\kappa}(\mbf{p})  \hat{\pi}^{\IMS}_{0,\epsilon,\underline{\pi}_0,\hat{\kappa}(\mbf{p})}(\mbf{p}')/\alpha ,
\end{align*}
by definition of $\hat{\kappa}(\mbf{p})$ as a maximum (Lemma~\ref{lem:kappahat} (ii)) and 
by using \eqref{equ:usefulpi0hat} in the last step.
By definition of $\hat{\kappa}(\mbf{p}')$, this means $\hat{\kappa}(\mbf{p})\leq\hat{\kappa}(\mbf{p}')$.

Conversely, let us prove $\hat{\kappa}(\mbf{p})\geq\hat{\kappa}(\mbf{p}')$. For this, we have
    \begin{align*}
    \hat{F}_m(\hat{\kappa}(\mbf{p}'),\mbf{p})  
    &= m^{-1} \sum_{j=1}^m \ind{p_j \leq \hat{\kappa}(\mbf{p}')}\ind{p_j\leq p_i} + m^{-1} \sum_{j=1}^m \ind{p_j \leq \hat{\kappa}(\mbf{p}')}\ind{p_j> p_i}\\
        &=m^{-1} \sum_{j=1}^m \ind{p_j\leq p_i} + m^{-1} \sum_{j=1}^m \ind{p'_j \leq \hat{\kappa}(\mbf{p}')}\ind{p_j> p_i}\\
        &=m^{-1} \sum_{j=1}^m \ind{p_j\leq p_i}\ind{p'_j\leq \hat{\kappa}(\mbf{p}')} + m^{-1} \sum_{j=1}^m \ind{p'_j \leq \hat{\kappa}(\mbf{p}')}\ind{p_j> p_i}\\
        &=  \hat{F}_m(\hat{\kappa}(\mbf{p}'),\mbf{p}') \geq \hat{\kappa}(\mbf{p}')  \hat{\pi}^{\IMS}_{0,\epsilon,\underline{\pi}_0,\hat{\kappa}(\mbf{p}')}(\mbf{p}')/\alpha=\hat{\kappa}(\mbf{p}')  \hat{\pi}^{\IMS}_{0,\epsilon,\underline{\pi}_0,\hat{\kappa}(\mbf{p}')}(\mbf{p})/\alpha,
    \end{align*}
    by definition of $\hat{\kappa}(\mbf{p}')$ as a maximum (Lemma~\ref{lem:kappahat} (ii)) and 
by using \eqref{equ:usefulpi0hat} in the last step (because $\hat{\kappa}(\mbf{p}')\geq\hat{\kappa}(\mbf{p})$).
By definition of $\hat{\kappa}(\mbf{p})$, this leads to $\hat{\kappa}(\mbf{p})\geq \hat{\kappa}(\mbf{p}')$ and the result is proved.
\end{proof}

\section{Additional figures}\label{appe:sec:simulation}

\subsection{Numerical evaluation of $c(s, \epsilon)$ and $d(s, \epsilon)$}\label{appe:sec:constant}

In Figure~\ref{fig:adjustment.constant}, we numerically evaluate $c(s, \epsilon)$ in \eqref{equ:gMS} and $d(s, \epsilon)$  in \eqref{equ:gIMS} based on independent $U(0,1)$ $p$-values for a variety of sample sizes $s$.
We also compute the constants for conformal $p$-values with validation sample size $n \in \{500,1000,1500, 2000,2500\}$.

We observe that both $c(s, \epsilon)$ and $d(s, \epsilon)$ decrease as $s$ increases.
For the same $p$-value distribution and tuning parameters, $c(s, \epsilon)$ is smaller than $d(s, \epsilon)$. Moreover, both constants are larger for conformal $p$-values than for independent uniform $p$-values.
When $s \ge 500$, the constant $c(s, \epsilon)$ is below $1.1$ for independent $p$-values and conformal $p$-values ($n \ge 1000$). For $d(s, \epsilon)$, the corresponding values are below $1.3$ for independent $p$-values and conformal $p$-values ($n \ge 1000$). 

\begin{figure}[tbp]
        \centering
        \begin{minipage}{0.43\textwidth}
                \centering
                \includegraphics[clip, trim = 0cm 0cm 0cm 0cm, height = 0.8\textwidth]{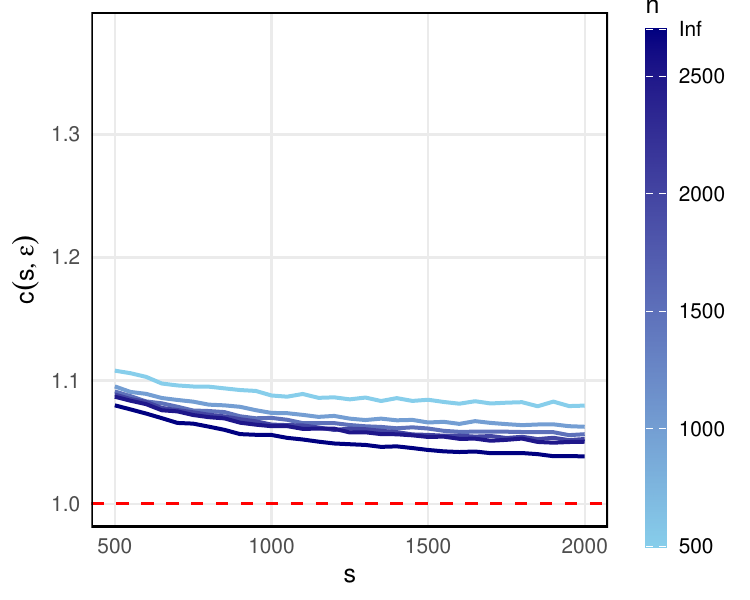}
                \subcaption{Numerical evaluation of  $c(s,\epsilon)$}
        \end{minipage}
         \begin{minipage}{0.43\textwidth}
                \centering
                \includegraphics[clip, trim = 0cm 0cm 0cm 0cm, height = 0.8\textwidth]{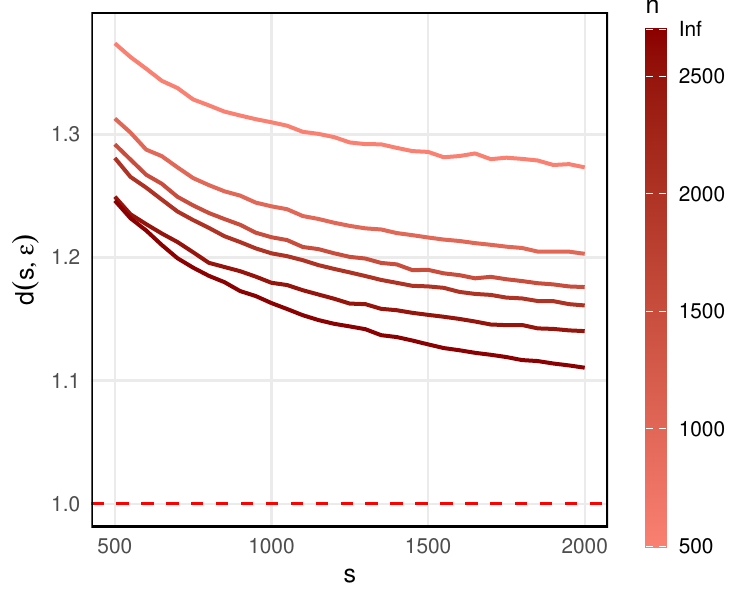}
                \subcaption{Numerical evaluation of  $d(s,\epsilon)$}
        \end{minipage}
            \caption{Numerical evaluation of $c(s,\epsilon)$ and $d(s,\epsilon)$ for $s \in \{500,550,\ldots,2000\}$, $\epsilon = 0.2$, and $n \in \{500,1000,1500,2000,2500, \infty\}$. Here $U(0,1)$ corresponds to the conformal setting in the limit $n=\infty$. Both $c(s,\epsilon)$ and $d(s,\epsilon)$ appear to decrease with $s$ and the validation sample size $n$. For fixed $s$, $\epsilon$, $n$, $c(s,\epsilon) \le d(s,\epsilon)$. Each value is estimated using $4000$ Monte Carlo samples.}
        \label{fig:adjustment.constant}
\end{figure}

\et{Say that the picture is worst if $\epsilon$ is near zero}

\subsection{Simulation}

\subsubsection{Conformal $p$-values}\label{appe:sec:simulation.conformal}

We replicate settings (a.1), (a.2), and (a.3) in Section~\ref{sec:simulation}, replacing the independent $p$-values with conformal $p$-values constructed using a validation dataset of size $n = 1000$. 
The power comparison is similar to that for independent $p$-values in Figure~\ref{fig:simulation.uniform}, except that the correction constants for IMS are more conservative (see panel (c.3)).

\begin{figure}[tbp]
        \centering
        \begin{minipage}{0.8\textwidth}
                \centering
                \includegraphics[clip, trim = 0cm 1cm 0cm 0cm, width = 1\textwidth]{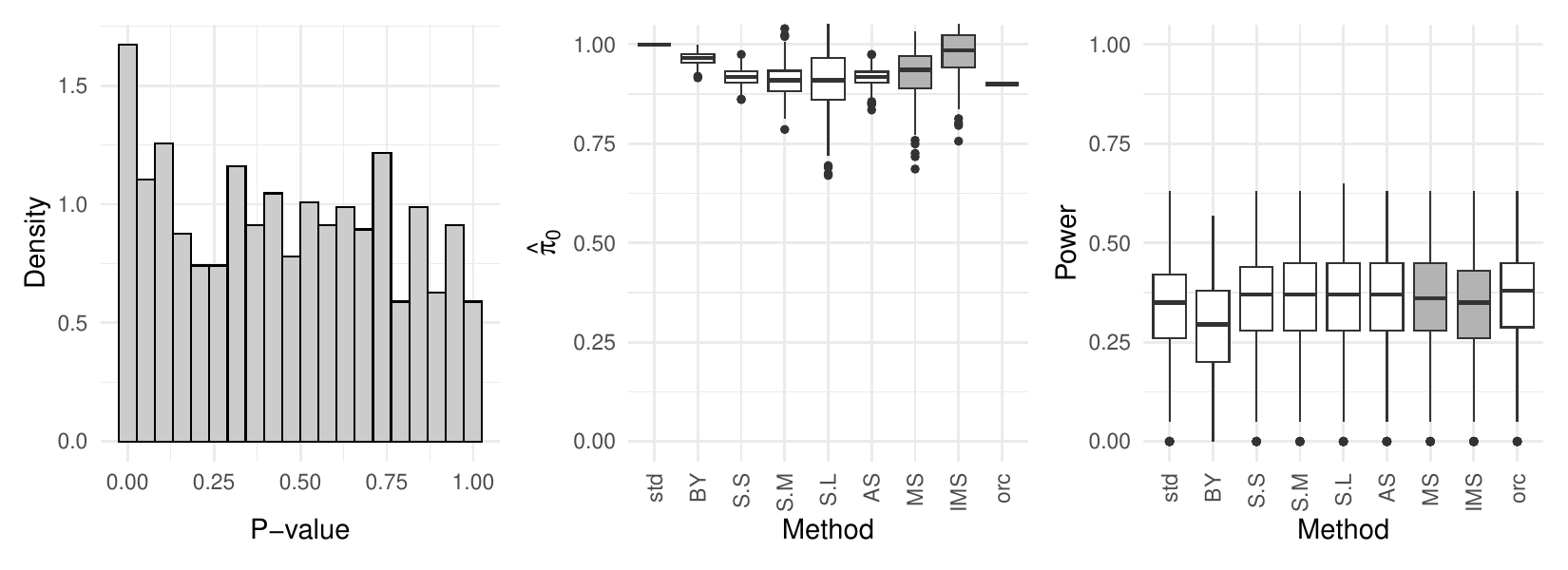}
        \subcaption*{(c.1) Few strong signals}
        \end{minipage}
        \begin{minipage}{0.8\textwidth}
                \centering
                \includegraphics[clip, trim = 0cm 1cm 0cm 0cm, width = 1\textwidth]{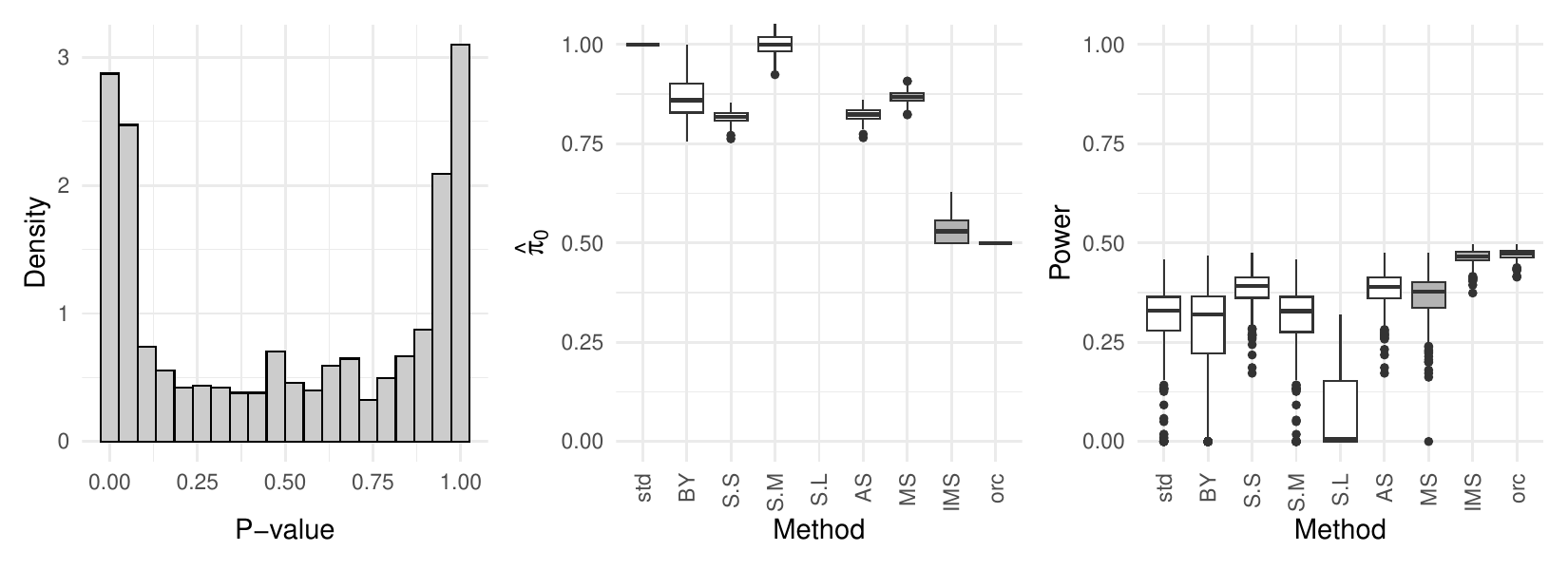}
        \subcaption*{ (c.2) Close-to-one alternatives}
        \end{minipage}    
        \begin{minipage}{0.8\textwidth}
                \centering
                \includegraphics[clip, trim = 0cm 1cm 0cm 0cm, width = 1\textwidth]{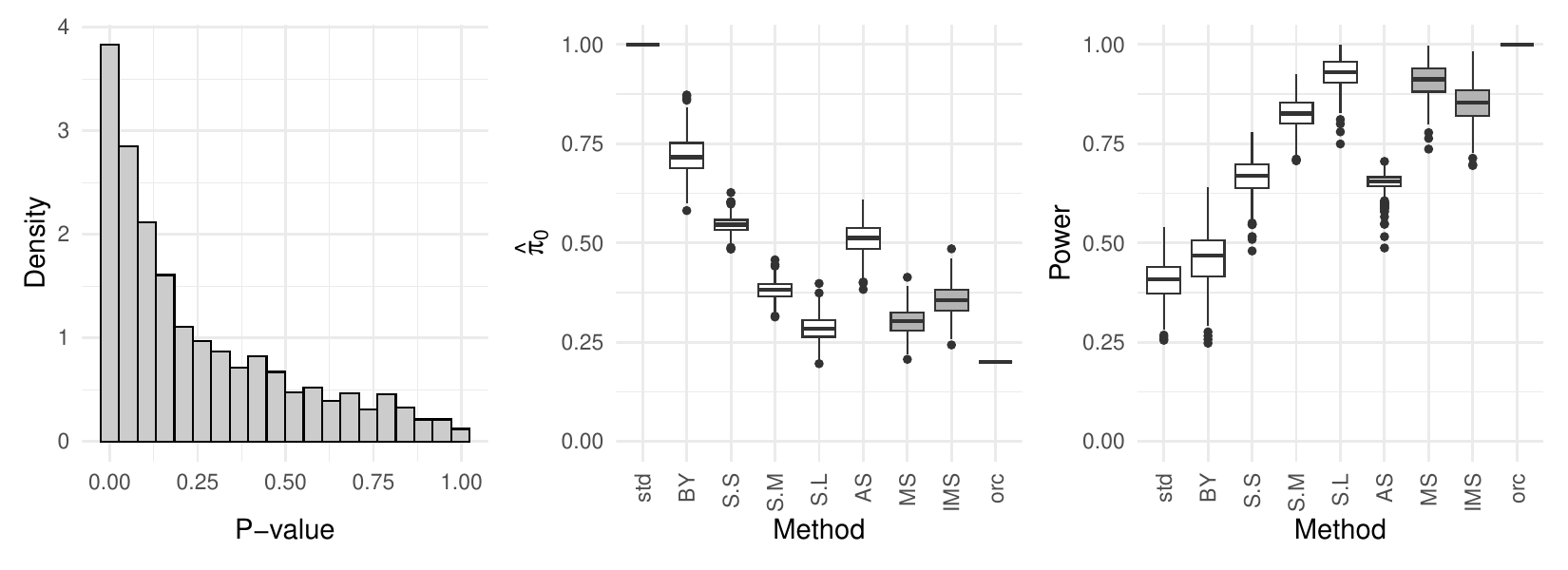}
        \subcaption*{ (c.3) Many strong signals}
        \end{minipage}    

            \caption{Comparison of null proportion estimators with conformal $p$-values, validation sample size $n = 1000$. For more details, see the caption of Figure~\ref{fig:simulation.uniform}. }
        \label{fig:simulation.conformal.1000}
\end{figure}

\subsubsection{Conservative $\underline{\pi}_0$}\label{appe:sec:simulation.conservative}
Figure~\ref{appe:fig:simulation.conservative.underline.pi0} follows setting (a.3) but MS and IMS use a conservative $\underline{\pi}_0 = 0.5 > \pi_0 = 0.2$. 
By Definitions~\ref{def:MS} and~\ref{def:IMS}, the conservativeness of $\underline{\pi}_0 = 0.5$ yields that the MS and IMS null proportion estimators equal $\underline{\pi}_0 = 0.5$ with high probability. As a result, their power is lower than those of the oracle procedure and Storey's estimator with a moderately large $\lambda = 0.5$ or $\lambda = 0.8$. 
AS is less powerful due to the capping.

\begin{figure}[tbp]
        \centering
        \begin{minipage}{0.8\textwidth}
                \centering
                \includegraphics[clip, trim = 0cm 1cm 0cm 0cm, width = 1\textwidth]{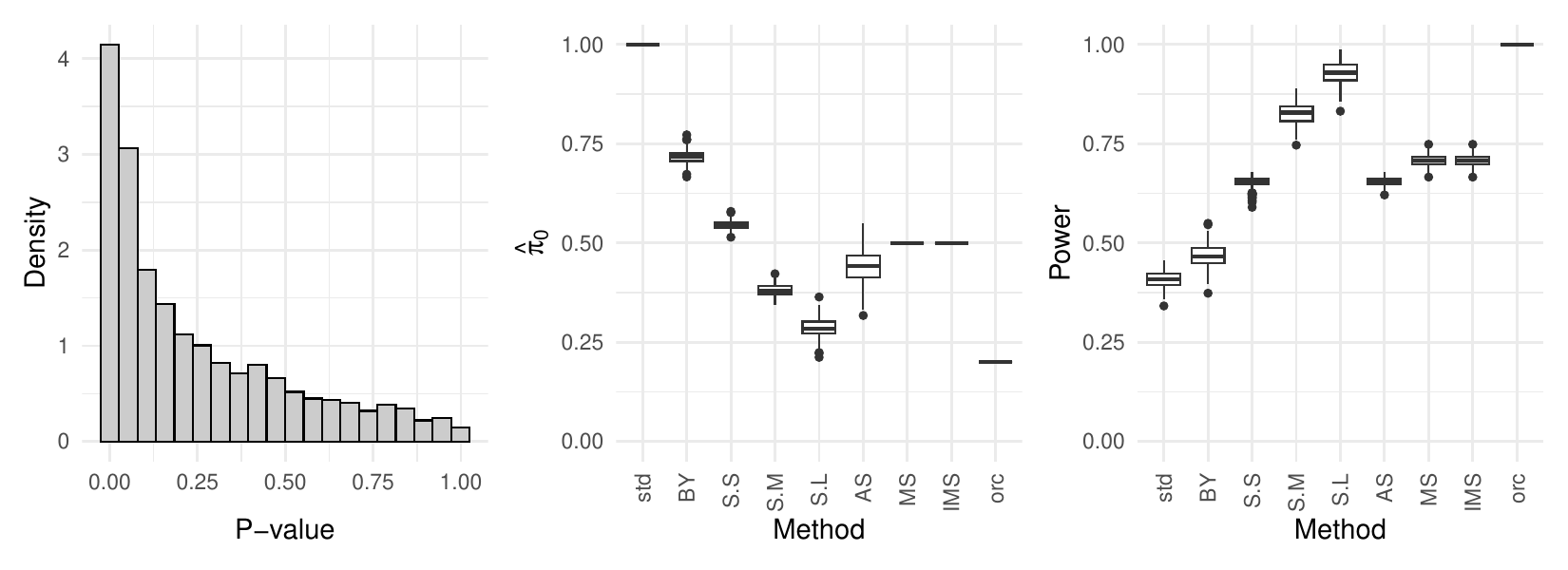}
        \end{minipage}   

            \caption{Comparison of null proportion estimators with \emph{independent} $p$-values under the many signal setting (a.3) in Section~\ref{sec:simulation}. For MS and IMS, we use a conservative $\underline{\pi}_0 = 0.5 > \pi_0 =0.2$.}
        \label{appe:fig:simulation.conservative.underline.pi0}
\end{figure}

\subsubsection{IMS with prefixed $\kappa$}\label{appe:sec:simulation.IMS.kappa}

We compare the IMS with pre-fixed $\kappa \in \{0.05, 0.1, \ldots, 0.5\}$ and data-driven $\hat \kappa$.
For MS and IMS, we use the same hyperparameter settings as in Section~\ref{sec:simulation} without further specification.

\begin{enumerate}
        \item [(e.1)] $f_1$ with a small minimizer. 
        We consider $n = 1000$, $\pi_0 = 0.5$. 
        Null $p$-values follow $U(0,1)$. 
        Alternative $p$-values follow the mixture $0.3\Phi(\mathcal{N}(-2,0.25)) + 0.7\Phi(\mathcal{N}(0.75,0.25))$, and $f_1$ takes the minimal value around $0.25$. For MS and IMS, we use $\underline{\pi}_0 = 0.5$.

    \item [(e.2)] Many signals. 
    The same as (a.3) in Section~\ref{sec:simulation}. For MS and IMS, we use $\underline{\pi}_0 = 0.2$.
\end{enumerate}

In panel (e.1) of Figure~\ref{appe:fig:IMS.kappa}, the minimizer $\lambda^*$ \eqref{pi0starunif} of $f_1$ is approximately $0.25$. Hence, when $\kappa > 0.25$, the IMS estimator of the null proportion becomes conservative, since the population minimizer is no longer contained in $[\kappa,1]$ (see the second plot of panel (e.1)).
In this setting, the ideal $\kappa$ should be smaller than $0.25$.
In panel (e.2) of Figure~\ref{appe:fig:IMS.kappa}, the proportion of alternatives is large, so the rejection threshold is also large (the rejection threshold of the oracle BH is one). The second plot of panel (e.2) shows that the resulting $\hat{\pi}_0$ values are nearly constant across $\kappa$; the third plot shows that the rejection threshold decreases as $\kappa$ becomes smaller, due to more severe capping and hence yielding lower power.
In this setting, a larger pre-fixed $\kappa$ performs better.
Across both settings, $\hat{\kappa}$ achieves the highest power by adaptively selecting a favorable $\kappa$, which is well expected from \eqref{equkappahatthebest}.
\begin{figure}[tbp]
        \centering
        \begin{minipage}{1\textwidth}
                \centering
                \includegraphics[clip, trim = 0cm 0cm 0cm 0cm, width = 1\textwidth]{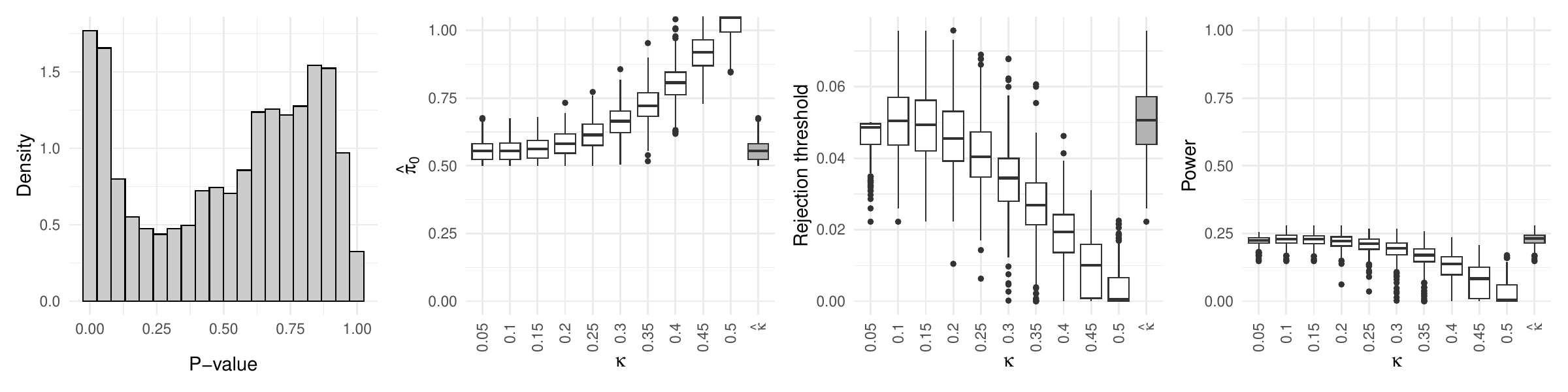}
        \subcaption*{(e.1) $f_1$ with a small minimizer}
        \end{minipage}   

         \begin{minipage}{1\textwidth}
                \centering
                \includegraphics[clip, trim = 0cm 0cm 0cm 0cm, width = 1\textwidth]{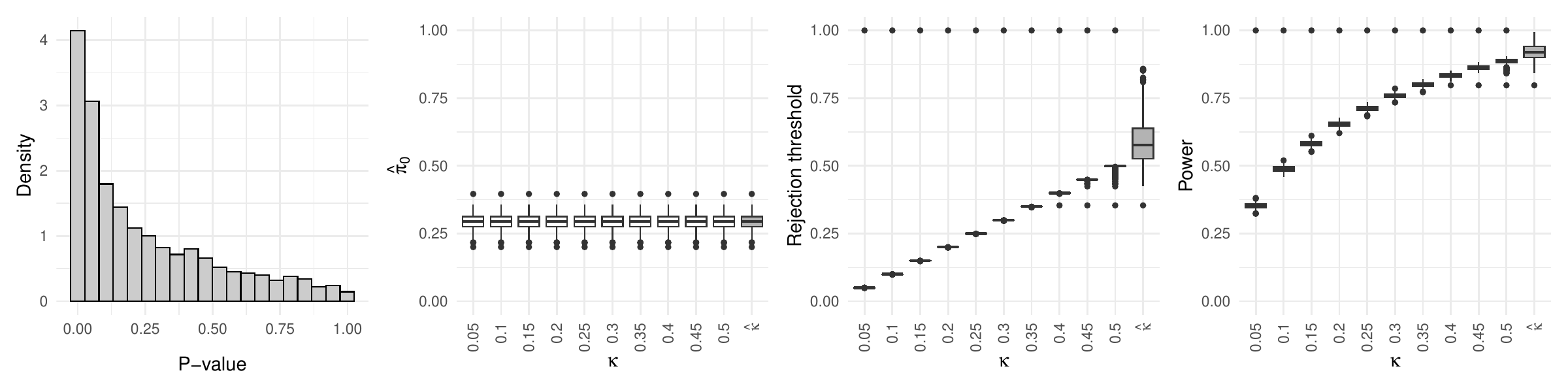}
        \subcaption*{(e.2) Many signals}
        \end{minipage}   

            \caption{Comparison of IMS with prefixed $\kappa \in \{0.05, 0.1, \ldots, 0.5\}$ and the data-driven $\hat \kappa$. }
        \label{appe:fig:IMS.kappa}
\end{figure}






\end{document}